\documentclass[a4paper,UKenglish,thm-restate]{lipics-v2019}
\pdfoutput=1

\title{What You Must Remember When Transforming Datawords} 

\titlerunning{What You Must Remember When Transforming Datawords} 

\author{M.~Praveen}{Chennai Mathematical Institute, India \and UMI
ReLaX, Indo-French joint research unit}{}{}{Partially supported by a grant from the Infosys foundation.}

\authorrunning{M.~Praveen} 

\Copyright{M.~Praveen} 


\category{} 

\relatedversion{} 

\supplement{}

\acknowledgements{The author thanks C.~Aiswarya, Kamal Lodaya, K.~Narayan Kumar
and anonymous reviewers for suggestions to improve the
presentation and pointers to related works.}

\nolinenumbers 

\hideLIPIcs  

\EventEditors{Nitin Saxena and Sunil Simon}
\EventNoEds{2}
\EventLongTitle{40th IARCS Annual Conference on Foundations of Software Technology and Theoretical Computer Science (FSTTCS 2020)}
\EventShortTitle{FSTTCS 2020}
\EventAcronym{FSTTCS}
\EventYear{2020}
\EventDate{December 14--18, 2020}
\EventLocation{BITS Pilani, K K Birla Goa Campus, Goa, India (Virtual Conference)}
\EventLogo{}
\SeriesVolume{182}
\ArticleNo{35}

\usepackage{amsfonts,amstext,amsmath}
\usepackage{mathrsfs,stmaryrd}
\usepackage{todonotes}
\usepackage{tikz}

\begin{document}

\maketitle

\newtheorem{construction}[theorem]{Construction}

\tikzstyle{state}=[draw=black, fill=black, circle, inner sep=0.01cm, minimum
height=0.1cm]
\usetikzlibrary{shapes.multipart}

\begin{abstract}
    Streaming Data String Transducers (SDSTs) were introduced to model
a class of imperative and a class of functional programs, manipulating
lists of data items. These can be used to write commonly used routines
such as insert, delete and reverse. SDSTs can handle data values from
a potentially infinite data domain. The model of Streaming String
Transducers (SSTs) is the fragment of SDSTs where
the infinite data domain is dropped and only finite alphabets are
considered. SSTs have been much studied from a language theoretical
point of view. We introduce data back into SSTs, just like data was
introduced to finite state automata to get register automata. The
result is Streaming String Register Transducers (SSRTs), which is a
subclass of SDSTs.

We use origin semantics for SSRTs and give a machine independent
characterization, along the lines of Myhill-Nerode theorem. Machine independent
characterizations for similar models are the basis of learning
algorithms and enable us to understand fragments of the models. Origin
semantics of transducers track which positions of the output originate from
which positions of the input.  Although a restriction, using origin semantics
is well justified and is known to simplify many problems related to transducers.
We use origin semantics as a technical building block, in addition to
characterizations of deterministic register automata. However, we need to build
more on top of these to overcome some challenges unique to SSRTs.
\end{abstract}

\newcommand{\partfn}{\rightharpoonup}
\newcommand{\set}[1]{\{#1\}}
\newcommand{\curr}{\mathit{curr}}
\newcommand{\Nat}{\mathbb{N}}
\newcommand{\Int}{\mathbb{I}}
\newcommand{\Whole}{\Nat^{+}}
\newcommand{\semantics}[1]{\llbracket #1 \rrbracket}
\newcommand{\lef}{\mathtt{left}}
\newcommand{\mi}{\mathtt{middle}}
\newcommand{\ri}{\mathtt{right}}
\newcommand{\data}{\mathtt{data}}
\newcommand{\lett}{\mathtt{letter}}
\newcommand{\ifl}{\mathtt{ifl}}
\newcommand{\aifl}{\mathtt{aifl}}
\newcommand{\ifls}{\mathtt{m}}
\newcommand{\iflp}{\mathtt{v}}
\newcommand{\iflpp}{\mathtt{iflpp}}
\newcommand{\iso}{\simeq}
\newcommand{\abs}{\mathtt{abs}}
\newcommand{\concr}{\mathtt{concr}}
\newcommand{\iflsp}{\mathtt{vm}}
\newcommand{\fequiv}{\equiv_{f}}
\newcommand{\requiv}{\equiv_{f}^{E}}
\newcommand{\uequiv}{\equiv_{U,f}^{E}}
\newcommand{\sequiv}{\equiv_{S}}
\newcommand{\tsc}[1]{\textsuperscript{#1}}
\newcommand{\bv}[1]{\langle #1 \rangle}
\newcommand{\drlst}{\nleftarrow}
\newcommand{\pbl}{\mathcal{P}}
\newcommand{\pref}{\mathit{pref}}
\newcommand{\suf}{\mathit{suf}}
\newcommand{\bl}{\mathit{bl}}
\newcommand{\ur}{\mathit{ur}}
\newcommand{\val}{\mathit{val}}
\newcommand{\ptr}{\mathit{ptr}}
\newcommand{\reu}{\mathit{reuse}}
\newcommand{\pcomp}{\odot}
\newcommand{\dropOrig}{\upharpoonright_{2}}
\newcommand{\update}{\mathit{ud}}


\section{Introduction}
\label{sec:intro}
Transductions are in general relations among words. Transducers are
theoretical models that implement transductions. Transducers are used
in a variety of applications, such as analysis of web sanitization
frameworks, host based intrusion detection, natural language
processing, modeling some classes of programming languages and
constructing programming language tools like evaluators, type
checkers and translators.
Streaming Data String Transducers (SDSTs) were introduced in
\cite{AC2011} to model a class of imperative and a class of functional
programs, manipulating lists of data items. Transducers have been used
in \cite{HISBJ2012} to infer semantic interfaces of data structures
such as stacks. Such applications use Angluin style learning, which involves
constructing transducers by looking at example operations of the
object under study. Since the transducer is still under construction,
we need to make inferences about the transduction without having
access to a transducer which implements it. Theoretical
bases for doing this are machine independent characterizations, which
identify what kind of transductions can be implemented by what kind of
transducers and give a template for constructing transducers. Indeed
the seminal Myhill-Nerode theorem gives a machine independent
characterization for regular languages over finite alphabets, which
form the basis of Angluin style learning of regular languages
\cite{Angluin87}. A similar characterization for a fragment of SDSTs
is given in \cite{Bojanczyk2014} and is used as a basis to design a
learning algorithm.

Programs deal with data from an infinite domain and transducers
modeling the programs should also treat data as such. For example in
\cite{HISBJ2012}, the state space reduced from $10^9$ to $800$ and the
number of learning queries reduced from billions to $4000$ by
switching to a transducer model that can deal with data from an
infinite domain. We give a machine independent characterization for a
fragment of SDSTs more powerful than those in \cite{HISBJ2012,
Bojanczyk2014}. The additional power comes from
significant conceptual differences. The transducers used in
\cite{HISBJ2012} produce the output in a linear fashion without
remembering what was output before. For example, they cannot output
the reverse of the input strings, which can be done by our model. The
model studied in \cite{Bojanczyk2014} are called Streaming String
Transducers (SSTs), the fragment obtained from SDSTs by dropping the
ability to deal with data values from an infinite domain. We retain
this ability in our model, called Streaming String Register
Transducers (SSRTs). It is obtained from SDSTs by dropping the ability
to deal with linear orders in the data domain.  Apart from Angluin
style learning algorithms, machine independent characterizations are
also useful for studying fragments of transducer models. E.g.~in
\cite{Bojanczyk2014}, machine independent characterization of SSTs is
used to study fragments such as non-deterministic automata with
output and transductions definable in First Order logic.

We use origin semantics of transducers, which are used in
\cite{Bojanczyk2014} to take into account how positions of the output
originate from the positions of the input. Using origin semantics is
known to ease some of the problems related to transducers, e.g.,
\cite{BMPP2018}.  Origin semantics is a restriction, but a reasonable
one and is used extensively in this paper.

\paragraph*{Contributions} Machine independent characterizations are
known for automata over data values from an infinite domain
\cite{FK2003, BCG2010} and for streaming transducers over finite
alphabets \cite{Bojanczyk2014}, but not for streaming transducers over
data values, which is what we develop here. This involves both
conceptual and technical challenges.  In \cite{FK2003,
BCG2010}, data values that must be remembered by an automaton while
reading a word from left to right are identified using a machine independent definition.
We lift this to transducers and identify that the concept of
factored outputs from \cite{Bojanczyk2014} is necessary for this. Factored
outputs can let us ignore some parts of transduction outputs, which is
necessary to define when two words behave similarly. However,
\cite{Bojanczyk2014} does not deal with data values from an infinite
domain and it takes quite a bit of manipulation with permutations on
data values to make ideas from there work here. In transductions,
suffixes can influence how prefixes are transformed. This is
elegantly handled in \cite{Bojanczyk2014} using two way transducer
models known to be equivalent to SSTs. There are no such models known
when data values are present. To handle it in a one way transducer
model, we introduce data structures based on trees that keep track of
all possible suffixes. This does raise the question of whether there
are interesting two way transducer models with data values. Recent
work \cite{BS2020} has made progress in this direction, which we
discuss at the end of this article. We concentrate here on
SDSTs and its fragments, which are known to be equivalent to classes
of imperative and functional programming languages. In \cite{AC2011},
it is explained in detail which features of programming languages
correspond to which features of the transducer. Over finite alphabets,
streaming string transducers are expressively equivalent to regular
transductions, which are also defined by two way deterministic
finite-state transducers and by monadic second order logic
\cite{AC2010}.

\paragraph*{Related Works} Studying transducer models capable of
handling data values from an infinite domain is an active area of
research \cite{EFR2019,EFR2020}. Streaming transducers like SDSTs have
the distinctive feature of using variables to store intermediate
values while computing transductions; this idea appears in an earlier
work \cite{CJ1977} that introduced \emph{simple programs on strings},
which implement the same class of transductions as those implemented
by SSTs. An Angluin style learning algorithm for deterministic automata
with memory is given in \cite{HSJC2012}. A machine independent
characterization of automata with finite memory is given in
\cite{CHJMS2015}, which is further extended to data domains with
arbitrary binary relations in \cite{CJHS2012}. The learning algorithm
of \cite{HSJC2012} is extended to Mealy machines with data in
\cite{HISBJ2012}. However, Mealy machines are not as powerful as SSRTs
that we consider here. Using a more abstract approach of nominal
automata, \cite{MSSKS2017} presents a learning algorithm for automata
over infinite alphabets. Logical characterizations of transducers that
can handle data are considered in \cite{DH2016}.  However, the
transducers in that paper cannot use data values to make decisions,
although they are part of the output.  Register automata with linear
arithmetic introduced in \cite{CLTW2017} shares some of the features
of the transducer model used here. Here, data words stored in
variables can be concatenated, while in register automata with linear
arithmetic, numbers stored in variables can be operated upon by linear
operators.

Proofs of some of the results in this paper are tedious and are moved to
the appendix to maintain flow of ideas in the main paper. Proofs of results stated in the main
part of the paper are in Sections \ref{app:Equivalences}, \ref{app:ConstSSRT}
and \ref{app:TransdSSRT}.  Section~\ref{app:Fundamental} states and proves some
basic properties of transductions and transducers that are only invoked in
Sections \ref{app:Equivalences}, \ref{app:ConstSSRT} and \ref{app:TransdSSRT}.
Section~\ref{app:Lengthy} contains proofs that are especially long.  They
consist of lengthy case analyses to rigorously verify facts that are
intuitively clear.

\section{Preliminaries}
\label{sec:prelim}
Let $\Int$ be the set of integers, $\Nat$ be the set of non-negative
integers and $D$ be an infinite set of data values. We will refer to
$D$ as the \emph{data domain}.  For $i,j \in \Int$, we denote by
$[i,j]$ the set $\set{k \mid i \le k \le j}$. For any set $S$, $S^{*}$
denotes the set of all finite sequences of elements from $S$. The empty
sequence is denoted by $\epsilon$. Given $u,v \in S^{*}$, $v$ is a
\emph{prefix} (resp.~suffix) of $u$ if there exists $w \in
S^{*}$ such that $u=vw$ (resp.~$u=wv$). The sequence $v$ is an
\emph{infix} of $u$ if there are sequences $w_{1},w_{2}$ such that
$u=w_{1} v w_{2}$.

 Let $\Sigma, \Gamma$ be finite alphabets. We will use $\Sigma$ for input
alphabet and $\Gamma$ for output alphabet. A \emph{data word} over $\Sigma$ is
a word in $(\Sigma \times D)^{*}$. A \emph{data word with origin information}
over $\Gamma$ is a word in $(\Gamma \times D \times \Nat)^*$. Suppose $\Sigma =
\set{\mathrm{title}, \mathrm{firstName}, \mathrm{lastName}}$ and
$\Gamma=\set{\mathrm{givenName}, \mathrm{surName}}$. An example data word over
$\Sigma$ is $(\mathrm{title}, \mathrm{Mr.}) (\mathrm{firstName},
\mathrm{Harry}) (\mathrm{lastName},\mathrm{Tom})$. If we were to give this as
input to a device that reverses the order of names, the output would be the
data word with origin information $(\mathrm{surName}, \mathrm{Tom},3)
(\mathrm{givenName},$ $\mathrm{Harry},2)$, over $\Gamma$. In the triple
$(\mathrm{givenName},\mathrm{Harry},2)$, the third component $2$ indicates that
the pair $(\mathrm{givenName},$ $\mathrm{Harry})$ originates from the second
position of the input data word. We call the third component \emph{origin} and
it indicates the position in the input that is responsible for
producing the output triple. If a transduction is being implemented by
a transducer, the origin of an output position is the position of the
input that the transducer was reading when it produced the output.
The data value at some position of the output may come from any
position (not necessarily the origin) of the input data word. We write
\emph{transduction} for any function from data words over $\Sigma$ to
data words with origin information over $\Gamma$.

For a data word $w$, $|w|$ is its length. For a position $i \in
[1,|w|]$, we denote by $\data(w,i)$ (resp.~$\lett(w,i)$) the data
value (resp.~the letter from the finite alphabet) at the
$i$\textsuperscript{th} position of $w$. We denote by $\data(w,*)$ the
set of all data values that appear in $w$. For positions $i \le j$, we
denote by $w[i,j]$ the infix of $w$ starting at position $i$ and
ending at position $j$. Note that $w[1,|w|]=w$. Two data words $w_{1},
w_{2}$ are \emph{isomorphic} (denoted by $w_{1} \iso w_{2}$) if
$|w_{1}| = |w_{2}|$, $\lett(w_{1},i)=\lett(w_{2},i)$
and $\data(w_{1},i)=\data(w_{1},j)$ iff
$\data(w_{2},i)=\data(w_{2},j)$ for all positions $i,j \in
[1,|w_{1}|]$.  For data values $d,d'$, we denote by $w[d/d']$ the data
word obtained from $w$ by replacing all occurrences of $d$ by $d'$. We
say that $d'$ is a \emph{safe replacement} for $d$ in $w$ if $w[d/d']
\iso w$. Intuitively, replacing $d$ by $d'$ doesn't introduce new
equalities/inequalities among the positions of $w$. For example,
$d_{1}$ is a safe replacement for $d_{2}$ in $(a,d_{3}) (b,d_{2})$,
but not in $(a,d_{1}) (b,d_{2})$.

A permutation on data values is any bijection $\pi: D \to D$.  For a
data word $u$, $\pi(u)$ is obtained from $u$ by replacing all its data
values by their respective images under $\pi$. A transduction $f$ is
\emph{invariant under permutations} if for every data word $u$ and
every permutation $\pi$, $f(\pi(u)) = \pi(f(u))$ (permutation can be
applied before or after the transduction).

Suppose a transduction $f$ has the property that for any triple
$(\gamma,d,o)$ in any output $f(w)$, there is a position $i \le o$ in
$w$ such that $\data(w,i)=d$. If the data value $d$ is
output from the origin $o$, then $d$ should have already occurred in
the input on or before $o$. Such transductions are said to be
\emph{without data peeking}. We say that a transduction has
\emph{linear blow up} if there is a constant $K$ such that for any
position $o$ of any input, there are at most $K$ positions in
the output whose origin is $o$.

\paragraph*{Streaming String Register Transducers}
We present an extension of SSTs to handle data values, just like
finite state automata were extended to finite memory automata
\cite{KF1994}. Our model is a subclass of SDSTs, which can
store intermediate values (which can be long words) in
variables. E.g., reversing an input word can be achieved as follows:
as each input symbol is read, concatenate it to the back of a variable
maintained for this purpose. At the end, the variable will have the
reverse of the input. There are also registers in these models, which
can store single data values. Transitions can be enabled/disabled
based on whether the currently read data value is equal/unequal to the
one stored in one of the registers.
\begin{definition}
    \label{def:fmsst}
    A \emph{Streaming String Register Transducer} (SSRT) is an eight tuple
    $S=(\Sigma, \Gamma, Q, q_{0}, R, X,O, \Delta)$, where
    \begin{itemize}
        \item the finite alphabets $\Sigma, \Gamma$ are used for
            input, output respectively,
        \item $Q$ is a finite set of states, $q_{0}$ is the initial
            state,
        \item $R$ is a finite set of registers and $X$ is a finite set
            of data word variables,
        \item $O:Q \partfn ( (\Gamma \times \hat{R}) \cup X)^{*}$ is a
            partial output function, where $\hat{R} = R \cup
            \set{\curr}$, with $\curr$ being a special symbol used to
            denote the current data value being read and
        \item $\Delta \subseteq (Q \times \Sigma \times \Phi \times Q
            \times 2^{R} \times U)$ is a finite set of transitions. The set
            $\Phi$ consists of all Boolean combinations of atomic constraints of
            the form $r^{=}$ or $r^{\ne}$ for $r \in R$. The set $U$ is
            the set of all functions from the set $X$ of data word variables to
            $( (\Gamma \times \hat{R}) \cup X)^{*}$.
    \end{itemize}
    It is required that
    \begin{itemize}
        \item For every $q \in Q$ and $x \in X$, there is at most one
            occurrence of $x$ in $O(q)$ and
        \item for every transition $(q,\sigma,\phi,q',R',\update)$ and for every
            $x \in X$, $x$ appears at most once in the set
            $\set{\update(y) \mid y \in X}$.
    \end{itemize}
\end{definition}
We say that the last two conditions above enforce a SSRT to be
\emph{copyless}, since it prevents multiple copies of contents being
made.

A \emph{valuation} $\val$ for a transducer $S$ is a partial function
over registers and data word variables such that for every register
$r \in R$, either $\val(r)$ is undefined or is a data value in $D$, and
for every data word variable $x \in X$, $\val(x)$ is a data word with
origin information over $\Gamma$. The valuation
$\val$ and data value $d$ satisfies the atomic constraint $r^{=}$
(resp.~$r^{\ne}$) if $\val(r)$ is defined and $d=\val(r)$
(resp.~undefined or $d \ne \val(r)$). Satisfaction is
extended to Boolean combinations in the standard way. We say that a
SSRT is \emph{deterministic} if for every two transitions
$(q,\sigma,\phi,q',R',u)$ and $(q,\sigma,\phi',q'',R'',u')$ with the
same source state $q$ and input symbol $\sigma$, the formulas $\phi$
and $\phi'$ are mutually exclusive (i.e., $\phi \land \phi'$ is
unsatisfiable). We consider only deterministic SSRTs here.

A configuration is a triple $(q,\val,i)$ where $q\in Q$
is a state, $\val$ is a valuation and $i$ is the number of symbols
read so far. The transducer starts in the configuration
$(q_{0},\val_{\epsilon},0)$ where $q_{0}$ is the initial state and
$\val_{\epsilon}$ is the valuation such that $\val_{\epsilon}(r)$ is
undefined for every register $r \in R$ and $\val_{\epsilon}(x) =
\epsilon$ for every data word variable $x \in X$.  From a
configuration $(q,\val,i)$, the transducer can read a pair $(\sigma,d)
\in \Sigma \times D$ and go to the configuration $(q',\val',i+1)$ if
there is a transition $(q,\sigma,\phi,q',R',\update)$ and 1)
$d$ and $\val$ satisfies $\phi$ and 2)
$\val'$ is obtained from $\val$ by assigning $d$ to all the
registers in $R'$ and for every $x \in X$, setting $\val'(x)$ to
$\update(x)[y\mapsto \val(y),(\gamma,\curr)\mapsto (\gamma,d,i+1),
(\gamma,r)\mapsto(\gamma,\val(r),i+1)]$ (in $\update(x)$,
replace every occurrence of $y$ by $\val(y)$ for every data word
variable $y \in X$, replace every occurrence of $(\gamma,\curr)$ by
$(\gamma,d,i+1)$ for every output letter $\gamma \in \Gamma$ and
replace every occurrence of $(\gamma,r)$ by $(\gamma,\val(r),i+1)$
for every output letter $\gamma \in \Gamma$ and every register $r \in
R$). After reading a data word $w$, if the transducer reaches some
configuration $(q,\val,n)$ and $O(q)$ is not defined, then the
transducer's output $\semantics{S}(w)$ is undefined for the input $w$.
Otherwise, the transducer's output is defined as
$\semantics{S}(w)=O(q)[y\mapsto \val(y),(\gamma,\curr)\mapsto
(\gamma,d,n), (\gamma,r)\mapsto(\gamma,\val(r),n)]$, where $d$ is the
last data value in $w$.

Intuitively, the transition $(q,\sigma,\phi,q',R',\update)$ checks that the
current valuation $\val$ and the data value $d$ being read satisfies
$\phi$, goes to the state $q'$, stores $d$ into the registers in $R'$
and updates data word variables according to the update function
$\update$. The condition that $x$ appears at most once in the set
$\set{\update(y) \mid y \in X}$ ensures that the contents of any data
word variable are not duplicated into more than one variable. This
ensures, among other things, that the length of the output is linear
in the length of the input. The condition that for every two transitions
$(q,\sigma,\phi,q',R',\update)$ and
$(q,\sigma,\phi',q'',R'',\update')$ with the same source state and
input symbol, the formulas $\phi$ and $\phi'$ are mutually exclusive
ensures that the transducer cannot reach multiple configurations after
reading a data word (i.e., the transducer is deterministic).

\begin{example}
    \label{ex:identityOrReverse}
    Consider the transduction that is the identity on inputs in which the
    first and last data values are equal. On the remaining inputs, the
    output is the reverse of the input. This can be implemented by a
    SSRT using two data word variables. As each input symbol is read,
    it is appended to the front of the first variable and to the back
    of the second variable. The first variable stores the input and
    the second one stores the reverse. At the end, either the first or
    the second variable is output, depending on whether the last data
    value is equal or unequal to the first data value (which is stored
    in a register).
\end{example}

In Section~\ref{sec:equivalences-abs}, we define an equivalence relation on
data words and state our main result in terms of the finiteness of the index of
the equivalence relation and a few other properties. In
Section~\ref{sec:constSSRT-abs}, we prove that transductions satisfying certain
properties can be implemented by SSRTs (the backward direction of the main
result) and we prove the converse in Section~\ref{sec:propTransdSSRT}.


\section{How Prefixes and Suffixes Influence Each Other}
\label{sec:equivalences-abs}
As is usual in many machine independent characterizations (like the
classic Myhill-Nerode theorem for regular languages), we define
an equivalence relation on the set of data words to identify similar
ones. If the equivalence relation has finite index, it can be used to
construct finite state models. We start by looking at what ``similar data
words'' mean in the context of transductions.

Suppose $L$ is the set of all even length words over some finite alphabet. The
words $a$ and $aaa$ do the same thing to any suffix $v$: $a \cdot v \in L$ iff
$aaa \cdot v \in L$. So, $a$ and $aaa$ are identified to be similar with
respect to $L$ in the classic machine independent characterization. Instead of
a language $L$, suppose we have a transduction $f$ and we are trying to
identify words $u_1,u_2$ that do the same thing to any suffix $v$.  The naive
approach would be to check if $f(u_1\cdot v)=f(u_2 \cdot v)$, but this does not
work. Suppose a transduction $f$ is such that $f(a \cdot b)=(a,1) (b,2)$,
$f(aaa \cdot b) = (a,1)(a,2)(a,3) \cdot (b,4)$ and $f(c \cdot
b)=(c,1)(b,2)(b,2)$ (we have ignored data values in this transduction). The
words $a$ and $aaa$ do the same thing to the suffix $b$ (the suffix is copied
as it is to the output), as opposed to $c$ (which copies the suffix twice to
the output). But $f(a \cdot b) \ne f(aaa \cdot b)$. The problem is that we are
not only comparing what $a$ and $aaa$ do to the suffix $b$, but also comparing
what they do to themselves. We want to indicate in some way that we want to
ignore the parts of the output that come from $a$ or $aaa$:
$f(\underline{a}\mid v)=\lef \cdot (b,2)$ and $f(\underline{aaa} \mid b)=\lef
\cdot (b,4)$.  We have underlined $a$ and $aaa$ on the input side to indicate
that we want to ignore them; we have replaced $a$ and $aaa$ in the output by
$\lef$ to indicate that they are coming from ignored parts of the input. This
has been formalized as factored outputs in \cite{Bojanczyk2014}. This is still
not enough for our purpose, since the outputs $(b,2)$ and $(b,4)$ indicate that
$a$ and $aaa$ have different lengths. This can be resolved by offsetting one of
the outputs by the difference in the lengths: $f(\underline{a}\mid v)=\lef
\cdot (b,2) = f_{-2}(\underline{aaa} \mid b)$. The subscript $-2$ in
$f_{-2}(\underline{aaa} \mid b)$ indicates that we want to offset the origins
by $-2$. We have formalized this in the definition below, in which we have
borrowed the basic definition from \cite{Bojanczyk2014} and added data values
and offsets.
\begin{definition}[Offset factored outputs]
    \label{def:factoredOutputs}
    Suppose $f$ is a transduction and $uvw$ is a data word over
    $\Sigma$. For a triple $(\gamma,d,o)$ in $f(uvw)$, the
    \emph{abstract origin} $\abs(o)$ of $o$ is $\lef$ (resp.~$\mi$,
    $\ri$) if $o$ is in $u$ (resp.~$v$, $w$). The factored output
$f(\underline{u} \mid v \mid w)$ is obtained from $f(uvw)$ by first
replacing every triple $(\gamma,d,o)$ by $(*,*,\abs(o))$ if
$\abs(o)=\lef$ (the other
triples are retained without change). Then all consecutive occurrences
of $(*,*,\lef)$ are replaced by a single triple $(*,*,\lef)$ to get
$f(\underline{u} \mid v \mid w)$. Similarly we get $f(u \mid
\underline{v} \mid w)$ and $f(u \mid v \mid \underline{w})$ by using
$(*,*,\mi)$ and $(*,*,\ri)$ respectively. We get $f(\underline{u} \mid
v)$ and $f(u \mid \underline{v})$ similarly, except that there is no
middle part. For an integer $z$, we obtain $f_z(\underline{u} \mid v)$
by replacing every triple $(\gamma,d,o)$ by $(\gamma,d,o+z)$ (triples
$(*,*,\lef)$ are retained without change).
\end{definition}
Let $w=(a,d_{1}) (a,d_2)(b,d_{3}) (c,d_{4})$ and $f$ be the transduction
in Example~\ref{ex:identityOrReverse}. Then $f(w) =
(c,d_{4},4)(b,d_{3},3)(a,d_2,2)(a,d_{1},1)$ (assuming $d_{4} \ne
d_{1}$). The factored output $f( \underline{(a,d_1)(a,d_2)} \mid
(b,d_3) \mid (c,d_4))$ is $(c,d_{4},4)(b,d_{3},3)(*,*,\lef)$.

It is tempting to say that two data words $u_1$ and $u_2$ are
equivalent if for all $v$, $f(\underline{u_1}\mid v) =
f_z(\underline{u_2} \mid v)$, where $z=|u_1| - |u_2|$. But this does
not work; continuing with the transduction $f$ from
Example~\ref{ex:identityOrReverse}, no two data words from the
infinite set
$\set{(a,d_i)\mid i \ge 1}$ would be equivalent:
$f(\underline{(a,d_i)} \mid (a,d_i)) \ne  f(\underline{(a,d_j)} \mid
(a,d_i))$ for $i \ne j$. To get an equivalence relation with finite
index, we need to realize that the important thing is not the first data value, but its (dis)equality with the last one.
So we can say that for every $i$, there is a permutation $\pi_i$ on
data values
mapping $d_i$ to $d_1$ such that $f(\underline{\pi_i(a,d_i)} \mid v) =
f(\underline{(a,d_1)}\mid v)$. This will get us an equivalence
relation with finite index but it is not enough, since the transducer
model we build must satisfy another property: it must use only
finitely many registers to remember data values. Next we examine which
data values must be remembered.

Suppose $L$ is the set of all data words in which the first and last
data values are equal. Suppose a device is reading the word
$d_1 d_2 d_3 d_1$ from left to right and trying to determine whether
the word belongs to $L$ (we are ignoring letters from the finite
alphabet here). The device must remember $d_1$ when it is read first,
so that it can be compared to the last data value. A machine
independent characterization of what must be remembered is given in
\cite[Definition 2]{BCG2010}; it says that the first occurrence of
$d_1$ in $d_1 d_2 d_3 d_1$ is \emph{$L$-memorable} because replacing
it with some fresh data value $d_4$ (which doesn't occur in the word)
makes a difference: $d_1 d_2 d_3 d_1 \in L$ but $d_4 d_2 d_3 d_1
\notin L$. We adapt this concept to transductions, by suitably
modifying the definition of ``making a difference''.
\begin{definition}[memorable values]
    \label{def:fMemorable}
    Suppose $f$ is a transduction. A data value $d$ is
    \emph{$f$-memorable} in a data word $u$ if there exists a
    data word $v$ and a safe replacement $d'$ for $d$ in $u$ such that
    $f(\underline{u[d/d']} \mid v)\ne f(\underline{u} \mid v)$.
\end{definition}
Let $f$ be the transduction of Example~\ref{ex:identityOrReverse} and
$d_{1},d_{2},d_3,d_{1}'$ be distinct data values. We have
$f(\underline{d_{1}d_{2}d_3}\mid d_1)=(*,*,\lef)(d_{1},4)$ and
$f(\underline{d_{1}'d_{2}d_3}\mid d_1)
=(d_{1},4)(*,*,\lef)$. Hence, $d_{1}$ is $f$-memorable in $d_1d_2d_3$.

We have to consider one more phenomenon in transductions. Consider the transduction $f$ whose output is $\epsilon$
for inputs of length less than five. For other inputs, the output is
the third (resp.~fourth) data value if the first and fifth are equal
(resp.~unequal). Let $d_1,d_2,d_3,d_4,d_5,d_1'$ be distinct data
values. We have $f(\underline{d_1d_2d_3d_4} \mid v)=\epsilon =
f(\underline{d_1'd_2d_3d_4} \mid v)$ if $v=\epsilon$ and $f(\underline{d_1d_2d_3d_4}
\mid v)=(*,*,\lef) = f(\underline{d_1'd_2d_3d_4} \mid v)$ otherwise. Hence,
$d_{1}$ is not $f$-memorable in $d_1d_2d_3d_4$. However, any device
implementing $f$ must remember $d_{1}$ after reading $d_1d_2d_3d_4$, so that
it can be compared to the fifth data value. Replacing $d_1$ by $d_1'$
does make a difference but we cannot detect it by comparing
$f(\underline{d_1d_2d_3d_4} \mid v)$ and $f(\underline{d_1'd_2d_3d_4} \mid v)$. We
can detect it as follows: $f(d_{1}d_{2}d_{3}d_{4}\mid
\underline{d_{1}})=(d_{3},3)\ne (d_{4},4)=f(d_{1}d_{2}d_{3}d_{4} \mid
\underline{d_{5}})$. Changing the suffix from $d_1$ to $d_5$
influences how the prefix $d_1d_2d_3d_4$ is transformed (in
transductions, prefixes are vulnerable to the influence of suffixes).
The value $d_1$ is also contained in the prefix $d_1d_2$, but $f(d_1d_2\mid
\underline{v})=f(d_1d_2\mid \underline{v[d_1/d_5]})$ for all $v$.
To detect that $d_1d_2$ is vulnerable, we
first need to append $d_3d_4$ to $d_1d_2$ and then have a suffix in which we
substitute $d_1$ with something else. We formalize this in the
definition below; it can be related to the example above by setting
$u=d_1d_2$, $u'=d_3d_4$ and $v=d_1$.
\begin{definition}[vulnerable values]
A data value $d$ is \emph{$f$-vulnerable} in a data word $u$
if there exist data words $u',v$ and a data value $d'$ such that $d$
does not occur in $u'$, $d'$ is a safe replacement for $d$ in $u\cdot
u'\cdot v$ and $f(u\cdot u' \mid \underline{v[d/d']}) \ne f(u\cdot u'
\mid \underline{v})$.
\end{definition}

Consider the transduction $f$ defined as $f(u)=f_1(u)\cdot f_2(u)$;
for $i \in [1,2]$, $f_i$ reverses its input if the $i$\tsc{th} and
last data values are distinct. On other inputs, $f_i$ is the identity
($f_1$ is the transduction given in
Example~\ref{ex:identityOrReverse}). In the two words
$d_1d_2d_3d_1d_2d_3$ and $d_1d_2d_3d_2d_1d_3$, $d_1$ and $d_2$ are
$f$-memorable. For every data word $v$,
$f(\underline{d_1d_2d_3d_1d_2d_3}\mid
v)=f(\underline{d_1d_2d_3d_2d_1d_3}\mid v)$, so it is tempting to say
that the two words are equivalent. But after reading
$d_1d_2d_3d_1d_2d_3$, a transducer would remember that $d_2$ is the
latest $f$-memorable value it has seen. After reading
$d_1d_2d_3d_2d_1d_3$, the transducer would remember that $d_1$ is the
latest $f$-memorable value it has seen. Different $f$-memorable values
play different roles and one way to distinguish which is which is to
remember the order in which they occurred last. So we distinguish between 
$d_1d_2d_3d_1d_2d_3$ and
$d_1d_2d_3d_2d_1d_3$. Suppose $d_{2},d_{1}$ are two data values in
some data word $u$. We say that $d_{1}$ is \emph{fresher} than $d_{2}$ in
$u$ if the last occurrence of $d_{1}$ in $u$ is to the right of the
last occurrence of $d_{2}$ in $u$.
\begin{definition}
    \label{def:iflSeq}
Suppose $f$ is a transduction and $u$ is a data word. We say that a
data value $d$ is $f$-influencing in $u$ if it is either $f$-memorable
or $f$-vulnerable in $u$. We denote by
$\ifl_{f}(u)$ the sequence $d_{m} \cdots d_{1}$, where
$\set{d_m,\ldots, d_1}$ is the set of all $f$-influencing values in
$u$ and for all $i \in [1,m-1]$, $d_{i}$ is fresher than
$d_{i+1}$ in $u$. We call $d_{i}$ the
$i$\textsuperscript{th} $f$-influencing data value in $u$. If a data
value $d$ is both $f$-vulnerable and $f$-memorable in $u$, we say
that $d$ is of \emph{type} $\iflsp$. If $d$ is $f$-memorable
but not $f$-vulnerable (resp.~$f$-vulnerable but not
$f$-memorable) in $u$, we say that $d$ is of type $\ifls$
(resp.~$\iflp$). We denote by $\aifl_{f}(u)$ the sequence $(d_{m},
t(d_{m})) \cdots (d_{1}, t(d_{1}))$, where $t(d_{i})$ is the type of
$d_{i}$ for all $i \in [1,m]$.
\end{definition}
To consider two data words $u_1$ and $u_2$ to be equivalent, we can
insist that $\aifl_f(u_1)=\aifl_f(u_2)$. But as before, this may
result in some infinite set of pairwise non-equivalent data words. We will
relax the condition by saying that there must be a permutation $\pi$
on data values such that $\aifl_f(\pi(u_2))=\aifl_f(u_1)$. This is
still not enough; we have overlooked one more thing that must be
considered in such an equivalence. Recall that in transductions,
prefixes are vulnerable to the influence of suffixes. So if $u_1$ is
vulnerable to changing the suffix from $v_1$ to $v_2$, then $\pi(u_2)$
must also have the same vulnerability. This is covered by the third
condition in the definition below.
\begin{definition}
    \label{def:fEquivalence}
    For a transduction $f$, we define the relation $\fequiv$ on data
    words as $u_{1} \fequiv u_{2}$ if there exists
    a permutation $\pi$ on data values satisfying the following
    conditions:
    \begin{itemize}
        \item $\lambda v.f_{z}(\underline{\pi(u_{2})} \mid v) = \lambda
	    v.f(\underline{u_{1}} \mid v)$, where $z=|u_{1}| -
            |u_{2}|$,
        \item $\aifl_{f}(\pi(u_{2})) = \aifl_{f}(u_{1})$ and
        \item for all $u,v_{1},v_{2}$, $f(u_{1} \cdot u \mid
            \underline{v_{1}}) = f(u_{1} \cdot u \mid \underline{v_{2}})$ iff
            $f(\pi(u_{2}) \cdot u \mid \underline{v_{1}}) = f(\pi(u_{2}) \cdot u \mid
            \underline{v_{2}})$.
    \end{itemize}
\end{definition}
As in the standard lambda calculus notation, $\lambda
v.f_{z}(\underline{u}\mid v)$ denotes the function that maps each
input $v$ to $f_{z}(\underline{u}\mid v)$. It is routine to verify
that for any data word $u$ and permutation $\pi$, $\pi(u) \fequiv
u$, since $\pi$ itself satisfies all the conditions above.
\begin{lemma}
    \label{lem:fEquivalence}
    If $f$ is invariant under permutations,
    then $\fequiv$ is an equivalence relation.
\end{lemma}

\iffalse
In the third condition of Definition~\ref{def:fEquivalence}, we take
two prefixes $u_1$ and $u_2$ and check if every suffix has the
property that it influences $u_1$ and $u_2$ in the same way. We apply
a permutation to one of the prefixes before checking this. But now, we
want to take two suffixes (say $v_1$ and $v_2$) and check whether
\emph{every prefix} has the property that it is influenced similarly
by $v_1$ and $v_2$. Before we check this, we must apply permutations
to every prefix; different prefixes may need different permutations.
This is formalized below.
\begin{definition}
    Let $f$ be a transduction and $\Pi$ be the set of all permutations
    on $D$. An \emph{equalizing scheme} for $f$ is a function $E:
    (\Sigma \times D)^{*} \to \Pi$ such that there exists a sequence
    $\delta_{1} \delta_{2} \cdots $ of data values satisfying the
    following condition: for every data word $u$ and every $i \in
    [1,|\ifl_f(u)|]$, the $i$\textsuperscript{th} $f$-influencing data
    value of $E(u)(u)$ is $\delta_{i}$.
\end{definition}
Note that $E(u)(u)$ denotes the application of the permutation $E(u)$
to the data word $u$.\footnote{Just $E(u)$ would be a better notation
instead of $E(u)(u)$. But later we will need to apply the inverse of
the permutation $E(u)$ to some other data word $v$, so we will write
$E(u)^{-1}(v)$ for that.} Left parts that have been equalized like
this will not have arbitrary influencing data values --- they will be
from the sequence $\delta_{1} \delta_{2} \cdots$. For the transduction
in Example~\ref{ex:identityOrReverse}, the first data value is the
only influencing value in any data word. An equalizing scheme will map
the first data value of all data words to $\delta_1$.

\begin{definition}
    For a transduction $f$ and equalizing scheme $E$, we define
    the relation $\requiv$ on data words as
    $v_{1} \requiv v_{2}$ if for every data word $u$,
    $f( E(u)(u) \mid \underline{v_{1}}) = f( E(u)(u) \mid
    \underline{v_{2}})$.
\end{definition}
It is routine to verify that $\requiv$ is an equivalence relation.
If $E_{1}$ and $E_{2}$ are two distinct equalizing schemes for
$f$, then in general $\equiv_{f}^{E_{1}}$ and $\equiv_{f}^{E_{2}}$ are
different. However, what matters is the index of the equivalence
relations.
\begin{lemma}
    \label{lem:equalizerUnimportant}
    Suppose $f$ is a transduction that is invariant under permutations
    and without data peeking and $E_{1}$, $E_{2}$ are equalizing
    schemes.
    Then $\equiv_{f}^{E_{1}}$ and $\equiv_{f}^{E_{2}}$ have
    the same index.
\end{lemma}
\fi

We denote by $[u]_f$ the equivalence class of $\fequiv$ containing $u$.
Following is the main result of this paper.
\begin{theorem}
    \label{thm:mainResult}
    A transduction $f$ is implemented by a SSRT iff $f$ satisfies the
    following properties: 1)$f$ is invariant under permutations, 2)
    $f$ is without data peeking, 3) $f$ has linear blow up and 4) $\fequiv$
    has finite index.
\end{theorem}


\section{Constructing a SSRT from a Transduction}
\label{sec:constSSRT-abs}
In this section, we prove the reverse direction of
Theorem~\ref{thm:mainResult}, by showing how to construct a
SSRT that implements a transduction, if it satisfies the four
conditions in the theorem. SSRTs read their input from left to right. Our first task is to get
SSRTs to identify influencing data values as they are read one by one.
Suppose a transducer that is intended to implement a transduction $f$
has read a data word $u$ and has stored in its registers the data
values that are $f$-influencing in $u$. Suppose the transducer reads
the next symbol $(\sigma,e)$. To identify the data values that are
$f$-influencing in $u \cdot (\sigma,e)$, will the transducer need to
read the whole data word $u \cdot (\sigma,e)$ again? The answer turns
out to be no, as the following result shows. The only data values that
can possibly be $f$-influencing in $u \cdot (\sigma,e)$ are $e$ and
the data values that are $f$-influencing in $u$.
\begin{lemma}
    \label{lem:iflValsMonotonic}
    Let $f$ be a transduction, $u$ be a data word, $\sigma \in \Sigma$
and $d,e$ be distinct data values. If $d$ is not $f$-memorable
(resp.~$f$-vulnerable) in $u$, then $d$ is not $f$-memorable
(resp.~$f$-vulnerable) in $u\cdot (\sigma,e)$.
\end{lemma}

Next, suppose that $d$ is $f$-influencing in $u$. How will we get the
transducer to detect whether $d$ continues to be $f$-influencing in $u
\cdot (\sigma,e)$? The following result provides a partial answer. If
$u_1 \fequiv u_2$ and the $i$\tsc{th} $f$-influencing value in $u_1$
continues to be $f$-influencing in $u_1 \cdot (\sigma,e)$, then the
$i$\tsc{th} $f$-influencing value in $u_2$ continues to be
$f$-influencing in $u_2 \cdot (\sigma,e)$. The following result
combines many such similar results into a single one.
\begin{lemma}
    \label{lem:rightCongruences}
    Suppose $f$ is a transduction that is invariant under permutations
    and without data peeking. Suppose $u_{1}, u_{2}$ are data words such
    that $u_{1} \fequiv u_{2}$, $\ifl_f(u_{1}) =
    d_{1}^{m} d_{1}^{m-1} \cdots d_{1}^{1}$ and $\ifl_f(u_{2}) =
    d_{2}^{m} d_{2}^{m-1} \cdots d_{2}^{1}$. Suppose $d_{1}^{0} \in D$
    is not $f$-influencing in $u_1$, $d_{2}^{0} \in D$
    is not $f$-influencing in $u_2$ and $\sigma \in \Sigma$. For all $i,j \in
[0,m]$, the following are true:
    \begin{enumerate}
        \item\label{rc:iflPreserved} $d_{1}^{i}$ is $f$-memorable
	(resp.~$f$-vulnerable) in $u_{1} \cdot (\sigma,d_{1}^{j})$ iff
	    $d_{2}^{i}$ is $f$-memorable (resp.~$f$-vulnerable) in
	$u_{2} \cdot (\sigma,d_{2}^{j})$.
        \item\label{rc:fEqPreserved} $u_{1} \cdot (\sigma,d_{1}^{j}) \fequiv u_{2} \cdot
            (\sigma,d_{2}^{j})$.
    \end{enumerate}
\end{lemma}

If $u_1\fequiv u_2$, there exists a permutation $\pi$ such that
$\aifl_f(u_1)=\aifl_f(\pi(u_2))$. Hence, all data words in the same
equivalence class of $\fequiv$ have the same number of $f$-influencing
values. If $\fequiv$ has finite index, then there is a bound (say $I$)
such that any data word has at most $I$ $f$-influencing data values.
We are going to construct a SSRT to identify $f$-influencing data
values. The construction is technically involved, so we motivate it by
stating the end result first.
Consider a SSRT $S_f^\ifl$ with the set of registers $R = \set{r_1, \ldots,
r_I}$. The states are of the form $([u]_{f},\ptr)$, where $u$ is some
data word and $\ptr:[1,|\ifl_f(u)|] \to R$ is a pointer function. Let
$\ptr_\bot$ be the trivial function from $\emptyset$ to $R$. The
transitions can be designed to satisfy the following.
\begin{lemma}
    \label{lem:SSRTIflValInv}
    Suppose the SSRT $S_f^\ifl$ starts in the configuration
	$(([\epsilon]_{f},\ptr_{\bot}),\val_{\epsilon},0)$ and reads some data
	word $u$. It reaches the configuration $(([u]_{f},\ptr),$ $\val,|u|)$
	such that $\val(\ptr(i))$ is the $i$\tsc{th} $f$-influencing value in
	$u$ for all $i \in [1,|\ifl_{f}(u)|]$.
\end{lemma}
In short, the idea is that we can hard code rules such as ``if the
data value just read is the $i$\tsc{th} $f$-influencing value in $u$, it
continues to be $f$-influencing in the new data word''.
Lemma~\ref{lem:rightCongruences} implies that the validity of such
rules depend only on the equivalence class $[u]_f$ containing $u$ and
does not depend on $u$ itself. So the SSRT need not remember the
entire word $u$; it just remembers the equivalence class $[u]_f$ in
its control state. The SSRT can check whether the new
data value is the $i$\tsc{th} $f$-influencing value in $u$, by
comparing it with the register $\ptr(i)$. To give the full details of
constructing $S_f^\ifl$, we need another concept explained in the
following paragraph.

Recall the transduction $f$ from Example~\ref{ex:identityOrReverse}
and the infinite set of data words $\set{(a,d_i)\mid i \ge 1}$. For
any $i \ne j$, $f(\underline{(a,d_i)} \mid (a,d_i)) \ne
f(\underline{(a,d_j)} \mid (a,d_i))$ for $i \ne j$. But for every $i$,
there is a permutation $\pi_i$ on data values mapping $d_i$ to $d_1$
so that $f(\underline{\pi_i(a,d_i)} \mid v) =
f(\underline{(a,d_1)}\mid v)$ for any data word $v$.  We have revealed
that all data words in $\set{(a,d_i)\mid i \ge 1}$ are equivalent by
applying a permutation to each data word, so that they all have the
same $f$-influencing data values. We formalize this idea below.
\begin{definition}
    \label{def:equalScheme}
    Let $f$ be a transduction and $\Pi$ be the set of all permutations
    on the set of data values $D$. An \emph{equalizing scheme} for $f$ is a function $E:
    (\Sigma \times D)^{*} \to \Pi$ such that there exists a sequence
    $\delta_{1} \delta_{2} \cdots $ of data values satisfying the
    following condition: for every data word $u$ and every $i \in
    [1,|\ifl_f(u)|]$, the $i$\textsuperscript{th} $f$-influencing data
    value of $E(u)(u)$ is $\delta_{i}$.
\end{definition}
Note that $E(u)(u)$ denotes the application of the permutation $E(u)$
to the data word $u$. We will write $E(u)(u)$ as $u_q$ for short
(intended to be read as ``equalized $u$'').
Note that $E(u)^{-1}(u_q)=u$.

Now we give the full details of constructing $S_f^\ifl$.
\begin{construction}
    \label{const:SSRTIflVal}
    Suppose $f$ is a transduction that is invariant under
permutations,
    $\fequiv$ has finite index and $E$ is an equalizing scheme. Let $I$ be the maximum number of
    $f$-influencing data values in any data word and $\delta_1 \cdots
\delta_I \in D^*$ be such that for any data word $u$, $\delta_i$ is the
$i$\textsuperscript{th} $f$-influencing value in $u_q$. Consider a SSRT $S_f^\ifl$
    with the set of registers $R=\set{r_{1}, \ldots, r_{I}}$. The
    states are of the form $([u]_{f},\ptr)$, where
    $u$ is some data word and $\ptr:[1,|\ifl_f(u)|] \to R$ is a pointer function. If $|\ifl_f(u)|=0$, then $\ptr=\ptr_{\bot}$, the trivial
    function from $\emptyset$ to $R$. We let the set $X$ of data word
    variables to be empty. Let $\update_{\bot}$ be the trivial update
    function for the empty set $X$. The initial state is
    $([\epsilon]_{f},\ptr_{\bot})$. Let $\delta_0$ be an arbitrary
data value in $D\setminus
\set{\delta_1, \ldots, \delta_I}$. From a state $([u]_f,\ptr)$, for every $\sigma \in \Sigma$ and
$i \in [0,|\ifl_f(u)|]$, there is a transition $( ([u]_f, \ptr),
\sigma, \phi, ([u_q\cdot (\sigma,\delta_i)]_f, \ptr'), R',
\update_\bot)$. The condition $\phi$ is as follows, where $m =
|\ifl_f(u)|$:

\begin{align*}
 \phi = \begin{cases}
\bigwedge_{j=1}^{m}\ptr(j)^{\ne} & i=0\\
\phi = \ptr(i)^{=} \land \bigwedge_{j \in
            [1,m]\setminus \set{i}}\ptr(j)^{\ne} & i \ne 0
\end{cases}
\end{align*} 
For every $j\in [1,|\ifl_f(u_q\cdot (\sigma,\delta_i))|]$, $\ptr'(j)$
is as follows: if the $j$\textsuperscript{th} $f$-influencing
value of $u_q\cdot (\sigma,\delta_i)$ is the
$k$\textsuperscript{th} $f$-influencing value of $u_q$ for some
$k$, then $\ptr'(j)=\ptr(k)$. Otherwise, $\ptr'(j)=r_\reu = \min(R \setminus \set{\ptr(k) \mid 1 \le k \le m,
\delta_k \text{ is } f\text{-influencing in } u_q\cdot
(\sigma,\delta_i)})$ (minimum is based on the order $r_1 <
r_2 < \cdots < r_I$). The set $R'$ is $\set{r_\reu}$ if $i=0$ and
$\delta_0$ is $f$-influencing in $u_q\cdot (\sigma,\delta_0)$;
$R'$ is $\emptyset$ otherwise.
\end{construction}
It is routine to verify that the SSRT constructed above is
deterministic.  The definition of the next pointer function $\ptr'$
ensures that the register $\ptr(j)$ always stores the
$j$\textsuperscript{th} $f$-influencing value in the data word read so
far. This is shown in the proof of Lemma~\ref{lem:SSRTIflValInv},
which can be found in Section~\ref{sec:recIflVals}.

Next we will extend the transducer to compute the output of a transduction.
Suppose the transducer has read the data word $u$ so far. The
transducer doesn't know what is the suffix that is going to come, so
whatever computation it does has to cover all possibilities. The idea
is to compute $\set{f(u\mid \underline{v}) \mid v \in (\Sigma\times
D)^*}$ and store them in data word variables, so that when it has to
output $f(u)$ at the end, it can output $f(u \mid
\underline{\epsilon})$.  However, this set can be infinite. If
$\fequiv$ has finite index, we can reduce it to a finite set.  

Left parts that have been equalized by an equalizing scheme will not
have arbitrary influencing data values --- they will be from the
sequence $\delta_{1} \delta_{2} \cdots$. For the transduction in
Example~\ref{ex:identityOrReverse}, the first data value is the only
influencing value in any data word. An equalizing scheme will map the
first data value of all data words to $\delta_1$.

The relation $\fequiv$ identifies two prefixes when they behave
similarly. We now define a relation that serves a similar propose, but
for suffixes.
\begin{definition}
    \label{def:rEquivalence}
    For a transduction $f$ and equalizing scheme $E$, we define
    the relation $\requiv$ on data words as
    $v_{1} \requiv v_{2}$ if for every data word $u$,
    $f( u_q \mid \underline{v_{1}}) = f( u_q \mid
    \underline{v_{2}})$.
\end{definition}
It is routine to verify that $\requiv$ is an equivalence relation.
Saying that $v_1$ and $v_2$ are similar suffixes if $f(u\mid
\underline{v_1}) = f(u \mid \underline{v_2})$ for all $u$ doesn't
work; this may result in infinitely many pairwise unequivalent
suffixes (just like $\fequiv$ may have infinite index if we don't
apply permutations to prefixes). So we ``equalize'' the prefixes so
that they have the same $f$-influencing data values, before checking
how suffixes influence them.
\begin{lemma}
\label{lem:finREquiv}
	Suppose $f$ is a transduction satisfying all the conditions of
Theorem~\ref{thm:mainResult}. If $E$ is an equalizing scheme for $f$,
then $\requiv$ has finite index.
\end{lemma}

Suppose we are trying to design a SSRT to implement a transduction
$f$, which has the property that $\requiv$ has finite index. The SSRT
can compute the set
$\set{f(u_q\mid \underline{v}) \mid v \in (\Sigma \times D)^*}$,
which is
finite (it is enough to consider one representative $v$ from every
equivalence class of $\requiv$). At the end when the SSRT has to
output $f(u)$, it can output $E(u)^{-1}(f(u_q \mid
\underline{\epsilon})) = f(u)$. The SSRT never knows what is the next
suffix; at any point of time, the next suffix could be $\epsilon$. So
the SSRT has to apply the permutation $E(u)^{-1}$ at each step.
Letting $V$ be a finite set of representatives from every equivalence
class of $\requiv$,  the SSRT computes $\set{f(u \mid
\underline{E(u)^{-1}(v)}) \mid v \in V}$ at every step.

Now suppose the SSRT has computed $\set{f(u \mid
\underline{E(u)^{-1}(v)}) \mid v \in V}$, stored them in data word
variables and it reads the next symbol
$(\sigma,d)$. The SSRT has to compute $\set{f(u \cdot (\sigma,d) \mid
\underline{E(u\cdot (\sigma,d))^{-1}(v)}) \mid v \in V}$ from
whatever it had computed for $u$.

To explain how the above computation is done, we use some terminology.
In factored outputs of the form $f(u \mid \underline{v})$,
$f(\underline{u} \mid \underline{v})$, $f(\underline{u} \mid v \mid
\underline{w})$ or $f(\underline{u} \mid \underline{v} \mid
\underline{w})$, a triple is said to come from $u$ if it has origin in
$u$ or it is the triple $(*,*,\lef)$. A left block in such a factored
output is a maximal infix of triples, all coming from the left part
$u$. Similarly, a non-right block is a maximal infix of triples, none
coming from the right part. Middle blocks are defined similarly.
For the transduction $f$ in Example~\ref{ex:identityOrReverse},
$f((a,d_{1}) (b,d_2)(c,d_{3}))$ is
$(c,d_{3},3)(b,d_2,2)(a,d_{1},1)$. In $f((a,d_{1}) (b,d_2) \mid
\underline{(c,d_{3})})$, $(b,d_2,2)(a,d_{1},1)$ is a left block. In
$f(\underline{(a,d_1)} \mid (b,d_2) \mid \underline{(c,d_3)})$,
$(b,d_2, 2)$ is a middle block. In $f(\underline{(a,d_1)} \mid
\underline{(b,d_2)} \mid \underline{(c,d_3)})$, $(*,*,\mi)(*,*,\lef)$
is a non-right block, consisting of one middle and one left block.

The concretization of the $i$\tsc{th} left block (resp.~middle block)
in $f(\underline{u} \mid \underline{v} \mid \underline{w})$ is defined
to be the $i$\tsc{th} left block in $f(u \mid \underline{vw})$
(resp.~the $i$\tsc{th} middle block in $f(\underline{u} \mid v \mid
\underline{w})$). The concretization of the $i$\tsc{th} non-right
block in $f(\underline{u} \mid \underline{v} \mid \underline{w})$ is
obtained by concatenating the concretizations of the left and middle
blocks that occur in the $i$\tsc{th} non-right block.  The following
is a direct consequence of the definitions.
\begin{proposition}
\label{prop:leftBlNRBlock}
The $i$\textsuperscript{th} left block of $f(u\cdot (\sigma,d) \mid
\underline{v})$ is the concretization of the $i$\tsc{th} non-right block of
$f(\underline{u} \mid \underline{(\sigma,d)} \mid \underline{v})$.
\end{proposition}
For the transduction $f$ from Example~\ref{ex:identityOrReverse}, the
first left block of $f((a,d_1)(b,d_2) \mid \underline{(c,d_3)})$ is
$(b,d_2,2)(a,d_1,1)$, which is the concretization of
$(*,*,\mi)(*,*,\lef)$, the first non-right block of
$f(\underline{(a,d_1)} \mid \underline{(b,d_2)} \mid
\underline{(c,d_3)})$.

From Proposition~\ref{prop:leftBlNRBlock}, we deduce that the
$i$\tsc{th} left block of $f(u \cdot (\sigma,d) \mid
\underline{E(u\cdot (\sigma,d))^{-1}(v)})$ is the concretization of
the $i$\tsc{th} non-right block of $f(\underline{u} \mid
\underline{(\sigma,d)} \mid \underline{E(u\cdot
(\sigma,d))^{-1}(v)})$. The concretizations come from the left blocks
of $f(u \mid \underline{(\sigma,d) \cdot E(u\cdot
(\sigma,d))^{-1}(v)})$ and the middle blocks of $f(\underline{u} \mid
(\sigma,d) \mid \underline{E(u\cdot (\sigma,d))^{-1}(v)})$. In the
absence of data values, the above two statements would be as follows:
The $i$\tsc{th} left block of $f(u \cdot \sigma \mid \underline{v})$
is the concretization of the $i$\tsc{th} non-right block of
$f(\underline{u} \mid \underline{\sigma} \mid \underline{v})$. The
concretizations come from the left blocks of $f(u \mid
\underline{\sigma \cdot v})$ and the middle blocks of $f(\underline{u}
\mid \sigma \mid \underline{v})$. This technique of incrementally
computing factored outputs was introduced in \cite{Bojanczyk2014} for
SSTs. In SSTs, $f(u \mid \underline{\sigma \cdot v})$ would have been
computed as $f(u \mid \underline{v'})$ when $u$ was read, where $v'$ is
some word that influences prefixes in the same way as $\sigma
\cdot v$. But in SSRTs, only $f(u \mid \underline{E(u)^{-1}(v')})$ would
have been computed for various $v'$; what we need is $f(u \mid
\underline{(\sigma,d) \cdot E(u\cdot (\sigma,d))^{-1}(v)})$. We work
around this by proving that a $v'$ can be computed such that $f(u \mid
\underline{(\sigma,d) \cdot E(u\cdot (\sigma,d))^{-1}(v)}) = f(u \mid
\underline{E(u)^{-1}(v')})$. This needs some technical work, which
follows next.

SSRTs will keep left blocks in
variables, so we need a bound on the number of blocks.
\begin{lemma}
    \label{lem:finiteLeftBlocks}
    Suppose $f$ is a transduction that is invariant under permutations
    and has linear blow up and $E$ is an equalizing scheme
    such that $\requiv$ has finite index. There
    is a bound $B \in \Nat$ such that for all data words $u,v$, the
    number of left blocks in $f(u \mid \underline{v})$ is at most $B$.
\end{lemma}

\begin{definition}
    \label{def:middleBlocks}
    Suppose $\ifl_f(u_q)=\delta_m \cdots \delta_1$, $\delta_0 \in D
    \setminus \set{\delta_{m}, \ldots, \delta_{1}}$, $\eta \in
    \set{\delta_{0}, \ldots, \delta_m}$ and $\sigma
	\in \Sigma$. We say that a permutation $\pi$
\emph{tracks influencing values} on $u_q\cdot (\sigma,\eta)$ if
$\pi(\delta_{i})$ is the $i$\tsc{th} $f$-influencing value
    in $u_q \cdot (\sigma,\eta)$ for all $i \in [1,|\ifl_f(u_q \cdot (\sigma,\eta))|]$.
\end{definition}
    Lemma~\ref{lem:iflValsMonotonic} implies that for $i \ge 2$ in the
above definition, $\pi(\delta_i) \in \set{\delta_{m}, \ldots,
\delta_{1}}$ and $\pi(\delta_1) \in \set{\delta_{m}, \ldots,
\delta_{0}}$. We can infer from Lemma~\ref{lem:rightCongruences} that
if $u \fequiv u'$ and $\pi$ tracks influencing values on
$E(u')(u')\cdot (\sigma,\eta)$, then it also tracks influencing values
on $u_q \cdot (\sigma,\eta)$.

\begin{lemma}
    \label{lem:computeLeftBlocks} 
Suppose $f$ is a transduction that is invariant under permutations and
without data peeking, $u,u',v$ are data words, $\sigma \in
\Sigma$, $\ifl_{f}(u)=d_{m} \cdots d_{1}$, $d_0 \in D \setminus
\set{d_{m}, \ldots, d_{1}}$, $\delta_0 \in D \setminus
\set{\delta_{m}, \ldots, \delta_{1}}$, $(d,\eta) \in
\set{(d_{i},\delta_{i}) \mid i \in [0,  m]}$, $\pi$ tracks influencing
values on $u_q\cdot (\sigma,\eta)$ and $u \fequiv u'$.
Then
$f(u \mid \underline{(\sigma,d) \cdot E(u\cdot (\sigma,d))^{-1}(v)}) =
f(u \mid \underline{E(u)^{-1}((\sigma,\eta)\cdot\pi(v))})$.
If $(d,\eta) \in
\set{(d_{i},\delta_{i}) \mid i \in [1,  m]}$, then
$f(\underline{u} \mid (\sigma,d) \mid \underline{E(u\cdot
(\sigma,d))^{-1}(v)}) = E(u)^{-1}(f_{z}(\underline{u'_q} \mid
(\sigma,\eta) \mid \underline{\pi(v)}))$, where $z=|u|-|u'|$.  If
$(d,\eta)=(d_0,\delta_0)$, then $f(\underline{u} \mid (\sigma,d) \mid
\underline{E(u\cdot (\sigma,d))^{-1}(v)}) = E(u)^{-1} \pcomp
\pi'(f_{z}(\underline{u'_q} \mid (\sigma,\eta) \mid
\underline{\pi(v)}))$, where $\pi'$ is the permutation that
interchanges $\delta_0$ and $E(u)(d_0)$ and doesn't change any other
data value ($\pcomp$ denotes composition of permutations).
\end{lemma}
The left blocks of $f(u \mid \underline{(\sigma,d) \cdot E(u\cdot
(\sigma,d))^{-1}(v)})$ are hence equal to those of the factored output $f(u \mid
\underline{E(u)^{-1}((\sigma,\eta)\cdot\pi(v))})$, which would have
been be stored as $f(u \mid
\underline{E(u)^{-1}(v')})$ in one of the data word variables when $u$
was read, where $v' \requiv (\sigma,\eta)\cdot\pi(v)$.

Suppose $v_1,v_2 \in V$ and $v'\requiv (\sigma,\eta)\cdot\pi(v_1)
\requiv (\sigma,\eta)\cdot\pi(v_2)$. The computation of $f(u \cdot
(\sigma,d) \mid \underline{E(u \cdot (\sigma,d))^{-1}(v_1)})$ requires the
left blocks of $f(u \mid \underline{E(u)^{-1}(v')})$ and the
computation of 
$f(u \cdot (\sigma,d) \mid \underline{E(u\cdot
(\sigma,d))^{-1}(v_2)})$ also requires the
left blocks of $f(u \mid \underline{E(u)^{-1}(v')})$. The
SSRT would have stored $f(u\mid \underline{E(u)^{-1}(v')})$ in a data
word variable and now it is needed for two computations. But in SSRTs,
the contents of one data word variable cannot be used in two
computations, since SSRTs are copyless. This problem is solved in
\cite{Bojanczyk2014} for SSTs using a two way transducer model
equivalent to SSTs. In this two way model, the suffix can be read and
there is no need to perform computations for multiple suffixes. We
cannot use that technique here, since there are no known two way
models equivalent to SSRTs.

We solve this problem by not performing the two computations
immediately. Instead, we remember the fact that there is a multiple
dependency on a single data word variable. The actual computation is
delayed until the SSRT reads more symbols from the input and gathers
enough information about the suffix to discard all but one of the
dependencies. Suppose we have delayed computing $f(u \cdot (\sigma,d)
\mid \underline{E(u\cdot (\sigma,d))^{-1}(v_1)})$ due to some
dependency. After reading the next symbol, $f(u \cdot (\sigma,d) \mid
\underline{E(u\cdot (\sigma,d))^{-1}(v_1)})$ itself might be needed
for multiple computations. We keep track of such nested dependencies
in a tree data structure called dependency tree. Dependency trees can
grow unboundedly, but if $\requiv$ has finite index, it can be shown
that some parts can be discarded from time to time to keep their size
bounded. We store such reduced dependency trees as part of the control
states of the SSRT. The details of this construction constitute the
rest of this section. 

For a transduction $f$, let $B$ be the maximum of the bounds on the
number of left blocks shown in Lemma~\ref{lem:finiteLeftBlocks} and
the number of middle blocks in factored outputs of the form
$f(\underline{u} \mid (\sigma,d) \mid \underline{v})$. Let $(\Sigma\times
D)^{*}/\requiv$ be the set of equivalence classes of $\requiv$, let
$\hat{X}=\set{\bv{\theta,i} \mid \theta \in ((\Sigma\times
D)^{*}/\requiv)^*, 1 \le i \le B^{2}+B}$ and for $\theta \in
((\Sigma\times D)^{*}/\requiv)^*$, let $X_{\theta}=\set{\bv{\theta,i}
\mid 1 \le i \le B^{2}+B}$.  We denote by $\theta\drlst$ the sequence
obtained from $\theta$ by removing the right most equivalence class.
We use a set $\pbl=\set{P_{1}, \ldots, P_{B}}$ of parent references in
the following definition. We use a finite subset of $\hat{X}$ as
data word variables to construct SSRTs.
\begin{definition}
    \label{def:anticipateTree}
    Suppose $f$ is a transduction and $E$ is an equalizing scheme for
    $f$. A \emph{dependency tree} $T$ is a tuple $(\Theta,\pref,\bl)$,
    where the set of nodes $\Theta$ is a prefix closed finite subset of
    $((\Sigma\times D)^{*}/\requiv)^*$ and $\pref,\bl$ are labeling
    functions. The root is $\epsilon$ and if $\theta \in \Theta
\setminus \set{\epsilon}$, its parent is
    $\theta\drlst$.  The labeling functions are $\pref:\Theta \to
    (\Sigma\times D)^{*}/\fequiv$ and $\bl: \Theta \times [1,B] \to
    (\hat{X} \cup \pbl)^{*}$. We call $\bl(\theta,i)$ a \emph{block
    description}. The dependency tree is said to be \emph{reduced} if the
    following conditions are satisfied:
    \begin{itemize}
        \item every sequence $\theta$ in $\Theta$ has length that is bounded
	by $|(\Sigma\times D)^{*}/\requiv|+1$,
        \item $\pref$ labels all the leaves with a single equivalence
            class of $\fequiv$,
\item for every equivalence class $[v]_{f}^{E}$, there is exactly one
leaf $\theta$ such that the last equivalence class in $\theta$ is $[v]_{f}^{E}$,
        \item $\bl(\theta,i)\in (X_{\theta}\cup \pbl)^{*}$ and is of
            length at most $2B+1$ for all $\theta \in \Theta$ and $i
            \in [1,B]$ and
	\item for all $\theta \in \Theta$, each element of
	$X_{\theta}\cup \pbl$ occurs at most once in $\set{\bl(\theta,i) \mid
	1 \le i \le B}$.
\end{itemize}
\end{definition}

If $\fequiv$ and $\requiv$ have finite indices, there are finitely
many possible reduced dependency trees. Suppose $\theta=\theta'
\cdot [v]_{f}^{E}$ is in $\Theta$, $\pref(\theta)=[u]_{f}$ and
$\bl(\theta,1)=P_{1}\bv{\theta,1}P_{2}$. The intended meaning is that
there is a data word $u'$ that has been read by a SSRT and $u' \fequiv
u$.  The block description $\bl(\theta,1)=P_{1}\bv{\theta,1}P_{2}$ is
a template for assembling the first left block of $f(u'\mid
\underline{E(u')^{-1}(v)})$ from smaller blocks: take the first left
block in the parent node $\theta'$ ($P_1$ refers to the first left
block of the factored output assembled in the parent node), append to
it the contents of the data word variable $\bv{\theta,1}$, then append
the second left block in the parent node $\theta'$.  Intuitively, if
$u'=u'' \cdot (\sigma,d)$, then the first non-right block of
$f(\underline{u''} \mid \underline{(\sigma,d)} \mid
\underline{E(u')^{-1}}(v))$ is $(*,*,\lef)(*,*,\mi)(*,*,\lef)$ and
$P_{1}$ refers to the concretization of the first left block
$(*,*,\lef)$, $\bv{\theta,1}$ contains the concretization of the first
middle block $(*,*,\mi)$ and so on. The first left block in the parent
node $\theta'$ itself may consist of some parent references and the
contents of some other data word variables. This ``unrolling'' is
formalized below.
\begin{definition}
    \label{def:anticipateUnroll}
    Suppose $T$ is a dependency tree with set of nodes $\Theta$.
    The function $\ur: \Theta \times (\hat{X} \cup \pbl)^{*}
    \to \hat{X}^{*}$ is defined as follows. For $\theta \in
    \Theta$ and $\mu \in (\hat{X} \cup \pbl)^{*}$, $\ur(\theta,\mu)$
    is obtained from $\mu$ by replacing every occurrence of a parent
    reference $P_{i}$ by $\ur(\theta\drlst,\bl(\theta\drlst,i))$
    (replace by $\epsilon$ if $\theta=\epsilon$) for all $i$.
\end{definition}
Intuitively, an occurrence of $P_{i}$ in $\mu$ refers to the
$i$\tsc{th} left block in the parent node. If the current node is
$\theta$, the parent node is $\theta\drlst$, so we unroll $\mu$ by
inductively unrolling the $i$\tsc{th} left block of $\theta$'s parent,
which is given by $\ur(\theta\drlst,\bl(\theta\drlst,i))$. We are
interested in dependency trees that allow to compute all factored
outputs of the form $f(u \mid \underline{E(u)^{-1}(v)})$ by unrolling
appropriate leaves.  For convenience, we assume that
$f(\epsilon)=\epsilon$. Let $T_{\bot}=(\set{\epsilon},
\pref_{\epsilon}, \bl_{\epsilon})$, where
$\pref_{\epsilon}(\epsilon)=[\epsilon]_{f}$ and
$\bl_{\epsilon}(\epsilon,i)=\epsilon$ for all $i \in [1,B]$.
\begin{definition}
    \label{def:completeAntTree}
    Suppose $f$ is a transduction, $\val$ is a valuation assigning a
    data word to every element of $\hat{X}$ and $T$ is a
    dependency tree. The pair $(T,\val)$ is complete for a data word
    $u$ if $u=\epsilon$ and $T=T_{\bot}$, or $u \ne \epsilon$ and the
    following conditions are satisfied: for every equivalence class
    $[v]_{f}^{E}$, there exists a \emph{leaf} node
    $\theta=\theta'\cdot [v]_{f}^{E}$ such that
    $\pref(\theta)=[u]_{f}$ and for every $i$, the $i$\tsc{th} left
    block of $f(u \mid \underline{E(u)^{-1}(v)})$ is
    $\val(\ur(\theta,\bl(\theta,i)))$.
\end{definition}
We construct SSRTs that will have dependency trees in its states,
which will be complete for the data word read so far. As more symbols
of the input data word are read, the dependency tree and the
valuation for $\hat{X}$ are updated as defined next.
\begin{definition}
    \label{def:extCompAntTree}
    Suppose $f$ is a transduction, $E$ is an equalizing scheme
    and $T$ is either $T_\bot$ or a reduced dependency tree in which $\pref$ labels all
the leaves with $[u]_f$ for some data word $u$. Suppose $\ifl_{f}(u)=d_{m} \cdots d_{1}$,
    $d_0 \in D \setminus \set{d_{m}, \ldots, d_{1}}$, $\delta_0 \in D \setminus
    \set{\delta_{m}, \ldots, \delta_{1}}$, $(d,\eta) \in
    \set{(d_{i},\delta_{i}) \mid i \in [0,  m]}$
    and $\sigma \in \Sigma$. Let $\pi$ be a permutation tracking
influencing values on $u_q\cdot (\sigma,\eta)$ as defined
    in Definition~\ref{def:middleBlocks}.
    For every equivalence class $[v]_{f}^{E}$, there is a leaf node
    $\theta_{v}=\theta'\cdot [(\sigma,\eta)\cdot \pi(v)]_{f}^{E}$
    (or $\theta_{v}=\epsilon$, the root of the trivial dependency
    tree in case $u = \epsilon$).
    Let $u'$ be an arbitrary data word in the equivalence class
    $[u]_{f}$. The $(\sigma,\eta)$ extension of $T$ is defined to be
the tree obtained from $T$ as
    follows: for every equivalence class $[v]_{f}^{E}$, create a new leaf
    $\theta=\theta_{v} \cdot [v]_{f}^{E}$ (with $\theta_v$ as parent)
	and set $\pref(\theta)=[u'_q \cdot (\sigma,\eta)]_{f}$. For every $i
	\in [1,B]$, let $z$ be the $i$\tsc{th} non-right block in
	$f(\underline{u'_q}\mid \underline{(\sigma,\eta)}
	\mid \underline{\pi(v)} )$ ($z$ is a sequence of left and
middle blocks). Let $z'$ be obtained from $z$ by replacing
	$j$\tsc{th} left block with $P_{j}$ and $k$\tsc{th} middle
	block with $\bv{\theta,k}$ for all $j,k$. Set $\bl(\theta,i)$ to be $z'$. If
	there are internal nodes (nodes that are neither leaves nor
        the root) of this extended tree which do not have any
	of the newly added leaves as descendants, remove such nodes. The
	resulting tree $T'$ is the $(\sigma,\eta)$ extension of $T$.
Suppose $\val$ is a valuation for $\hat{X}$ such that $(T, \val)$ is
complete for $u$. The
	$(\sigma,d)$ extension $\val'$ of $\val$ is defined to be the
valuation obtained from $\val$ by
	setting $\val'( \bv{\theta,k})$ to be the $k$\tsc{th} middle block of
	$f(\underline{u} \mid (\sigma,d) \mid \underline{E(u\cdot
	(\sigma,d))^{-1}(v)})$ for every newly added leaf
	$\theta=\theta_v \cdot [v]_{f}^{E}$ and
	every $k \in [1,B]$. For all other variables, $\val'$
        coincides with $\val$. We call $(T', \val')$ the
	$(\sigma,d)$ extension of $(T, \val)$. 
\end{definition}
If some internal nodes are removed as described in
Definition~\ref{def:extCompAntTree}, it means that some dependencies have
vanished due to the extension. For a newly added leaf $\theta$, every element of
	$X_{\theta}\cup \pbl$ occurs at most once in $\set{\bl(\theta,i) \mid
	1 \le i \le B}$.
\begin{lemma}
    \label{lem:extAntTreeInv}
    If $(T,\val)$ is complete for some data word $u$ and $(T',\val')$
    is the $(\sigma,d)$ extension of $(T,\val)$, then $(T',\val')$
    is complete for $u \cdot (\sigma,d)$.
\end{lemma}

If $(T,\val)$ is complete for $u$ and $(T',\val')$ is the
$(\sigma,d)$ extension of $(T,\val)$, then the data word
$\val'(\bv{\theta,k})$ is the $k$\tsc{th} middle block of
$f(\underline{u} \mid (\sigma,d) \mid \underline{E( u \cdot
(\sigma,d))^{-1}(v)})$. We call $\bv{\theta,k}$ a
\emph{new middle block variable} and refer to it
later for defining variable updates in transitions of SSRTs. The
tree $T'$ may not be reduced since it may contain branches that
are too long.  Next we see how to eliminate long branches.

\begin{definition}
    \label{def:antTreeShort}
    Suppose $T$ is a dependency tree. A shortening of $T$ is
    obtained from $T$ as follows: let $\theta$ be an internal node
    that has only one child. Make the child of $\theta$ a child of
    $\theta$'s parent, bypassing and removing the original node
    $\theta$. Any descendant $\theta \cdot \theta'$ of $\theta$ in $T$
    is now identified by $\theta\drlst \cdot \theta'$. Set
    $\pref(\theta\drlst \cdot \theta')$ to be $\pref(\theta\cdot
    \theta')$, the label given by $\pref$ for the original descendant
    $\theta\cdot \theta'$ in $T$. Suppose $\theta \cdot [v]_{f}^{E}$
    is the only child of $\theta$ in $T$. For every $i \in [1,B]$, set
    $\bl(\theta\drlst \cdot [v]_{f}^{E},i)=\mu$, where $\mu$ is
    obtained from $\bl(\theta\cdot [v]_{f}^{E},i)$ by replacing every
    occurrence of $P_{j}$ by $\bl(\theta,j)$. For strict descendants
    $\theta\drlst \cdot [v]_{f}^{E} \cdot \theta'$ of $\theta\drlst \cdot
    [v]_{f}^{E}$ and for every $i \in [1,B]$, set $\bl(\theta\drlst \cdot
    [v]_{f}^{E} \cdot \theta',i) = \bl(\theta\cdot [v]_{f}^{E} \cdot
    \theta',i)$.
\end{definition}
Intuitively, $\theta$ has only one child, so only
one factored output is dependent on the factored output stored in
$\theta$ (all but one of the dependencies have vanished). Therefore,
we can remove $\theta$ and pass on the information stored there to its
only child. This is accomplished by replacing any occurrence of $P_j$
in a block description of the child by $\bl(\theta,j)$.
Figure~\ref{fig:antTreeShort} shows an example, where $\theta_1$ is
the only child of $\theta$. So $\theta$ is removed, $\theta_1$ becomes
$\theta_2$ and a child of $\theta\drlst$.
\begin{figure}[!htp]
    \begin{center}
        \begin{tikzpicture}[>=stealth]
            \node[state, label=180:$\theta\drlst$] (a1) at (0cm,0cm) {};
            \node[state, label=180:$\theta$] (a2) at ([yshift=-1cm]a1) {};
            \node[state, label=180:${\theta_{1}=\theta\cdot [v]_{f}^{E}}$] (a3) at ([yshift=-1cm]a2) {};

            \draw (a1) -- (a2) -- (a3);

            \node (a4) at ([yshift=-0.7cm]a3) {$\bl(\theta_{1},1)=$};
            \node[rectangle split, rectangle split parts=3, rectangle
split horizontal=true, draw=gray!40] (a5) at ([yshift=-0.5cm]a4)
{$P_{1}$\nodepart{two}$\bv{\theta_{1},1}$\nodepart{three}$P_{2}$};
            \node at ([xshift=1.5cm, yshift=0.2cm]a2) {$\bl(\theta,1)=P_{1}\bv{\theta,1}$};
            \node at ([xshift=1.5cm, yshift=-0.2cm]a2) {$\bl(\theta,2)=P_{2}\bv{\theta,2}$};
            \draw[dotted] (a3) -- (a4);

            \node[state, label=0:$\theta\drlst$] (b1) at ([xshift=6cm, yshift=-0.5cm]a1) {};
            \node[state, label=0:${\theta_{2}=\theta\drlst\cdot [v]_{f}^{E}}$] (b2) at ([yshift=-1cm]b1) {};

            \draw (b1) -- (b2);
            \node (b3) at ([yshift=-0.7cm]b2) {$\bl(\theta_{2},1)=$};
            \node[rectangle split, rectangle split parts=3, rectangle
split horizontal=true, draw=gray!40] (b4) at ([yshift=-0.5cm]b3)
{$P_{1}\bv{\theta,1}$\nodepart{two}$
\bv{\theta_{1},1}$\nodepart{three}$P_{2}\bv{\theta,2}$};
            \draw[dotted] (b2) -- (b3);
        \end{tikzpicture}
    \end{center}
    \caption{A dependency tree (left) and its shortening (right)}
    \label{fig:antTreeShort}
\end{figure}
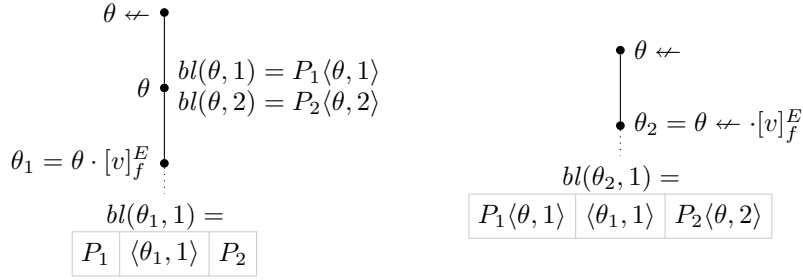
\begin{lemma}
  \label{lem:antTreeShortInv}
  If $(T,\val)$ is complete for a data word $u$ and $T'$ is a
shortening of $T$, then $(T',\val)$ is also complete for $u$.
\end{lemma}

Note that the valuation $\val$ need not be changed to maintain
completeness of $(T',\val)$. Hence, any new middle block variable
will continue to store some middle block as before. Shortening will reduce
the lengths of paths in the tree; still the resulting tree may not be reduced,
since some node $\theta$ may have a block description $\bl(\theta,i)$ that is
too long and/or contains variables not in $X_{\theta}$. Next we explain how to
resolve this.

In a block description $\bl(\theta,i)$, a non-parent block is any infix
$\bl(\theta,i)[j,k]$ such that 1)$j=1$ or the $(j-1)\tsc{th}$ element
of $\bl(\theta,i)$ is a parent reference, 2)$k=|\bl(\theta,i)|$ or the
$(k+1)$\tsc{th} element of $\bl(\theta,i)$ is a parent reference and
3) for every $k'\in [j,k]$, the $k'$\tsc{th} element of
$\bl(\theta,i)$ is not a parent reference. Intuitively, a non-parent
block of $\bl(\theta,i)$ is a maximal infix consisting of elements of
$\hat{X}$ only.
\begin{definition}
    \label{def:antTreeTrim}
    Suppose $T$ is a dependency tree and $\val$ is a valuation for
    $X$. The trimming of $T$ is obtained from $T$ by performing the
    following for every node $\theta$: enumerate the set $\set{z \mid z
    \text{ is a non-parent block in }\bl(\theta,i), 1 \le i \le B}$ as $z_{1},z_{2}, \ldots,
    z_{r}$, choosing the order arbitrarily.
    If $\bl(\theta,i)$ for some $i$ contains $z_{j}$ for some $j$, replace $z_{j}$
    by $\bv{\theta,j}$. Perform such replacements for all $i$ and $j$.
	The trimming $\val'$ of $\val$ is obtained from $\val$ by setting
	$\val'(\bv{\theta,j})=\val(z_{j})$ for all $j$ and
	$\val'(\bv{\theta',k})=\epsilon$ for all $\bv{\theta',k}$ occurring in
	any $z_{j}$. For elements of $\hat{X}$ that neither occur in any $z_{j}$
	nor replace any $z_{j}$, $\val$ and $\val'$ coincide.
\end{definition}
For example, 
$\bl(\theta_{2},1)=P_{1}
\bv{\theta,1}\bv{\theta_{1},1}P_{2}\bv{\theta,2}$ in
Figure~\ref{fig:antTreeShort} is replaced by $P_{1}
\bv{\theta_{2},1}P_{2}\bv{\theta_{2},2}$. In the new valuation, we
have
$\val'(\bv{\theta_{2},1})=\val(\bv{\theta,1})\cdot
\val(\bv{\theta_{1},1})$, $\val'(\bv{\theta_{2},2})=\val(\theta,2)$
and
$\val'(\bv{\theta,1})=\val'(\bv{\theta_{1},1})=\val'(\bv{\theta,2})=\epsilon$.
The following result follows directly from definitions.
\begin{proposition}
    \label{prop:antTreeTrimInv}
    If $(T,\val)$ is complete for a data word $u$, then so is the
    trimming $(T',\val')$.
\end{proposition}

States of the SSRT we construct will have reduced dependency trees.
The following result is helpful in defining the SSRT transitions,
where we have to say how to obtain a new tree from an old one.
\begin{lemma}
    \label{lem:antTreeExtShortTrimRed}
    Suppose $T$ is a reduced dependency tree or $T_\bot$, $T_{1}$ is the
    $(\sigma,\eta)$ extension of $T$ for some $(\sigma,\eta) \in
    \Sigma \times \set{\delta_0,\delta_1, \ldots}$, $T_{2}$ is obtained from $T_{1}$ by shortening
    it as much as possible and $T_{3}$ is the trimming of
    $T_{2}$. Then $T_{3}$ is a reduced dependency tree.
\end{lemma}

We will now extend the SSRT constructed in Construction~\ref{const:SSRTIflVal}
to transform input data words to output data words with origin
information. For any data word with origin information $w$, let
$\dropOrig(w)$ be the data word obtained from $w$ by discarding the third
component in every triple.
\begin{construction}
    \label{const:SSRTTransd}
    Suppose $f$ is a transduction satisfying all the conditions in
    Theorem~\ref{thm:mainResult}. Let $I$ be the maximum number of
    $f$-influencing values in any data word and let $B$ be the maximum
    number of blocks in any factored output of the form $f(u | \mid
\underline{v})$ or $f(\underline{u} \mid v \mid \underline{w})$. Consider a SSRT with set
    of registers $R=\set{R_{1}, \ldots, R_{I}}$ and data word
    variables $X=\set{\bv{\theta,i} \mid \theta \in ( (\Sigma \times
    D)^{*}/\requiv)^{*}, |\theta|\le |(\Sigma \times
    D)^{*}/\requiv|+1, i \in [1,B^{2}+B]}$. Every state is a triple
    $([u]_{f},\ptr,T)$ where $u$ is some data word, $T$ is a
	reduced dependency tree or $T_\bot$
	such that $\pref$ labels every leaf in $T$ with $[u]_f$ and
    $\ptr:[1,|\ifl_f(u)|]\to R$ is a pointer function. The initial
    state is $([\epsilon]_{f},\ptr_{\bot},T_{\bot})$. Let $\delta_0
\notin \set{\delta_{|\ifl_f(u)|},\ldots,\delta_{1}}$ be an arbitrary data
value. For every $T$ and for every transition $(([u]_f,\ptr), \sigma,\phi,
([u_q\cdot (\sigma,\delta_i)]_f, \ptr'), R',\update_\bot)$ given in
Construction~\ref{const:SSRTIflVal}, we will have
the following transition: $(([u]_f,\ptr, T), \sigma,\phi,
([u_q\cdot (\sigma,\delta_i)]_f, \ptr', T'), R',\update)$.
    Let $T_{1}$ be the $(\sigma,\delta_i)$ extension of $T$ and let
    $T_{2}$ be obtained from $T_{1}$ by shortening it as much as
    possible. $T'$ is defined to be the trimming of $T_{2}$.
    We define the update function $\update$ using an intermediate
function $\update_1$ and an arbitrary data word $u' \in [u]_f$.
    For every data word variable
    $\bv{\theta,i}$ that is not a new middle block variable in $T_1$, set
    $\update_{1}(\bv{\theta,i})=\bv{\theta,i}$. For every new middle block
    variable $\bv{\theta,k}$, say $\theta=\theta_{v} \cdot
    [v]_{f}^{E}$. Set $\update_{1}(\bv{\theta,k})=\dropOrig(z)$, where $z$ is obtained
    from the $k$\tsc{th} middle block of $f(\underline{E(u')(u')} \mid
    (\sigma,\delta_i) \mid \underline{\pi(v)})$ by replacing every occurrence
    of $\delta_{j}$ by $\ptr(j)$ for all $j \in [1,|\ifl_f(u)|]$ and replacing every occurrence
    of $\delta_0$ by $\curr$.  Here, $\pi$ is a permutation tracking
influencing values in $E(u')(u') \cdot (\sigma,\delta_i)$ as given in
    Definition~\ref{def:middleBlocks}. Next we define the function $\update$.
    While trimming $T_{2}$, suppose a non-parent block $z_{j}$ in a
    node $\theta$ was replaced by a data word variable
    $\bv{\theta,j}$. Define $\update(\bv{\theta,j})=\update_{1}(z_{j})$. For every
    data word variable $\bv{\theta_{1},k}$ occurring in $z_{j}$, define
    $\update(\bv{\theta_{1},k})=\epsilon$. For all other data word variables
    $\bv{\theta_{2},k'}$, define
    $\update(\bv{\theta_{2},k'})=\update_{1}(\bv{\theta_{2},k'})$.
    The output function $O$ is defined as follows: for every state
    $([u]_{f},\ptr,T)$, $O(([u]_{f},\ptr,T))=\ur(\theta,\bl(\theta,1))
    \cdot \cdots \cdot \ur(\theta,\bl(\theta,B))$ where $\theta$ is the leaf of
    $T$ such that $\theta=\theta'\cdot [\epsilon]_{f}^{E}$ ends in the
    equivalence class $[\epsilon]_{f}^{E}$.
\end{construction}
Lemma~\ref{lem:antTreeExtShortTrimRed} implies that if $T$ is $T_\bot$
or a reduced dependency tree, then so is $T'$. It is routine to
verify that this SSRT is deterministic and copyless.

\begin{lemma}
    \label{lem:SSRTTransdInv}
    Let the SSRT constructed in Construction~\ref{const:SSRTTransd} be $S$.
    After reading a data word $u$, $S$ reaches the configuration
    $( ([u]_{f},\ptr,T),\val, |u|)$ such that $\ptr(i)$ is the
    $i$\tsc{th} $f$-influencing value in $u$ and $(T,\val)$ is
    complete for $u$.
\end{lemma}

\begin{proof}[Proof of reverse direction of
    Theorem~\ref{thm:mainResult}]
    Let $f$ be a transduction that satisfies all the properties
    stated in Theorem~\ref{thm:mainResult}.
    We infer from Lemma~\ref{lem:SSRTTransdInv} that
    the SSRT $S$ constructed in Construction~\ref{const:SSRTTransd}
satisfies the following property.
    After reading a data word $u$, $S$ reaches the configuration
    $( ([u]_{f},\ptr,T),\val, |u|)$ such that $\ptr(i)$ is the
    $i$\tsc{th} $f$-influencing value in $u$ and $(T,\val)$ is
    complete for $u$.
We define the output function of the SSRT such that
$\semantics{S}(u)=\val(\ur(\theta,\bl(\theta,1)) \cdot
    \cdots \cdot \ur(\theta,\bl(\theta,B)))$, where $\theta=\theta'
    \cdot [\epsilon]_{f}^{E}$ is the leaf of $T$ ending with
    $[\epsilon]_{f}^{E}$. Since $(T,\val)$ is complete for $u$, we
    infer that
    $\val(\ur(\theta,\bl(\theta,1)) \cdot \cdots \cdot
    \ur(\theta,\bl(\theta,B)))$ is the concatenation of the left
    blocks of $f(u \mid \underline{E(u)^{-1}(\epsilon)})=f(u)$. Hence,
    the SSRT $S$ implements the transduction $f$.
\end{proof}

\section{Properties of Transductions Implemented by SSRTs}
\label{sec:propTransdSSRT}
In this section, we prove the forward direction of our main result
(Theorem~\ref{thm:mainResult}).

For a valuation $\val$ and permutation
$\pi$, we denote by $\pi(\val)$ the valuation that assigns
$\pi(\val(r))$ to every register $r$ and $\pi(\val(x))$ to every data
word variable $x$.
 The following two results easily follow from
definitions.

\begin{proposition}
    \label{prop:SSRTInvPerm}
    Suppose a SSRT $S$ reaches a configuration $(q,\val,n)$ after
    reading a data word $u$. If $\pi$ is any permutation, then
    $S$ reaches the configuration $(q,\pi(\val),n)$ after reading
    $\pi(u)$.
\end{proposition}

\begin{proposition}
    \label{prop:SSRTTransdInvPerm}
    If a SSRT $S$ implements a transduction $f$, then $f$ is invariant
    under permutations and is without data peeking.
\end{proposition}

After a SSRT reads a data word, data values that are not stored in any
of the registers will not influence the rest of the operations.

\begin{lemma}
    \label{lem:ifValsInRegisters}
    Suppose a SSRT $S$ implements the transduction $f$. Any data value $d$ that
	is $f$-influencing in some data word $u$ will be stored in one of the registers
	of $S$ after reading $u$.
\end{lemma}
 Now we identify data words
after reading which, a SSRT reaches similar cofigurations.
\begin{definition}
    \label{def:sequiv}
    For a SSRT $S$, we define a binary relation $\sequiv$ on data words as
    follows: $u_{1} \sequiv u_{2}$ if they satisfy the following
    conditions. Suppose $f$ is the transduction implemented by $S$, which
    reaches the configuration $(q_{1},\val_{1}, |u_1|)$ after reading $u_{1}$ and
    reaches $(q_{2},\val_{2}, |u_2|)$ after reading $u_{2}$.
    \begin{enumerate}
        \item $q_{1}=q_{2}$,
        \item for any two registers $r_{1},r_{2}$,
            we have $\val_{1}(r_{1})=\val_{1}(r_{2})$ iff
            $\val_{2}(r_{1})=\val_{2}(r_{2})$,
        \item\label{sequiv:ifl} for any register $r$, $\val_{1}(r)$ is the $i$\tsc{th}
            $f$-suffix influencing value (resp.~$f$-prefix influencing
            value) in $u_{1}$ iff $\val_{2}(r)$ is the $i$\tsc{th}
            $f$-suffix influencing value (resp.~$f$-prefix influencing
            value) in $u_{2}$,
        \item\label{sequiv:varEmpty} for any data word variable $x$, we have
            $\val_{1}(x)=\epsilon$ iff $\val_{2}(x)=\epsilon$ and
        \item\label{sequiv:varArrEquiv} for any two subsets
            $X_{1},X_{2} \subseteq X$ and any arrangements
            $\chi_{1},\chi_{2}$ of $X_{1},X_{2}$ respectively,
            $\val_{1}(\chi_{1})=\val_{1}(\chi_{2})$ iff
            $\val_{2}(\chi_{1})=\val_{2}(\chi_{2})$.
    \end{enumerate}
\end{definition}
An arrangement of a finite set $X_{1}$ is a sequence in $X_{1}^{*}$ in
which every element of $X_{1}$ occurs exactly once. It is routine to
verify that $\sequiv$ is an equivalence relation of finite index.

Suppose a SSRT $S$ reads a data word $u$, reaches the configuration
$(q,\val, |u|)$ and from there, continues to read a data word $v$. For
some data word variable $x \in X$, if $\val(x)$ is some data word $z$,
then none of the transitions executed while reading $v$ will split $z$
--- it might be appended or prepended with other data words and may be
moved to other variables but never split. Suppose $X=\set{x_{1},
\ldots, x_{m}}$. The transitions executed while reading $v$ can
arrange $\val(x_{1}), \ldots, \val(x_{m})$ in various ways, possibly
inserting other data words (whose origin is in $v$) in between. Hence,
any left block of $\semantics{S}(u \mid \underline{v})$ is
$\val(\chi)$, where $\chi$ is some arrangement of some subset
$X'\subseteq X$.  The following result is shown by proving that
$\sequiv$ refines $\fequiv$. The most difficult part of this proof is
to prove that if $u_1 \sequiv u_2$, then there exists a permutation
$\pi$ such that for all data words $u,v_1,v_2$, $f(u_1 \cdot u \mid
\underline{v_1}) = f(u_1 \cdot u \mid \underline{v_2})$ iff
$f(\pi(u_2) \cdot u \mid \underline{v_1}) = f(\pi(u_2) \cdot u \mid
\underline{v_2})$. The idea is to show that if $f(u_1 \cdot u \mid
\underline{v_1}) \ne f(u_1 \cdot u \mid \underline{v_2})$, then for
some arrangements $\chi_1,\chi_2$ of some subsets $X_1,X_2 \subseteq
X$, $\val_1(\chi_1) \ne \val_1(\chi_2)$ ($\val_1$ (resp.~$\val_2$) is
the valuation reached by $S$ after reading $u_1$ (resp.~$u_2$)).
Since $u_1 \sequiv u_2$, this implies that $\val_2(\chi_1) \ne
\val_2(\chi_2)$, which in turn implies that $f(\pi(u_2) \cdot u \mid
\underline{v_1}) \ne f(\pi(u_2) \cdot u \mid \underline{v_2})$.
\begin{lemma}
    \label{lem:fEquivFiniteIndex}
    If a SSRT $S$ implements a transduction $f$, then $\fequiv$ has
    finite index.
\end{lemma}

\begin{proof}[Proof of forward direction of Theorem~\ref{thm:mainResult}]
Suppose $f$ is the transduction implemented by a SSRT $S$.
Lemma~\ref{lem:fEquivFiniteIndex} implies that $\fequiv$ has finite
index.
Proposition~\ref{prop:SSRTTransdInvPerm} implies that $f$ is invariant under
permutations and is without data peeking. The output of $S$ on any input
is the concatenation of the data words stored in some variables in $S$
and constantly many symbols coming from the output finction of $S$.
The contents of data word variables are generated by transitions when
reading input symbols and each transition can write only constantly
many symbols into any data word variable after reading one input
symbol. After some content is written into a data word variable, it is
never duplicated into multiple copies since the transitions of $S$ are
copyless. Hence, any input position can be the origin of only
constantly many output positions. Hence, $f$ has linear blow up.
\end{proof}

\section{Future Work}
One direction to explore is whether there is a notion of minimal
canonical SSRT and if a given SSRT can be reduced to an equivalent
minimal one.  Adding a linear order on the data domain, logical
characterization of SSRTs and studying two way transducer models
with data are some more interesting studies.

Using nominal automata, techniques for
finite alphabets can often be elegantly carried over to infinite
alphabets, as done in \cite{MSSKS2017}, for example. It would be
interesting to see if the same can be done for streaming transducers over
infinite alphabets. Using concepts from the theory of  nominal
automata, recent work \cite{BS2020} has shown that an atom extension
of streaming string transducers is equivalent to a certain class of two way
transducers.  This model of transducers is a restriction of SSRTs and
is robust like regular languages over finite alphabets. It would also
be interesting to see how can techniques in this extended abstract
be simplified to work on the transducer model presented in
\cite{BS2020}.

\bibliography{biblio}
%
%
\appendix
\onecolumn
\section{Fundamental Properties of Transductions}
\label{app:Fundamental}
The following result says that if a transduction is invariant under
permutations, then so are all its factored outputs.
\begin{lemma}
    \label{lem:factoredOutputInvariant}
    Suppose $f$ is a transduction that is invariant under
    permutations, $u,v,w$ are data words, $\pi$ is any permutation and
    $z$ is any integer.
    Then $\pi( f_{z}( \underline{u} \mid v)) = f_{z}(\underline{\pi(u)} \mid
    \pi(v))$, $\pi( f_{z}(u \mid \underline{v})) = f_{z}( \pi(u) \mid
    \underline{\pi(v)})$ and $\pi(f_{z}(\underline{u} \mid v \mid
    \underline{w})) = f_{z}(\underline{\pi(u)} \mid \pi(v) \mid
    \underline{\pi(w)})$.
\end{lemma}
\begin{proof}
    From the invariance of $f$ under permutations, we have $f(\pi(u)
    \cdot \pi(v)) = \pi( f( u \cdot v))$. Adding $z$ to every triple
    on both sides, we get
    \begin{align*}
        f_{z}(\pi(u) \cdot \pi(v)) &= \pi( f_{z}( u \cdot v)) \enspace .
    \end{align*}
    For every $i \in [1,|f_{z}(\pi(u) \cdot \pi(v))|]$, we perform the
    following on the LHS of the above equation: let $(\gamma,d,o)$ be
    the $i$\textsuperscript{th} triple in the LHS; if $o-z \in [1,|u|]$,
    replace the triple by $(*,*,\lef)$. After performing this change
    for every $i$, merge consecutive occurrences of $(*,*,\lef)$ into
    a single triple $(*,*,\lef)$. At the end, we get
    $f_{z}(\underline{\pi(u)} \mid \pi(v))$.

    Now perform exactly the same operations not on the RHS $\pi(f_z (u
    \cdot v))$, but on $f_{z}(u \cdot v)$. The $i$\textsuperscript{th}
    triple will be $(\gamma,\pi^{-1}(d),o)$ and it changes to
    $(*,*,\lef)$ iff the $i$\textsuperscript{th} triple $(\gamma,d,o)$
    in the LHS changed to $(*,*,\lef)$. Now, if we merge consecutive
    occurrences of $(*,*,\lef)$ into a single triple $(*,*,\lef)$, we
    get $f_{z}(\underline{u} \mid v)$. If we now apply the permutation
    $\pi$ to this, we get $\pi(f_{z}(\underline{u} \mid v))$, but we also
    get exactly the same sequence of triples we got from LHS after the
    changes, which is $f_{z}(\underline{\pi(u)} \mid \pi(v))$. Hence,
    $f_{z}(\underline{\pi(u)} \mid \pi(v)) = \pi(f_{z}(\underline{u} \mid
    v))$. The proofs of the other two equalities are similar.
\end{proof}

The following result says that the influencing values of a data word
are affected by a permutation as expected.
\begin{lemma}
    \label{lem:iflInvUnderPerm}
    If $f$ is a transduction that is invariant under permutations
    and $u$ is a data word, then for any permutation
    $\pi$, $\aifl_{f}(\pi(u)) = \pi(\aifl_{f}(u))$.
\end{lemma}
\begin{proof}
    It is sufficient to prove that for any position $j$ of $u$, the
data value in the $j$\textsuperscript{th} position of $u$ is a
$f$-memorable value in $u$ iff the data value in the
$j$\textsuperscript{th} position of $\pi(u)$ is a $f$-memorable
value in $\pi(u)$ and similarly for $f$-vulnerable
values. Indeed, suppose $d$ is the data value in the
$j$\textsuperscript{th} position of $u$ and it is a $f$-memorable
value in $u$. By Definition~\ref{def:fMemorable},
there exists a data word $v$ and a data value $d'$ that is a safe
replacement for $d$ in $u$ such that $f(\underline{u[d/d']}
\mid v) \ne f(\underline{u} \mid v)$. The data value at
$j$\textsuperscript{th} position of $\pi(u)$ is $\pi(d)$ and the word
$\pi(v)$ and the data value $\pi(d')$ witnesses that $\pi(d)$ is a
$f$-memorable in $\pi(u)$. Indeed, if $f(\underline{u[d/d']}
\mid v) \ne f(\underline{u} \mid v)$, then
Lemma~\ref{lem:factoredOutputInvariant} implies that $f(\underline{\pi(u)[\pi(d)/\pi(d')]} \mid
\pi(v)) \ne f(\underline{\pi(u)} \mid \pi(v))$. The converse direction of
the proof is symmetric, using the permutation $\pi^{-1}$.

Suppose $d$ is the data value in the $j$\textsuperscript{th} position
of $u$ and it is a $f$-vulnerable value in $u$. By
Definition~\ref{def:fMemorable}, there exist data words $u',v$
and a data value $d'$ that is a safe replacement for $d$ in $u\cdot u'
\cdot v$ such that $d$ doesn't occur in $u'$ and $f(u \cdot u' \mid
\underline{v}) \ne f(u \cdot u' \mid \underline{v[d/d']})$. The data
value at $j$\textsuperscript{th} position of $\pi(u)$ is $\pi(d)$ and
the words $\pi(u'),\pi(v)$ and the data value $\pi(d')$ witnesses
that $\pi(d)$ is a $f$-vulnerable in $\pi(u)$. Indeed, since
$f(u \cdot u' \mid \underline{v}) \ne f(u \cdot u' \mid
\underline{v[d/d']})$, Lemma~\ref{lem:factoredOutputInvariant} implies
that $f(\pi(u) \cdot \pi(u') \mid \underline{\pi(v)}) \ne f(\pi(u)
\cdot \pi(u') \mid \underline{\pi(v)[\pi(d)/\pi(d')]})$. The converse
direction of the proof is symmetric, using the permutation $\pi^{-1}$.
\end{proof}

A data value that does not occur in a data word can not influence how
it is transformed.
\begin{lemma}
    \label{lem:prefIflOccurs}
    Suppose $f$ is a transduction that is invariant under permutations
    and without data peeking and a data value $d$ is $f$-vulnerable
    in a data word $u$. Then $d$ occurs in $u$.
\end{lemma}
\begin{proof}
    Suppose $d$ does not occur in $u$. We will prove that $d$ is not
    $f$-vulnerable in $u$. Let $u',v$ be any data words such
    that $d$ does not occur in $u'$. Suppose $d'$ is a safe
    replacement for $d$ in $u\cdot u' \cdot v$. Let $\pi$ be the
    permutation that interchanges $d$ and $d'$ and does not change any
    other value. Neither $d$ nor $d'$ occurs in $u \cdot u'$, so
    $\pi(u \cdot u')=u \cdot u'$. The data value $d'$ does not occur
    in $v$, so $\pi(v)=v[d/d']$. Since $f$ is without data peeking, only
    data values in occurring in $u\cdot u'$ occur $f(u\cdot u' \mid
    \underline{v})$, so neither $d$ nor $d'$ occur in $f(u\cdot u' \mid
    \underline{v})$, so $\pi(f(u\cdot u' \mid
    \underline{v}))=f(u\cdot u' \mid
    \underline{v})$. Since $f$ is invariant under
    permutations, we infer from
    Lemma~\ref{lem:factoredOutputInvariant} that $\pi(f(u\cdot u' \mid
    \underline{v}))=f(\pi(u \cdot u') \mid \underline{\pi(v)})$. This
    implies that $f(u\cdot u' \mid
    \underline{v})=f(u\cdot u' \mid
    \underline{v[d/d']})$.  Hence, $d$ is not $f$-vulnerable
    in $u$.
\end{proof}

Data values in a prefix can be permuted without changing the way it
affects suffixes, as long as we don't change the influencing
values.
\begin{lemma}
    \label{lem:nonSufIflPermutable} Suppose $f$ is a transduction that
    is invariant under permutations, $u,v$ are data words and $\pi$
    is any permutation that is identity on the set of data values that
    are $f$-influencing in $u$. Then $f(\underline{ \pi(u)} \mid v) =
    f(\underline{u} \mid v)$ and $\aifl_{f}(u) = \aifl_{f}(\pi(u))$.
\end{lemma}
\begin{proof}
    Let $\set{d_{1}, \ldots, d_{n}}$ be the set of all data values
    occurring in $u$ that are not $f$-influencing in $u$. Let $d_{1}',
    \ldots, d_{n}'$ be safe replacements for $d_{1}, \ldots, d_{n}$
    respectively in $u$, such that $\set{d_{1}', \ldots, d_{n}'} \cap
    (\set{d_{1}, \ldots, d_{n}} \cup \set{\pi(d_{1}), \ldots,
    \pi(d_{n})}) = \emptyset$. Since $d_{1}$ is not $f$-memorable in
    $u$, we have $f(\underline{u[d_{1}/d_{1}']} \mid v) =
    f(\underline{u} \mid v)$. Since $d_{2}$ is not $f$-influencing in
    $u$, we infer from Lemma~\ref{lem:iflValuesPreserved} that
    $d_{2}$ is not $f$-influencing in $u[d_{1}/d_{1}']$. Hence,
    $f(\underline{u[d_{1}/d_{1}'][d_{2}/d_{2}']} \mid v) =
    f(\underline{u[d_{1}/d_{1}']} \mid v) = f(\underline{u}\mid v)$.
    Also from Lemma~\ref{lem:iflValuesPreserved}, we infer that
    $d_{1}'$ is not $f$-influencing in $u[d_{1}/d_{1}']$ (put $e=d_{1}'$ in
    Lemma~\ref{lem:iflValuesPreserved} to see this).  Similarly,
    neither $d_{1}'$ nor $d_{2}'$ are $f$-influencing in
    $u[d_{1}/d_{1}'][d_{2}/d_{2}']$. On the other hand,  we infer from
Lemma~\ref{lem:iflValuesPreserved} that all the data values that
    are $f$-memorable (resp.~$f$-vulnerable) in $u$ are also $f$-memorable (resp.~$f$-vulnerable) in
    $u[d_{1}/d_{1}'][d_{2}/d_{2}']$. This reasoning can be routinely
    extended to an induction on $i$ to infer that
    $f(\underline{u[d_{1}/d_{1}', \ldots, d_{i}/d_{i}']} \mid v) =
    f(\underline{u} \mid v)$ and $d_{1}',\ldots, d_{i}'$ are not
    $f$-influencing in $u[d_{1}/d_{1}', \ldots, d_{i}/d_{i}']$. Hence,
    $f(\underline{u[d_{1}/d_{1}', \ldots, d_{n}/d_{n}']} \mid v) =
    f(\underline{u} \mid v)$. In addition, all the data values that are $f$-memorable (resp.~$f$-vulnerable) in $u$ are also
    $f$-memorable (resp.~$f$-vulnerable) in $u[d_{1}/d_{1}', \ldots, d_{n}/d_{n}']$.

    Now we prove that $\pi(d_{1}), \ldots, \pi(d_{n})$ are safe
    replacements for $d_{1}', \ldots, d_{n}'$ in $u[d_{1}/d_{1}', \ldots,
    d_{n}/d_{n}']$. We know that $\data(u[d_{1}/d_{1}', \ldots,
    d_{n}/d_{n}'], *)=\set{d_{1}', \ldots, d_{n}'} \cup \set{d \mid d
    \text{ is }f\text{-influencing in }u}$. We have $\set{\pi(d_{1}), \ldots, \pi(d_{n})} \cap
    \set{d_{1}',\ldots, d_{n}'} = \emptyset$ by choice. Since $\pi$ is
    identity on $\set{d \mid d
    \text{ is }f\text{-influencing in }u}$ and $d_{1}, \ldots, d_{n}$ are not
    $f$-influencing in $u$, we have $\set{\pi(d_{1}), \ldots, \pi(d_{n})} \cap \set{d \mid d
    \text{ is }f\text{-influencing in }u} = \emptyset$. This proves that
    $\pi(d_{1}), \ldots, \pi(d_{n})$ are safe replacements for
    $d_{1}', \ldots, d_{n}'$ in $u[d_{1}/d_{1}', \ldots,
    d_{n}/d_{n}']$.

    As we did in the first 
    paragraph of this proof, we conclude that
    $f(\underline{u[d_{1}/d_{1}', \ldots,
    d_{n}/d_{n}'][d_{1}'/\pi(d_{1}), \ldots, d_{n}'/\pi(d_{n})]} \mid
    v) = f(\underline{u[d_{1}/d_{1}', \ldots,
    d_{n}/d_{n}']} = f(\underline{u} \mid v)$. Since $u[d_{1}/d_{1}', \ldots,
    d_{n}/d_{n}'][d_{1}'/\pi(d_{1}), \ldots, d_{n}'/\pi(d_{n})] =
    u[d_{1}/\pi(d_{1}), \ldots, d_{n}/\pi(d_{n})] = \pi(u)$, we infer
    that $f(\underline{\pi(u)} \mid v) = f(\underline{u} \mid v)$. In
    addition, $\pi(d_{1}), \ldots, \pi(d_{n})$ are not $f$-influencing in
    $\pi(u)$ and all the values that are $f$-memorable (resp.~$f$-vulnerable) in $u$ are also
    $f$-memorable (resp.~$f$-vulnerable) in $\pi(u)$.
    Hence, $\aifl(\pi(u)) = \aifl(u)$.
\end{proof}

Data values in a suffix can be permuted without changing the way it
affects prefixes, as long as we don't change the prefix influencing
values.
\begin{lemma}
    \label{lem:nonPreIflPermutable}
    Suppose $f$ is a transduction that is invariant under permutations
    and without data peeking, $u,v$ are data data words and $\pi$ is any
    permutation that is identity on the set of data values that are
    $f$-vulnerable in $u$. Then $f(u \mid \underline{\pi(v)})
    = f(u \mid \underline{v})$.
\end{lemma}
\begin{proof}
  Let $\set{d_{1}, \ldots, d_{n}}$ be the set of all data values
  occurring in $v$ that are not $f$-vulnerable in $u$. Let
  $d_{1}', \ldots, d_{n}'$ be safe replacements for $d_{1}, \ldots,
  d_{n}$ respectively in $u \cdot v$, such that $\set{d_{1}', \ldots,
  d_{n}'} \cap (\set{d_{1}, \ldots, d_{n}} \cup \set{\pi(d_{1}),
  \ldots, \pi(d_{n})}) = \emptyset$. Since $d_{1}$ is not $f$-vulnerable
  in $u$, we have $f(u \mid \underline{v[d_{1}/d_{1}']}) =
  f(u \mid \underline{v})$. Since $d_{2}$ is not $f$-vulnerable
  in $u$, we have $f(u \mid
  \underline{v[d_{1}/d_{1}'][d_{2}/d_{2}']}) = f(u \mid
  \underline{v[d_{1}/d_{1}']}) = f(u \mid \underline{v})$. The same
  reasoning can be used in an induction to conclude that $f(u \mid
  \underline{v[d_{1}/d_{1}', d_{2}/d_{2}', \ldots, d_{n}/d_{n}']}) =
  f(u \mid \underline{v})$.
  
  Now we will prove that $\pi(d_{1}),
  \ldots, \pi(d_{n})$ are safe replacements for $d_{1}',\ldots,
  d_{n}'$ respectively in $v[d_{1}/d_{1}', d_{2}/d_{2}', \ldots,
  d_{n}/d_{n}']$. We have $\data(v[d_{1}/d_{1}', \ldots,
  d_{n}/d_{n}'], *)=\set{d_{1}', \ldots, d_{n}'} \cup \set{d \mid d
  \text{ is }f\text{-vulnerable in }u}$. We have
  $\set{\pi(d_{1}), \ldots, \pi(d_{n})} \cap \set{d_{1}',\ldots,
  d_{n}'} = \emptyset$ by choice. Since $\pi$ is identity on $\set{d
  \mid d \text{ is }f\text{-vulnerable in }u}$ and $d_{1},
  \ldots, d_{n}$ are not $f$-vulnerable in $u$, we have
  $\set{\pi(d_{1}), \ldots, \pi(d_{n})} \cap \set{d \mid d
  \text{ is }f\text{-vulnerable in }u} = \emptyset$.  This
  proves that $\pi(d_{1}), \ldots, \pi(d_{n})$ are safe replacements
  for $d_{1}', \ldots, d_{n}'$ in $v[d_{1}/d_{1}', \ldots,
  d_{n}/d_{n}']$.
  
  Now we claim that $f(u \mid
  \underline{v[d_{1}/d_{1}', d_{2}/d_{2}', \ldots,
  d_{n}/d_{n}'][d_{1}'/\pi(d_{1})]}) = f(u \mid
  \underline{v[d_{1}/d_{1}', d_{2}/d_{2}', \ldots, d_{n}/d_{n}']})$.
  Suppose not, i.e., $f(u \mid
  \underline{v[d_{1}/d_{1}', d_{2}/d_{2}', \ldots,
  d_{n}/d_{n}'][d_{1}'/\pi(d_{1})]}) \ne f(u \mid
  \underline{v[d_{1}/d_{1}', d_{2}/d_{2}', \ldots, d_{n}/d_{n}']})$.
  This can be written equivalently as $f(u \mid
  \underline{v[d_{1}/d_{1}', d_{2}/d_{2}', \ldots,
  d_{n}/d_{n}'][d_{1}'/\pi(d_{1})]}) \ne f(u \mid
  \underline{v[d_{1}/d_{1}', d_{2}/d_{2}', \ldots,
  d_{n}/d_{n}'][d_{1}'/\pi(d_{1})][\pi(d_{1})/d_{1}']})$. Then we
  infer from Definition~\ref{def:fMemorable} that $\pi(d_{1})$ is
  $f$-vulnerable in $u$, which contradicts the hypothesis that
  $\pi$ is identity on all values that are $f$-vulnerable in
  $u$. Hence, $f(u \mid \underline{v[d_{1}/d_{1}', d_{2}/d_{2}',
  \ldots, d_{n}/d_{n}'][d_{1}'/\pi(d_{1})]}) = f(u \mid
  \underline{v[d_{1}/d_{1}', d_{2}/d_{2}', \ldots, d_{n}/d_{n}']})$.
  
  Similar
  reasoning can then be used to infer that $f(u \mid
  \underline{v[d_{1}/d_{1}', d_{2}/d_{2}', \ldots,
  d_{n}/d_{n}'][d_{1}'/\pi(d_{1}), \ldots, d_{n}'/\pi(d_{n})]}) = f(u
  \mid \underline{v[d_{1}/d_{1}', d_{2}/d_{2}', \ldots,
  d_{n}/d_{n}']}) = f(u \mid \underline{v})$. Hence, $f(u \mid
  \underline{\pi(v)}) = f(u \mid \underline{v})$.
\end{proof}

If two factored outputs are equal, factoring out the same word from the
same positions of the inputs will not destroy the equality.
\begin{lemma}
    \label{lem:rightCongruenceHelper}
    Suppose $f$ is a transduction, $u, u_{1}, u_{2},v, v_{1}, v_{2}$
    are data words, $\sigma \in \Sigma$, $d$ is a data value and
    $z=|u_{1}|-|u_{2}|$.
    \begin{enumerate}
        \item\label{rch:rightToAbstract} If $f(\underline{u_{1}} \mid u \cdot v) =
            f_{z}(\underline{u_{2}} \mid u \cdot v)$, then
            $f(\underline{u_{1}} \mid u \mid \underline{v}) =
            f_{z}(\underline{u_{2}} \mid u \mid \underline{v})$.
        \item\label{rch:middleToLeft} If $f(\underline{u_{1}} \mid u \cdot v) =
            f_{z}(\underline{u_{2}} \mid u \cdot v)$, then
            $f(\underline{u_{1} \cdot u} \mid v) =
            f_{z}(\underline{u_{2} \cdot u} \mid v)$.
        \item\label{rch:middleToRight} If $f(u \cdot v \mid \underline{v_{1}}) = f(u \cdot v
            \mid \underline{v_{2}})$, then $f(u \mid \underline{v
            \cdot v_{1}}) = f(u \mid \underline{v \cdot v_{2}})$.
        \item\label{rch:leftToAbstract} If $f(u \cdot v \mid \underline{v_{1}}) = f(u \cdot v
            \mid \underline{v_{2}})$, then $f(\underline{u} \mid v
            \mid \underline{v_{1}}) = f(\underline{u} \mid v \mid
            \underline{v_{2}})$.
    \end{enumerate}
\end{lemma}
\begin{proof}
    We prove the first statement. Others are similar. We have the
    following equality from the hypothesis.
    \begin{align*}
      f(\underline{u_{1}} \mid u \cdot v) &=
    f_{z}(\underline{u_{2}} \mid u \cdot v)  
    \end{align*}
    For every $i \in [1,|f(\underline{u_{1}} \mid u \cdot v)|]$,
    we perform the following on the LHS of the above equation: let
    $(\gamma,d,o)$ be the $i$\textsuperscript{th} triple in the LHS;
    if $o > |u_{1}|+|u|$, replace the triple by $(*,*,\ri)$
    (the origin of such a triple is in $v$).
    Otherwise, don't change the triple. After
    performing this change for every $i$, merge consecutive
    occurrences of $(*,*,\ri)$ into a single triple $(*,*,\ri)$. At
    the end, we get $f(\underline{u_{1}}\mid u \mid \underline{v})$.

    Now perform exactly the same operations on the RHS
    $f_{z}(\underline{u_{2}} \mid u \cdot v)$. The
    $i$\textsuperscript{th} triple $(\gamma,d,o)$ will change to
    $(*,*,\ri)$ (resp.~will not change) iff the
    $i$\textsuperscript{th} triple $(\gamma,d,o)$ in the LHS changed
    to $(*,*,\ri)$ (resp.~did not change). Note that if
    $o>|u_{1}|+|u|$, $o-z>|u_{2}|+|u|$. Hence, the triples that change to
    $(*,*,\ri)$ in
    the RHS are precisely the triples whose origin is in $v$. Now, if we merge
    consecutive occurrences of $(*,*,\ri)$ into a single triple
    $(*,*,\ri)$, we get $f_{z}(\underline{u_{2}} \mid u \mid
    \underline{v})$.  This is also the same sequence of triples we got from LHS
    after the changes, which is $f(\underline{u_{1}} \mid u \mid
    \underline{v})$. Hence, $f(\underline{u_{1}} \mid u \mid
    \underline{v}) =
    f_{z}(\underline{u_{2}} \mid u \mid \underline{v})$.
\end{proof}

\begin{lemma}
    \label{lem:nonSufIflPermutableSecond}
    Suppose $f$ is a transduction that is invariant under
    permutations, $u,v,w$ are data words and $\pi,\pi' \in \Pi$ are
    permutations on the data domain $D$. If $\pi$ and $\pi'$ coincide
    on those data values that are $f$-influencing in $u\cdot v$, then
    $\pi(f(\underline{u} \mid v \mid
    \underline{w}))=f(\underline{\pi(u)} \mid \pi(v) \mid \underline{\pi'(w)})$.
\end{lemma}
\begin{proof}
    Since $\pi$ and $\pi'$ coincide on those data values that are
    $f$-influencing in $u\cdot v$, we infer from
    Lemma~\ref{lem:nonPreIflPermutable} that $f(\pi(u\cdot v) \mid
    \underline{\pi(w)})=f(\pi(u \cdot v) \mid \underline{\pi'(w)})$.
    From point~\ref{rch:leftToAbstract} of
    Lemma~\ref{lem:rightCongruenceHelper}, we conclude that
    $f(\underline{\pi(u)} \mid \pi(v) \mid
    \underline{\pi(w)})=f(\underline{\pi(u)} \mid \pi(v) \mid
    \underline{\pi'(w)})$. We have from
    Lemma~\ref{lem:factoredOutputInvariant} that
    $\pi(f(\underline{u} \mid v \mid \underline{w})) =
    f(\underline{\pi(u)} \mid \pi(v) \mid \underline{\pi(w)})$.
    Combining the last two equalities, we get the result.
\end{proof}

The following result is in some sense the converse of points
\eqref{rch:middleToRight} and \eqref{rch:leftToAbstract} in
Lemma~\ref{lem:rightCongruenceHelper}.
\begin{lemma}
    \label{lem:facOpEqualityFromParts}
    Let $f$ be a transduction and $u,v,w_{1},w_{2}$ be data words. If
    $f(\underline{u}\mid v \mid
    \underline{w_{1}})=f(\underline{u} \mid v \mid
    \underline{w_{2}})$ and $f(u \mid \underline{vw_{1}}) = f(u \mid
    \underline{vw_{2}})$, then $f(uv \mid \underline{w_{1}}) =
    f(uv \mid \underline{w_{2}})$.
\end{lemma}
\begin{proof}
    The number of occurrences of the triple $(*,*,\ri)$ is the same in
    $f(\underline{u} \mid v \mid \underline{w_{1}})$ and $f(uv \mid
    \underline{w_{1}})$. The number of occurrences of the triple
    $(*,*,\ri)$ is the same in $f(\underline{u} \mid v \mid
    \underline{w_{2}})$ and $f(uv \mid \underline{w_{2}})$.
    Suppose $f(uv \mid \underline{w_{1}}) \ne f(uv \mid
    \underline{w_{2}})$. If the number of occurrences of the triple
    $(*,*,\ri)$ are different in $f(uv \mid \underline{w_{1}})$ and
    $f(uv \mid \underline{w_{2}})$, then the number of occurrences of the triple
    $(*,*,\ri)$ are different in $f(\underline{u}\mid v \mid
    \underline{w_{1}})$ and $f(\underline{u}\mid v \mid
    \underline{w_{2}})$ and we are done. So assume that the number of
    occurrences of the triple $(*,*,\ri)$ is the same in $f(uv \mid
    \underline{w_{1}})$ and $f(uv \mid \underline{w_{2}})$. Let
    $i$ be the first position where $f(uv \mid
    \underline{w_{1}})$ and $f(uv \mid \underline{w_{2}})$ differ.

    Case 1: at position $i$, $f(uv \mid \underline{w_{1}})$ contains
    $(*,*,\ri)$ and $f(uv \mid \underline{w_{2}})$ contains a triple
    whose origin is in $u$ or $v$. If the
    $i$\textsuperscript{th} triple in $f(uv \mid \underline{w_{2}})$
    has origin in $u$, there will be a position in $f(u \mid
    \underline{vw_{2}})$ that will have a triple whose origin is in
    $u$ and the same position in $f(u \mid \underline{vw_{1}})$ will
    have $(*,*,\ri)$ and we are done. If the
    $i$\textsuperscript{th} triple in $f(uv \mid \underline{w_{2}})$
    has origin in $v$, there will be a position in $f(\underline{u} \mid
    v \mid \underline{w_{2}})$ that will have a triple whose origin is in
    $v$ and the same position in $f(\underline{u} \mid v \mid \underline{w_{1}})$ will
    have $(*,*,\ri)$ and we are done.

    Case 2: at position $i$, $f(uv \mid \underline{w_{2}})$ contains
    $(*,*,\ri)$ and $f(uv \mid \underline{w_{1}})$ contains a triple
    whose origin is in $u$ or $v$. This can be handled similarly as
    above, with the role of $w_{1}$ and $w_{2}$ interchanged.

    Case 3: at position $i$, $f(uv \mid \underline{w_{1}})$ contains a
    triple whose origin is in $u$ and $f(uv \mid \underline{w_{2}})$
    contains a triple whose origin is in $v$. In this case,
    $f(\underline{u} \mid v \mid \underline{w_{1}})$ will have a
    position with the triple $(*,*,\lef)$ and the same position in
    $f(\underline{u} \mid v \mid \underline{w_{2}})$ will have a
    triple whose origin is in $v$ and we are done.

    Case 4: at position $i$, $f(uv \mid \underline{w_{1}})$ contains a
    triple whose origin is in $v$ and $f(uv \mid \underline{w_{2}})$
    contains a triple whose origin is in $u$. This case can be handled
    similarly as above.

    Case 5: at position $i$, both $f(uv \mid \underline{w_{1}})$ and
    $f(uv \mid \underline{w_{2}})$ has triples whose origin is in $u$
    but the contents are different. In this case, there will be a
    position where $f(u \mid \underline{vw_{1}})$ and $f(u \mid
    \underline{vw_{2}})$ differ and we are done.

    Case 6: at position $i$, both $f(uv \mid \underline{w_{1}})$ and
    $f(uv \mid \underline{w_{2}})$ has triples whose origin is in $v$
    but the contents are different. In this case, there will be a
    position where $f(\underline{u} \mid v \mid \underline{w_{1}})$
    and $f(\underline{u} \mid v \mid \underline{w_{2}})$ differ and we
    are done.
\end{proof}

The following result makes it easier to compute certain factored
outputs.
\begin{lemma}
    \label{lem:currOrIfl}
    Suppose $f$ is a transduction without data peeking, $u,v$ are data
    words, $\sigma \in \Sigma$ and $d \in D$. The data values occurring
    in $f(\underline{u} \mid (\sigma,d) \mid \underline{v})$ are
    either $d$ or those that are $f$-memorable in $u$.
\end{lemma}
\begin{proof}
    From the hypothesis that $f$ is without data peeking, we infer that
    the data values occurring in $f(\underline{u} \mid (\sigma,d) \mid
    \underline{v})$ are either $d$ or those that occur in $u$.
    Suppose a data value $e \ne d$ occurs in $f(\underline{u} \mid
    (\sigma,d) \mid \underline{v})$. Let $e'$ be a safe replacement
    for $e$ in $u$. We have $f(\underline{u[e/e']} \mid (\sigma,d)
    \mid \underline{v}) \ne f(\underline{u} \mid (\sigma,d) \mid
    \underline{v})$, since $e$ cannot occur in $f(\underline{u[e/e']}
    \mid (\sigma,d) \mid \underline{v})$ but it does occur in
    $f(\underline{u} \mid (\sigma,d) \mid \underline{v})$. Applying
    the contrapositive of point~\ref{rch:rightToAbstract} in
    Lemma~\ref{lem:rightCongruenceHelper} to the above inequality, we
    infer that $f(\underline{u[e/e']}\mid (\sigma,d) \cdot
    v) \ne f(\underline{u}\mid (\sigma,d) \cdot v)$. According to
    Definition~\ref{def:fMemorable}, this certifies that $e$ is
    $f$-memorable in $u$.
\end{proof}

The following result uses the binary relation $\fequiv$ from
Definition~\ref{def:fEquivalence} and equalizing schemes from
Definition~\ref{def:equalScheme}.
\begin{lemma}
    \label{lem:fEqSameEffect}
    Suppose $f$ is a transduction that is invariant under
    permutations, $E$ is an equalizing scheme for $f$ and $u,u',v,w$
    are data words. If $u \fequiv u'$, then $f(\underline{E(u)(u)}
    \mid v \mid \underline{w}) = f_z(\underline{E(u')(u')}\mid v
    \mid \underline{w})$, where $z = |u| - |u'|$.
\end{lemma}
\begin{proof}
    Since $E(u)(u) \iso u$, we have $E(u)(u) \fequiv u$. So we infer
    that $E(u)(u) \fequiv u \fequiv u' \fequiv E(u')(u')$.
    Since $\fequiv$ is transitive, $E(u)(u) \fequiv E(u')(u')$. So we
    infer from Definition~\ref{def:fEquivalence} that there exists a
    permutation $\pi$ such that
    $\pi(\aifl_{f}(E(u')(u')))=\aifl_{f}(E(u)(u))$ and
    $f(\underline{E(u)(u)} \mid v\cdot w) =
    f_{z}(\underline{\pi(E(u')(u'))} \mid v \cdot w)$. Since $u
    \fequiv u'$, we infer from Definition~\ref{def:fEquivalence} and
    Definition~\ref{def:equalScheme} that
    $\aifl_{f}(E(u')(u'))=\aifl_{f}(E(u)(u))$, so $\pi$ (and hence
    $\pi^{-1}$) is identity on those data values that are
    $f$-influencing in $E(u')(u')$. Hence we infer from
    Lemma~\ref{lem:nonSufIflPermutable} that
    $f_{z}(\underline{\pi(E(u')(u'))} \mid v \cdot
    w)=f_{z}(\underline{\pi^{-1}\pcomp\pi(E(u')(u'))} \mid v \cdot
    w)=f_{z}(\underline{E(u')(u')} \mid v \cdot w)$. Hence,
    $f(\underline{E(u)(u)} \mid v \cdot w) =
    f_z(\underline{E(u')(u')} \mid v \cdot w)$. We infer from
    point~\ref{rch:rightToAbstract} of
    Lemma~\ref{lem:rightCongruenceHelper} that $f(\underline{E(u)(u)}
    \mid v \mid \underline{w}) =
    f_z(\underline{E(u')(u')} \mid v \mid \underline{w})$.
\end{proof}

Suppose a SSRT is at a configuration and reads a data word running a
sequence of transitions. If a permutation is applied to the
configuration  and the data word, then the new data word is read by
the SSRT starting from the new configuration running the same sequence
of transitions. This is formalized in the following result.
\begin{lemma}
    \label{lem:permToMatchTrans}
    Suppose $S$ is a SSRT, the set of registers $R$ is partitioned
    into two parts $R_{1},R_{2}$ and $(q,\val_{1},n_1), (q,\val_{2}, n_2)$ are
    configurations satisfying the following properties:
    \begin{itemize}
        \item $\val_{1}$ and $\val_{2}$ coincide on $R_{1}$,
        \item for every $r_{1},r_{2}\in R$,
            $\val_{1}(r_{1})=\val_{1}(r_{2})$ iff
            $\val_{2}(r_{1})=\val_{2}(r_{2})$ and
        \item $\set{\val_{1}(r) \mid r \in R_{1}} \cap
            \set{\val_{1}(r) \mid r \in R_{2}} = \emptyset =
            \set{\val_{2}(r) \mid r \in R_{1}} \cap \set{\val_{2}(r)
            \mid r \in R_{2}}$.
    \end{itemize}
    There exists a permutation $\pi$ that is identity on
    $\set{\val_{1}(r) \mid r \in R_{1}}$ such that for any data word
    $v$, the sequence of transitions executed when reading $v$ from
    $(q,\val_{1})$ is same as the sequence executed when reading $\pi(v)$ from
    $(q,\val_{2})$.
\end{lemma}
\begin{proof}
    Let $\pi$ be a permutation that is identity on $\set{\val_{1}(r)
    \mid r \in R_{1}}$ such that for every $r_{2} \in R_{2}$,
    $\pi(\val_{1}(r_{2}))=\val_{2}(r_{2})$. For every register
    $r$ and every position $i$ of $v$, $\val_{1}(r)=\data(v,i)$ iff
    $\val_{2}(r)=\data(\pi(v),i)$. The result follows by a routine
    induction on $|v|$.
\end{proof}

The next result says that if two strings belong to the same
equivalence class of $\fequiv$, then they can be equalized by an
equalizing scheme after which both will be transformed similarly by any
suffix. It uses the binary relation $\sequiv$ and the concept of
arrangements of elements of a set from
Section~\ref{sec:propTransdSSRT}.
\begin{lemma}
    \label{lem:dontCareNonIfl}
    Suppose $S$ is a SSRT implementing a transduction $f$, $u_{1}
    \sequiv u_{2}$, $S$ reaches the configuration $(q_{1},\val_{1}, |u_1|)$
    after reading $E(u_{1})(u_{1})$ and reaches $(q_{2},\val_{2}, |u_2|)$
    after reading $E(u_{2})(u_{2})$. For any data word $v$ and any
    $i$, if the $i$\tsc{th} left block of $f(E(u_{1})(u_{1})\mid
    \underline{v})$ is $\val_{1}(\chi)$ where $\chi$ is some
    arrangement of some subset $X'\subseteq X$, then the $i$\tsc{th}
    left block of $f(E(u_{2})(u_{2})\mid \underline{v})$ is
    $\val_{2}(\chi)$.
\end{lemma}
\begin{proof}
    Since $u_{1} \sequiv u_{2}$, $E(u_{1})(u_{1}) \sequiv
    E(u_2)(u_{2})$, so $q_{1}=q_{2}$, say $q_{1}=q_{2}=q$. For any
    $i$, the $i$\tsc{th} $f$-influencing value is $\delta_{i}$ in both
    $E(u_{1})(u_{1})$ and $E(u_{2})(u_{2})$. From
    condition~\ref{sequiv:ifl} of Definition~\ref{def:sequiv}, we
    infer that $\val_{1}$ and $\val_{2}$ coincide on all the registers
    that store $f$-influencing values. Suppose for the sake of
    contradiction that for some data word $v$ and some $i$, the
    $i$\tsc{th} left block of $f(E(u_{1})(u_{1})\mid \underline{v})$
    is $\val_{1}(\chi)$ and the $i$\tsc{th} left block of
    $f(E(u_{2})(u_{2})\mid \underline{v})$ is $\val_{2}(\chi') \ne
    \val_{2}(\chi)$. This means that while reading $v$ from
    $(q,\val_{2})$, the sequence of transitions is different from the
    sequence when reading $v$ from $(q,\val_{1})$. This difference is
    due to the difference between $\val_{1}$ and $\val_{2}$ in
    registers that don't store $f$-influencing values. Hence, we infer
    from Lemma~\ref{lem:permToMatchTrans} that there
    exists a permutation $\pi$ that is identity on $f$-influencing
    values such that the sequence of transitions executed when reading
    $v$ from $(q,\val_{1})$ is the same sequence executed when reading
    $\pi(v)$ from $(q,\val_{2})$. Hence, the $i$\tsc{th} left block of
    $f(E(u_{2})(u_{2})\mid \underline{\pi(v)})$ is $\val_{2}(\chi)$, which
    is different from the $i$\tsc{th} left block of
    $f(E(u_{2})(u_{2})\mid \underline{v})$, which is
    $\val_{2}(\chi')$. Since $f$ is invariant under permutations and
    without data peeking (from
    Proposition~\ref{prop:SSRTTransdInvPerm}), this contradicts
    Lemma~\ref{lem:nonPreIflPermutable}.
\end{proof}

\section{Proofs of Results in Section~\ref{sec:equivalences-abs}}
\label{app:Equivalences}
\begin{proof}[Proof of Lemma~\ref{lem:fEquivalence}]
    We have $u \fequiv u$ for all $u$, since the identity permutation
    satisfies all the conditions of Definition~\ref{def:fEquivalence}.
    Hence, $\fequiv$ is reflexive.

    Suppose $u_{1} \fequiv u_{2}$ and there exists a permutation $\pi$
    satisfying all the conditions of
    Definition~\ref{def:fEquivalence}. We have $\aifl_{f}(\pi(u_{2})) =
    \aifl_{f}(u_{1})$ and applying the permutation $\pi^{-1}$ on both
    sides gives us $\pi^{-1}(\aifl_{f}(\pi(u_{2}))) =
    \pi^{-1}(\aifl_{f}(u_{1}))$. Since $f$ is invariant under
    permutations, we infer from Lemma~\ref{lem:iflInvUnderPerm} that
    $\aifl_{f}(u_{2}) = \aifl_{f}(\pi^{-1}(u_{1}))$. For any $v$, we
    have $f_{z}(\underline{\pi(u_{2})} \mid \pi(v)) =
    f(\underline{u_{1}} \mid \pi(v))$, where $z=|u_1|-|u_2|$. Applying $\pi^{-1}$ on both
    sides and using Lemma~\ref{lem:factoredOutputInvariant}, we get
    $f_{z}(\underline{u_{2}} \mid v) = f(\pi^{-1}(u_{1}) \mid v)$ for
    any $v$. Hence, $\lambda v. f(\underline{u_{2}} \mid v) = \lambda
    v. f_{-z}(\underline{\pi^{-1}(u_{1})} \mid v)$. For all data
    words $u, v_{1}, v_{2}$, we have $f(u_{1} \cdot \pi(u) \mid
    \underline{\pi(v_{1})}) =
    f(u_{1} \cdot \pi(u) \mid \underline{\pi(v_{2})})$ iff
    $f(\pi(u_{2}) \cdot \pi(u) \mid
    \underline{\pi(v_{1})}) = f(\pi(u_{2}) \cdot \pi(u) \mid
    \underline{\pi(v_{2})})$. Applying $\pi^{-1}$ on both sides of
    both the equalities and using
    Lemma~\ref{lem:factoredOutputInvariant}, we get
    $f(\pi^{-1}(u_{1}) \cdot u \mid \underline{v_{1}}) = f(\pi^{-1}(u_{1}) \cdot u
    \mid \underline{v_{2}})$ iff $f(u_{2} \cdot u \mid \underline{v_{1}}) =
    f(u_{2} \cdot u \mid \underline{v_{2}})$. Hence, $\pi^{-1}$ satisfies all
    the conditions of Definition~\ref{def:fEquivalence}, so
    $u_{2} \fequiv u_{1}$, so $\fequiv$ is symmetric.

    Suppose $u_{1} \fequiv u_{2}$ and there exists a permutation $\pi$
    satisfying all the conditions of
    Definition~\ref{def:fEquivalence}. Suppose $u_{2} \fequiv u_{3}$
    and there exists a permutation $\pi'$ satisfying all the
    conditions of Definition~\ref{def:fEquivalence}. Let $\pi \pcomp
    \pi'$ be the composition of $\pi$ and $\pi'$  ($\pi \pcomp
    \pi'(u) = \pi(\pi'(u))$ for all u). It is routine
    verify the following equalities: $\ifl_{f}(\pi \pcomp
    \pi'(u_{3})) = \ifl_{f}(u_{1})$, $\lambda v. f_{z+z'}(\underline{\pi \pcomp
    \pi'(u_{3}}) \mid v) = f(\underline{u_{1}} \mid v)$ where
    $z=|u_{1}| - |u_{2}|$ and $z' = |u_{2}| - |u_{3}|$ and for
    all data words $u, v_{1}, v_{2}$, $f(u_{1} \cdot u \mid \underline{v_{1}}) =
    f(u_{1} \cdot u \mid \underline{v_{2}})$ iff $f(\pi\pcomp \pi'(u_{3}) \cdot u \mid
    \underline{v_{1}}) = f(\pi\pcomp \pi'(u_{3}) \cdot u \mid
    \underline{v_{2}})$. Hence $\fequiv$ is transitive.
\end{proof}

\iffalse
\begin{proof}[Proof of Lemma~\ref{lem:equalizerUnimportant}]
    Let $\delta_{1} \delta_{2} \cdots $ be a sequence of data values
    such that for every data word $u$ and every $i$, the
    $i$\textsuperscript{th} $f$-influencing data value of
    $E_{1}(u)(u)$ is $\delta_{i}$. Let $\eta_{1} \eta_{2} \cdots $ be
    a sequence of data values such that for every data word $u$ and
    every $i$, the $i$\textsuperscript{th} $f$-influencing data value
    of $E_{2}(u)(u)$ is $\eta_{i}$. Let $\pi$ be a permutation such
    that $\pi(\delta_{1} \delta_{2} \cdots) = \eta_{1} \eta_{2}
    \cdots$. Let $\set{v_{1}, \ldots, v_{m}}$ be a set of data words
    such that no two of them are in the same equivalence class of
    $\equiv_{f}^{E_{1}}$. We will show that no two data words in
    $\set{\pi(v_{1}), \ldots, \pi(v_{m})}$ are in the same equivalence
    class of $\equiv_{f}^{E_{2}}$.

    For all $i,j \in \set{1, \ldots, m}$ with $i \ne j$, let
    $u_{ij}$ be a data word such that $f(E_{1}(u_{ij})(u_{ij}) \mid
    \underline{v_{i}}) \ne f(E_{1}(u_{ij})(u_{ij}) \mid
    \underline{v_{j}})$. Applying the permutation $E_{2}(u_{ij}) \cdot
    E_{1}^{-1}(u_{ij})$ to both sides and using
    Lemma~\ref{lem:factoredOutputInvariant}, we infer that $f(E_{2}(u_{ij})(u_{ij}) \mid
    \underline{E_{2}(u_{ij}) \cdot
    E_{1}^{-1}(u_{ij})(v_{i})}) \ne f(E_{2}(u_{ij})(u_{ij}) \mid
    \underline{E_{2}(u_{ij}) \cdot
    E_{1}^{-1}(u_{ij})(v_{j})})$. Suppose
    $\ifl_{f}(E_{2}(u_{ij})(u_{ij})) = \eta_{1}\cdots \eta_{r}$. We
    will prove that there exist permutations $\pi_{i}, \pi_{j}$ such
    that they are identity
    on $\eta_{1} \cdots \eta_{r}$, $\pi_{i} \pcomp E_{2}(u_{ij})
    \pcomp E_{1}^{-1}(u_{ij})(v_{i}) = \pi(v_{i})$ and $\pi_{j} \pcomp
    E_{2}(u_{ij}) \pcomp
    E_{1}^{-1}(u_{ij})(v_{j}) = \pi(v_{j})$. Then, using
    Lemma~\ref{lem:nonPreIflPermutable}, we get
    \begin{align*}
      f(E_{2}(u_{ij})(u_{ij}) \mid \pi(v_{i})) & = f(E_{2}(u_{ij})(u_{ij}) \mid \pi_{i} \cdot E_{2}(u_{ij}) \cdot
    E_{1}^{-1}(u_{ij})(v_{i}))\\
    & = f(E_{2}(u_{ij})(u_{ij}) \mid E_{2}(u_{ij}) \cdot
    E_{1}^{-1}(u_{ij})(v_{i})) \\
    & \ne f(E_{2}(u_{ij})(u_{ij}) \mid
    \underline{E_{2}(u_{ij}) \cdot
    E_{1}^{-1}(u_{ij})(v_{j})})\\
    &=f(E_{2}(u_{ij})(u_{ij}) \mid
    \underline{\pi_{j} \cdot E_{2}(u_{ij}) \cdot
    E_{1}^{-1}(u_{ij})(v_{j})})\\
    &=f(E_{2}(u_{ij})(u_{ij}) \mid
    \underline{\pi(v_{j})}) \enspace .
    \end{align*}
    Hence, no two data words in $\set{\pi(v_{1}), \ldots, \pi(v_{m})}$
    are in the same equivalence class of $\equiv_{f}^{E_{2}}$.

    Now we will prove that there exists a permutation $\pi_{i}$ such
    that it is identity on $\eta_{1} \cdots \eta_{r}$ and $\pi_{i}
    \cdot E_{2}(u_{ij}) \cdot E_{1}^{-1}(u_{ij})(v_{i}) = \pi(v_{i})$.
    Let $\ifl_{f}(u_{ij}) = d_{1}\cdots d_{r}$. For all $i \in \set{1,
    \ldots, r}$, $\pi: \delta_{i} \mapsto \eta_{i}$ and
    $E_{1}^{-1}(u_{ij}): \delta_{i} \mapsto d_{i}$,
    $E_{2}(u_{ij})(u_{ij}): d_{i} \mapsto \eta_{i}$.  Define $\pi_{i}$
    such that $\pi_{i}: \eta_{i} \mapsto \eta_{i}$.  For $\delta
    \notin \set{\delta_{1}, \ldots, \delta_{r}}$, suppose $\pi: \delta
    \mapsto \eta$, $E_{1}^{-1}(u_{ij}): \delta \mapsto d$ and
    $E_{2}(u_{ij}): d \mapsto \eta'$. Define $\pi_{i}$ such that
    $\pi_{i}: \eta' \mapsto \eta$. Now $\pi_{i}$ is identity on
    $\eta_{1} \cdots \eta_{r}$ and $\pi_{i} \cdot E_{2}(u_{ij}) \cdot
    E_{1}^{-1}(u_{ij})(v_{i}) = \pi(v_{i})$. The existence of
    $\pi_{j}$ can be proved similarly.
\end{proof}
\fi

\section{Technical Details and Proofs of Results in Section~\ref{sec:constSSRT-abs}}
\label{app:ConstSSRT}
\subsection{Recognizing Influencing Values}
\label{sec:recIflVals}
\begin{proof}[Proof of Lemma~\ref{lem:iflValsMonotonic}]
    Suppose $d$ is $f$-memorable in $u\cdot (\sigma,e)$.
    There exists a data value $d'$ that is a safe replacement for
    $d$ in $u\cdot(\sigma,e)$ and a data word $v$ such that the next inequality is
    true.
    \begin{align*}
        f( \underline{(u \cdot (\sigma,e))[d/d']} \mid v) &\ne
        f(\underline{u \cdot (\sigma,e)} \mid v)\\
        f( \underline{u[d/d'] \cdot (\sigma,e)} \mid v) &\ne
        f(\underline{u \cdot (\sigma,e)} \mid v) && [d \ne e]\\
        f(\underline{u[d/d']} \mid (\sigma,e)\cdot v) &\ne
        f(\underline{u} \mid (\sigma,e)\cdot v)&&
        [\text{contrapositive of
        Lemma~\ref{lem:rightCongruenceHelper},
    point~\ref{rch:middleToLeft}}]
    \end{align*}
    The last inequality above shows that $d$ is $f$-memorable
    in $u$.

    Suppose $d$ is $f$-vulnerable in $u \cdot (\sigma,e)$.
    Then there exist data words $u',v$ and a data value $d'$ such that
    $d$ doesn't occur in $u'$, $d'$ is a  safe replacement for $d$ in
    $u \cdot (\sigma,e) \cdot u'\cdot v$ and $f(u\cdot (\sigma,e)
    \cdot u' \mid \underline{v[d/d']}) \ne f(u\cdot (\sigma,e) \cdot
    u' \mid \underline{v})$.  Since $d$ doesn't occur in $u'$ and $d
    \ne e$, $d$ doesn't occur in $(\sigma,e)\cdot u'$. We observe that
    $f(u\cdot ((\sigma,e) \cdot u') \mid \underline{v[d/d']}) \ne
    f(u\cdot ((\sigma,e) \cdot u') \mid \underline{v})$ to conclude
    that $d$ is $f$-vulnerable in $u$.
\end{proof}

\begin{proof}[Proof of Lemma~\ref{lem:SSRTIflValInv}]
By induction on $|u|$. The base case with $|u|=0$ is trivial. As
induction hypothesis, suppose
that after reading a data word $u$, the SSRT reaches the configuration
$( ([u]_f, \ptr), \val, |u|)$ such that $\val(\ptr(i))$ is the
$i$\tsc{th} $f$-influencing value in $u$ for all $i \in
[1,m]$, where $m=|\ifl_f(u)|$.
Suppose the SSRT reads $(\sigma,d) \in \Sigma \times D$ next. We give
the proof for the case where $d$ is not $f$-influencing in $u$ and it
is $f$-influencing in $u \cdot (\sigma,d)$. The other cases are
similar. Let $m'$ be the number of $f$-influencing values in
$E(u)(u)\cdot(\sigma,\delta_0)$. We infer from
Lemma~\ref{lem:rightCongruences} that $\delta_0$ is $f$-influencing in
$E(u)(u)\cdot(\sigma,\delta_0)$. We prove that the transition from
$([u]_f,\ptr)$ corresponding to $i=0$ in
Construction~\ref{const:SSRTIflVal} can be executed. We infer from
Lemma~\ref{lem:rightCongruences} that $u \cdot (\sigma,d) \fequiv
E(u)(u)\cdot (\sigma,\delta_0)$ so $[u \cdot (\sigma,d)]_f=[E(u)(u)
\cdot (\sigma,\delta_0)]_f$, the next state of the SSRT. The condition
$\phi=\bigwedge_{j=1}^{j=m}\ptr(j)^{\ne}$ is satisfied since $d$ is not
$f$-influencing in $u$ and for all $j \in [1,m]$, $\val(\ptr(j))$ is
the $j$\tsc{th} $f$-influencing value in $u$, which is not equal to
$d$. We infer from Lemma~\ref{lem:rightCongruences} that $u \cdot
(\sigma,d)$ has $m'$ $f$-influencing values. For every $j\in [1,m]$,
$\delta_j$ is $f$-influencing in $E(u)(u) \cdot (\sigma,\delta_0)$ iff
the $j$\tsc{th} $f$-influencing value in $u$ (which is assigned to
$\ptr(j)$ by $\val$) is $f$-influencing in $u \cdot (\sigma,d)$.
Since
$\delta_0$ is the $1$\tsc{st} $f$-influencing value in $E(u)(u)\cdot
(\sigma,\delta_0)$, $\ptr'(1)=r_{\reu}$ as given in
Construction~\ref{const:SSRTIflVal}. Since $r_\reu$ is the first
register in the set $R \setminus \set{\ptr(l) \mid 1 \le l \le m,
\delta_l \text{ is } f\text{-influencing in } E(u)(u)\cdot
(\sigma,\delta_0)}$, $r_{\reu}$ is the first
register that is not holding a data value that is $f$-influencing in
$u$ and in $u \cdot (\sigma,d)$. Since $R'=\set{r_{\reu}}$, the transition of
the SSRT changes the valuation to $\val'$ such that $\val'(r_{\reu})=d$. So $\val'(\ptr'(1)) = \val'(r_{\reu})=d$, the first
$f$-influencing value in $u\cdot (\sigma,d)$. Suppose $j \in [2,m']$
and the $j$\tsc{th} $f$-influencing value in
$E(u)(u)\cdot (\sigma,\delta_0)$ is $\delta_k$, the $k$\tsc{th} $f$-influencing value
in $E(u)(u)$ (this will be true for some $k$, by Lemma~\ref{lem:iflValsMonotonic}).
Since $R=\set{r_{\reu}}$, $\val$ and $\val'$ coincide on all registers except
$r_{\reu}$.  Since $r_{\reu}$ is the first
register in the set $R \setminus \set{\ptr(l) \mid 1 \le l \le m,
\delta_l \text{ is } f\text{-influencing in } E(u)(u)\cdot
(\sigma,\delta_0)}$, $r_{\reu} \ne \ptr(k)$ and $\val$
and $\val'$ coincide on $\ptr(k)$. Hence, $\val'(\ptr(k)) =
\val(\ptr(k))$. Since the $j$\tsc{th} $f$-influencing value in
$E(u)(u)\cdot (\sigma,\delta_0)$ is $\delta_k$, the $k$\tsc{th} $f$-influencing value
in $E(u)(u)$, we infer from Lemma~\ref{lem:rightCongruences} that the
$j$\tsc{th} $f$-influencing value in $u\cdot (\sigma,d)$ is the
$k$\tsc{th} $f$-influencing value in $u$. Hence,
$\val'(\ptr'(j))=\val'(\ptr(k))=\val(\ptr(k))$, which is the
$k$\tsc{th} $f$-influencing value in $u$ and the $j$\tsc{th}
$f$-influencing value in $u \cdot (\sigma,d)$. The first equality
above follows since $\ptr'(j) = \ptr(k)$ as given in
Construction~\ref{const:SSRTIflVal}.
\end{proof}

\subsection{Computing Transduction Outputs}
\label{sec:computeTransdOp}
\begin{lemma}
\label{lem:rEquivDist}
Suppose $f$ is a transduction that is invariant under permutations and
without data peeking and $E_1,E_2$ are equalizing schemes. Suppose a
set $V=\set{v_1,v_2, \ldots}$ has the following property: for any
$i\ne j$, there exists $u_{i,j}$ such that $f(E_1(u_{i,j})(u_{i,j}) \mid
\underline{v_i}) \ne f(E_1(u_{i,j})(u_{i,j}) \mid \underline{v_j})$.
Then there exists a set $V'=\set{v_1',v_2',\ldots}$ of the same
cardinality as $V$ such that for any $i \ne j$, $f(E_2(u_{i,j})(u_{i,j})
\mid \underline{v_i'}) \ne f(E_2(u_{i,j})(u_{i,j}) \mid
\underline{v_j'})$. For any $i \ne j$, the same $u_{i,j}$ works for
both $V$ and $V'$; we use the equalizing scheme $E_1$ for $V$ and $E_2$
for $V'$.
\end{lemma}
\begin{proof}
    Let $\delta_{1} \delta_{2} \cdots $ be the sequence of data values
    such that for every data word $u$ and every $i \in [1,|\ifl_f(u)|]$, the
    $i$\textsuperscript{th} $f$-influencing data value of
    $E_{1}(u)(u)$ is $\delta_{i}$. Let $\eta_{1} \eta_{2} \cdots $ be
    the sequence of data values such that for every data word $u$ and
    every $i \in [1,|\ifl_f(u)|]$, the $i$\textsuperscript{th} $f$-influencing data value
    of $E_{2}(u)(u)$ is $\eta_{i}$. Let $\pi$ be a permutation such
    that $\pi(\delta_{1} \delta_{2} \cdots) = \eta_{1} \eta_{2}
    \cdots$. Let $V'=\set{\pi(v_{i}) \mid v_i\in V}$. We will show
that $V'$ satisfies the condition of the lemma.

    For any $i \ne j$, let
    $u_{i,j}$ be a data word such that $f(E_{1}(u_{i,j})(u_{i,j}) \mid
    \underline{v_{i}}) \ne f(E_{1}(u_{i,j})(u_{i,j}) \mid
    \underline{v_{j}})$. Applying the permutation $E_{2}(u_{i,j}) \cdot
    E_{1}^{-1}(u_{i,j})$ to both sides and using
    Lemma~\ref{lem:factoredOutputInvariant}, we infer that
$f(E_{2}(u_{i,j})(u_{i,j}) \mid
    \underline{E_{2}(u_{i,j}) \cdot
    E_{1}^{-1}(u_{i,j})(v_{i})}) \ne f(E_{2}(u_{i,j})(u_{i,j}) \mid
    \underline{E_{2}(u_{i,j}) \cdot
    E_{1}^{-1}(u_{i,j})(v_{j})})$. Suppose
    $\ifl_{f}(E_{2}(u_{i,j})(u_{i,j})) = \eta_{1}\cdots \eta_{r}$. We
    will prove that there exist permutations $\pi_{i}, \pi_{j}$ such
    that they are identity
    on $\eta_{1} \cdots \eta_{r}$, $\pi_{i} \pcomp E_{2}(u_{i,j})
    \pcomp E_{1}^{-1}(u_{i,j})(v_{i}) = \pi(v_{i})$ and $\pi_{j} \pcomp
    E_{2}(u_{i,j}) \pcomp
    E_{1}^{-1}(u_{i,j})(v_{j}) = \pi(v_{j})$. Then, using
    Lemma~\ref{lem:nonPreIflPermutable}, we get
    \begin{align*}
      f(E_{2}(u_{i,j})(u_{i,j}) \mid \pi(v_{i})) & =
f(E_{2}(u_{i,j})(u_{i,j}) \mid \pi_{i} \cdot E_{2}(u_{i,j}) \cdot
    E_{1}^{-1}(u_{i,j})(v_{i}))\\
    & = f(E_{2}(u_{i,j})(u_{i,j}) \mid E_{2}(u_{i,j}) \cdot
    E_{1}^{-1}(u_{i,j})(v_{i})) \\
    & \ne f(E_{2}(u_{i,j})(u_{i,j}) \mid
    \underline{E_{2}(u_{i,j}) \cdot
    E_{1}^{-1}(u_{i,j})(v_{j})})\\
    &=f(E_{2}(u_{i,j})(u_{i,j}) \mid
    \underline{\pi_{j} \cdot E_{2}(u_{i,j}) \cdot
    E_{1}^{-1}(u_{i,j})(v_{j})})\\
    &=f(E_{2}(u_{i,j})(u_{i,j}) \mid
    \underline{\pi(v_{j})}) \enspace .
    \end{align*}

    Now we will prove that there exists a permutation $\pi_{i}$ such
    that it is identity on $\eta_{1} \cdots \eta_{r}$ and $\pi_{i}
    \cdot E_{2}(u_{i,j}) \cdot E_{1}^{-1}(u_{i,j})(v_{i}) = \pi(v_{i})$.
    Let $\ifl_{f}(u_{i,j}) = d_{1}\cdots d_{r}$. For all $i \in \set{1,
    \ldots, r}$, $\pi: \delta_{i} \mapsto \eta_{i}$ and
    $E_{1}^{-1}(u_{i,j}): \delta_{i} \mapsto d_{i}$,
    $E_{2}(u_{i,j})(u_{i,j}): d_{i} \mapsto \eta_{i}$.  Define $\pi_{i}$
    such that $\pi_{i}: \eta_{i} \mapsto \eta_{i}$.  For $\delta
    \notin \set{\delta_{1}, \ldots, \delta_{r}}$, suppose $\pi: \delta
    \mapsto \eta$, $E_{1}^{-1}(u_{i,j}): \delta \mapsto d$ and
    $E_{2}(u_{i,j}): d \mapsto \eta'$. Define $\pi_{i}$ such that
    $\pi_{i}: \eta' \mapsto \eta$. Now $\pi_{i}$ is identity on
    $\eta_{1} \cdots \eta_{r}$ and $\pi_{i} \cdot E_{2}(u_{i,j}) \cdot
    E_{1}^{-1}(u_{i,j})(v_{i}) = \pi(v_{i})$. The existence of
    $\pi_{j}$ can be proved similarly.
\end{proof}

\begin{proof}[Proof of Lemma~\ref{lem:finREquiv}]
Suppose for the sake of contradiction that $\requiv$ has infinite
index. Then, there is an infinite set $\set{v_i}_{i \ge 1}$ of data
words such that for any $j \ne k$, there exists a data word $u_{k,j}$
such that $f(E(u_{k,j})(u_{k,j}) \mid \underline{v_k}) \ne
f(E(u_{k,j})(u_{k,j})
\mid \underline{v_j})$. Let us say that a set $U$ of data words covers
a set $V \subseteq \set{v_i}_{i \ge 1}$ using $E$ if for every $v,v' \in V$,
there exists $E(u)(u) \in U$ such that $f(E(u)(u)\mid \underline{v})
\ne f(E(u)(u)\mid \underline{v'})$. Since $\fequiv$ has finite index,
at least one equivalence class of $\fequiv$ (say $U$) covers an
infinite subset (say $V$) of $\set{v_i}_{i \ge 1}$.

Now we build another equalizing scheme $E'$ as follows. Fix an
arbitrary data word $u \in U$. We infer from
Definition~\ref{def:fEquivalence} that for every $v \in U \setminus
\set{u}$, there exists a permutation $\pi_v$ such that
$\aifl_f(\pi(v))=\aifl_f(u)$. Let $E'$ be an equalizing scheme such
that $E'(u)$ is the identity permutation and for all $v \in U
\setminus \set{u}$, $E'(v)=\pi_v$. From Lemma~\ref{lem:rEquivDist}, we
infer that there exists an infinite set $V'$ that is covered by $U$
using $E'$.

We claim that for any $v_i' \ne v_j' \in V'$,
$f(u \mid \underline{v_i'}) \ne f(u \mid \underline{v_j'})$. Since $U$
covers $V'$ using $E'$, we infer that there exists a data word
$E'(u_{i,j})(u_{i,j}) \in U$ such that $f(E'(u_{i,j})(u_{i,j}) \mid
\underline{v_i'}) \ne f(E'(u_{i,j})(u_{i,j}) \mid
\underline{v_j'})$. Since $u \fequiv u_{i,j} \fequiv
E'(u_{i,j})(u_{i,j})$, we infer from the third condition of
Definition~\ref{def:fEquivalence} that there exists a permutation
$\pi_{i,j}$ such that $f(\pi_{i,j}(u_{i,j}) \mid \underline{v_i'}) \ne
f(\pi_{i,j}(u_{i,j}) \mid \underline{v_j'})$ iff $f(u\mid
\underline{v_i'}) \ne f(u \mid \underline{v_j'})$. Since we chose the
equalizing scheme $E'$ such that $\pi_{i,j}=E'(u_{i,j})$ and
$f(E'(u_{i,j})(u_{i,j}) \mid \underline{v_i'}) \ne
f(E'(u_{i,j})(u_{i,j}) \mid \underline{v_j'})$, we conclude that
$f(u\mid \underline{v_i'}) \ne f(u \mid \underline{v_j'})$, proving
the claim.

Now, $\set{f(u \mid \underline{v'}) \mid v' \in V'}$ is an infinite
set. Since there is no data peeking in $f$, $f(u \mid \underline{v'})$
contains data values only from $u$ for any $v' \in V'$. Hence, the
only way $\set{f(u \mid \underline{v'}) \mid v' \in V'}$ can be
infinite is that there is no bound on the length of the factored
outputs in that set. Since there are a fixed number of positions in $u$,
this contradicts the fact that $f$ has linear blow up. Hence, $\requiv$
has finite index.
\end{proof}

\begin{proof}[Proof of Lemma~\ref{lem:finiteLeftBlocks}]
    Suppose for the sake of contradiction that there is no such bound
    $B$. Then there is an infinite family of pairs of data words
    $(u_{1},v_{1}), (u_{2},v_{2}), \ldots$ such that for all
    $i \ge 1$, $f(u_{i} \mid \underline{v_{i}})$ has at least
    $i$ left blocks. Applying any permutation to $f(u_{i} \mid
    \underline{v_{i}})$ will not change the number of left blocks.
    From Lemma~\ref{lem:factoredOutputInvariant}, we infer that for
    all $i \ge 1$, $f(E(u_{i})(u_{i}) \mid
    \underline{E(u_{i})(v_{i})})$ has at least $i$ left blocks. Since
    $\requiv$ has finite index, there is at least one equivalence
    class of $\requiv$ that contains $E(u_{i})(v_{i})$ for infinitely
    many $i$. Let $v$ be a data word from this equivalence class. From
    the definition of $\requiv$ (Definition~\ref{def:rEquivalence}), we
    infer that for infinitely many $i$, $f(E(u_{i})(u_{i}) \mid
    \underline{v})$ has at least $i$ left blocks. Hence, for
    infinitely many $i$, $f(\underline{E(u_{i})(u_{i})} \mid v)$ has
    at least $(i-1)$ right blocks. Triples in the right blocks have
    origin in $v$. Since the number of positions in
    $v$ is bounded, this contradicts the hypothesis that $f$ has
    linear blow up.
\end{proof}

\begin{proof}[Proof of Lemma~\ref{lem:computeLeftBlocks}]
    Since $E(u\cdot (\sigma,d))^{-1}(v)$ and $E(u)^{-1}(\pi(v))$
    are obtained from applying different permutations to $v$, they are
    isomorphic. We will prove that for all $j$ and $i \ge 2$, if the
    $j$\tsc{th} position of $E(u\cdot (\sigma,d))^{-1}(v)$ contains
    the $i$\tsc{th} $f$-influencing value of $u\cdot (\sigma,d)$,
    then the same is contained in the $j$\tsc{th} position of
    $E(u)^{-1}(\pi(v))$.
    \begin{enumerate}
        \item The $j$\tsc{th} position of $E(u\cdot
            (\sigma,d))^{-1}(v)$ contains the $i$\tsc{th}
            $f$-influencing value of $u\cdot (\sigma,d)$.
        \item\label{leftBlocks:one} Hence, the $j$\tsc{th} position of
            $v$ contains $\delta_{i}$, by definition of equalizing
            schemes (Definition~\ref{def:equalScheme}).
        \item For $i \ge 2$,  the $i$\tsc{th} $f$-influencing value of
            $u\cdot (\sigma,d)$ is among $\set{d_{m}, \ldots,
            d_{1}}$, the $f$-influencing values in $u$, by
            Lemma~\ref{lem:iflValsMonotonic}.
        \item\label{leftBlocks:four} Say $d_{k}$ is the $i$\tsc{th} $f$-influencing value of
            $u\cdot (\sigma,d)$. Then $\delta_{k}$ is the $i$\tsc{th}
            $f$-influencing value of $E(u')(u')\cdot (\sigma,\eta)$,
            by Lemma~\ref{lem:rightCongruences}.
        \item\label{leftBlocks:two} The permutation $\pi$ maps $\delta_{i}$ to
            $\delta_{k}$, by
            Definition~\ref{def:middleBlocks}.
        \item\label{leftBlocks:three} The permutation $E(u)^{-1}$ maps $\delta_{k}$ to
            $d_{k}$, by definition of equalizing schemes
            (Definition~\ref{def:equalScheme}).
        \item The $j$\tsc{th} position of $E(u)^{-1}(\pi(v))$ contains
            $d_{k}$, by points \eqref{leftBlocks:one},
            \eqref{leftBlocks:two} and \eqref{leftBlocks:three} above.
        \item By point \eqref{leftBlocks:four} above, $d_{k}$ is the
            $i$\tsc{th} $f$-influencing value of $u\cdot (\sigma,d)$,
            so the $j$\tsc{th} position of $E(u)^{-1}(\pi(v))$
            contains the $i$\tsc{th} $f$-influencing value of $u\cdot
            (\sigma,d)$.
    \end{enumerate}
    Suppose $(d,\eta)\in\set{(d_{i},\delta_{i}) \mid i \in \set{1,
    \ldots, m}}$ or $(d,\eta)=(d_0,\delta_0)$ and the first $f$-influencing
    value in $u \cdot (\sigma,d)$ is among $\set{d_{m}, \ldots,
    d_{1}}$. Then we can put $i\ge 1$
    in the above reasoning to infer that for all $j$ and $i \ge 1$, if
    the $j$\tsc{th} position of $E(u\cdot (\sigma,d))^{-1}(v)$
    contains the $i$\tsc{th} $f$-influencing value of $u\cdot
    (\sigma,d)$, then the same is contained in the $j$\tsc{th}
    position of $E(u)^{-1}(\pi(v))$. Hence, we get the following
    equality.
    \begin{align*}
        f(u \cdot (\sigma,d) \mid \underline{ E(u \cdot
            (\sigma,d))^{-1}(v)}) &= f(u \cdot (\sigma,d) \mid \underline{ E(u
                )^{-1}(\pi(v))}) &&
                [\text{Lemma~\ref{lem:nonPreIflPermutable}}]\\
        f(u \mid \underline{(\sigma,d)  \cdot E(u \cdot
            (\sigma,d))^{-1}(v)}) &= f(u \mid \underline{(\sigma,d) \cdot E(u
                )^{-1}(\pi(v))}) &&
                [\text{Lemma~\ref{lem:rightCongruenceHelper},
                point~\ref{rch:middleToLeft}}]\\
        f(u \mid \underline{(\sigma,d)  \cdot E(u \cdot
            (\sigma,d))^{-1}(v)}) &= f(u \mid \underline{E(u
                )^{-1}((\sigma,\eta) \cdot \pi(v))}) &&
                [E(u)(d)=\eta]\\
    \end{align*}

    Suppose $(d,\eta)=(d_0,\delta_0)$ and the first $f$-influencing
    value in $u \cdot (\sigma,d)$ is $d$. Then $\pi$ maps
    $\delta_{1}$ to $\eta=\delta_0$. Let $\pi'$ be the permutation that
    interchanges $E(u)^{-1}(\eta)$ and $d$ and doesn't change any
    other value. For all $j$ and $i \ge 1$, if
    the $j$\tsc{th} position of $E(u\cdot (\sigma,d))^{-1}(v)$
    contains the $i$\tsc{th} $f$-influencing value of $u\cdot
    (\sigma,d)$, then the same is contained in the $j$\tsc{th}
    position of $\pi'\pcomp E(u)^{-1}(\pi(v))$. Hence, we get the
    following equality.
    \begin{align*}
        f(u \cdot (\sigma,d) \mid \underline{ E(u \cdot
            (\sigma,d))^{-1}(v)}) &= f(u \cdot (\sigma,d) \mid
            \underline{ \pi' \pcomp E(u)^{-1}(\pi(v))}) &&
                [\text{Lemma~\ref{lem:nonSufIflPermutable}}]\\
        f(u \mid \underline{(\sigma,d)  \cdot E(u \cdot
            (\sigma,d))^{-1}(v)}) &= f(u \mid \underline{(\sigma,d)
\cdot \pi' \pcomp E(u
                )^{-1}(\pi(v))}) &&
                [\text{Lemma~\ref{lem:rightCongruenceHelper},
                point~\ref{rch:middleToLeft}}]
    \end{align*}
    Since $\eta=\delta_0$ does not occur in
    $\set{\delta_{m},\ldots,\delta_{1}}$, $E(u)^{-1}(\eta)$ does not
    occur in $\set{d_{m}, \ldots, d_{1}}$, the $f$-influencing data
    values in $u$. Since $d$ also does not occur in $\set{d_{m},
    \ldots, d_{1}}$, $\pi'$ only interchanges two data values that are
    not $f$-influencing in $u$ and doesn't change any other value. So
    we infer from Lemma~\ref{lem:nonPreIflPermutable} that $f(u \mid
    \underline{(\sigma,d) \cdot \pi' \pcomp E(u)^{-1}(\pi(v))}) =
    f(u \mid \underline{\pi'((\sigma,d)) \cdot \pi' \pcomp \pi'
    \pcomp E(u)^{-1}(\pi(v))}) = f(u \mid
    \underline{E(u)^{-1}((\sigma,\eta)) \cdot
    E(u)^{-1}(\pi(v))}) = f(u \mid
    \underline{E(u)^{-1}((\sigma,\eta) \cdot
    \pi(v)}))$. Combining this with the equality above, we get
    $f(u \mid \underline{(\sigma,d)  \cdot E(u \cdot
        (\sigma,d))^{-1}(v)})=f(u \mid
    \underline{E(u)^{-1}((\sigma,\eta) \cdot
    \pi(v)}))$.

    Now we will prove the statements about $f(\underline{u} \mid
(\sigma,d) \mid \underline{E(u\cdot (\sigma,d))^{-1}(v)})$.
    Let $g$ be a function such that for $i \ge 2$, the $i$\tsc{th}
$f$-influencing value in $E(u)(u)\cdot (\sigma,\eta)$ is
$\delta_{g(i)}$.

    Case 1: $(d,\eta) \in \set{(d_{i},\delta_{i}) \mid i \in
    [1, m]}$. Let $\ifl_{f}(E(u\cdot (\sigma,d))(u\cdot
    (\sigma,d)))=\delta_{r} \cdots \delta_{1}$. We will first prove
    that $E(u) \pcomp E(u \cdot (\sigma,d))^{-1}$ coincides with
    $\pi$ on $\delta_{r}, \ldots, \delta_{1}$. For $i \ge 2$, $E(u
    \cdot (\sigma,d))^{-1}(\delta_{i})$ is the $i$\tsc{th}
    $f$-influencing value in $u \cdot (\sigma,d)$ and we infer from
    Lemma~\ref{lem:rightCongruences} that the $i$\tsc{th}
    $f$-influencing value in $u \cdot (\sigma,d)$ is $d_{g(i)}$,
    the $g(i)$\tsc{th} $f$-influencing value in $u$
    (since the $i$\tsc{th} $f$-influencing value in $E(u)(u)\cdot
    (\sigma,\eta)$ is $\delta_{g(i)}$, the $g(i)$\tsc{th}
    $f$-influencing value in $E(u)(u)$). By
    Definition~\ref{def:equalScheme}, $E(u)$ maps $d_{g(i)}$ to
    $\delta_{g(i)}$. Hence, for $i \ge 2$, $E(u)
    \pcomp E(u \cdot (\sigma,d))^{-1}$ maps $\delta_{i}$ to
    $\delta_{g(i)}$, which is exactly what $\pi$ does to $\delta_{i}$.
    Say the first $f$-influencing value in $E(u')(u')\cdot(\sigma,\eta)$ is
    $\delta_{j}$. We infer from Lemma~\ref{lem:rightCongruences} that
    the first $f$-influencing value in $u \cdot (\sigma,d)$ is
    $d_{j}$. Hence, $E(u) \pcomp E(u \cdot (\sigma,d))^{-1}$ maps
    $\delta_{1}$ to $\delta_{j}$, which is exactly what $\pi$ does to
    $\delta_{1}$. Hence, $E(u) \pcomp E(u \cdot (\sigma,d))^{-1}$
    coincides with $\pi$ on $\delta_{r}, \ldots, \delta_{1}$, the
    $f$-influencing values of $E(u\cdot (\sigma,d))(u \cdot
    (\sigma,d))$.

    \begin{align*}
        & E(u)(f(\underline{u} \mid (\sigma,d) \mid \underline{E(u \cdot
            (\sigma,d))^{-1}(v)}))\\
            &= E(u) \pcomp E(u \cdot (\sigma,d))^{-1} \pcomp E(u \cdot
            (\sigma,d))(f(\underline{u} \mid (\sigma,d) \mid \underline{E(u \cdot
            (\sigma,d))^{-1}(v)}))\\
            &=E(u) \pcomp E(u \cdot
            (\sigma,d))^{-1}(f(\underline{E(u \cdot (\sigma,d))(u)} \mid
            E(u \cdot (\sigma,d))(\sigma,d)) \mid \underline{v})
            && [\text{Lemma~\ref{lem:factoredOutputInvariant}}]\\
            &=f(\underline{E(u)(u)} \mid E(u)(\sigma,d) \mid
            \underline{\pi(v)} ) &&
            [\text{Lemma~\ref{lem:nonSufIflPermutableSecond}}]\\
            &=f_{z}(\underline{E(u')(u')} \mid (\sigma,\eta) \mid
            \underline{\pi(v)}) && [\text{Lemma~\ref{lem:fEqSameEffect}}]
    \end{align*}
    In the last inequality above, apart from
    Lemma~\ref{lem:fEqSameEffect}, we also use the fact that
    $E(u)(d)=E(u)(d_{i})=\delta_{i}=\eta$. So we get
    $E(u)(f(\underline{u} \mid (\sigma,d) \mid \underline{E(u \cdot
    (\sigma,d))^{-1}(v)}))=f_{z}(\underline{E(u')(u')} \mid
    (\sigma,\eta) \mid \underline{\pi(v)})$, concluding the proof
    for this case.

    Case 2: $(d,\eta)=(d_0,\delta_0)$. Let $\pi_{1}$ be any permutation
    satisfying the following conditions:
    \begin{itemize}
        \item For $i \ge 2$,
            $\pi_{1}(\delta_{i})=\pi(\delta_{i})$,
        \item if the first $f$-influencing value in
            $E(u')(u')\cdot(\sigma,\eta)$ is $\delta_{j}$ for some $j
\ge 1$, then
            $\pi_{1}(\delta_{1})=\pi(\delta_{1})$ and
        \item if the first $f$-influencing value in
            $E(u')(u')\cdot(\sigma,\eta)$ is $\eta=\delta_0$, then
            $\pi_{1}(\delta_{1})=E(u)(d)=E(u)(d_0)$.
    \end{itemize}
    As seen in case 1, $E(u) \cdot E(u
    \cdot (\sigma,d))^{-1}$ coincides with $\pi_{1}$ on $\delta_{r},
    \ldots, \delta_{2}$. If the first $f$-influencing value in
    $E(u)(u)\cdot(\sigma,\eta)$ is $\delta_{j}$ for some $j \ge 1$, then again as in case 1,
    $E(u) \cdot E(u \cdot (\sigma,d))^{-1}$ coincides with
    $\pi_{1}$ on $\delta_{1}$. If the first $f$-influencing value in
    $E(u)(u)\cdot(\sigma,\eta)$ is $\eta$, we infer from
    Lemma~\ref{lem:rightCongruences} that the first $f$-influencing
    value in $u \cdot (\sigma,d)$ is $d$, so $E(u \cdot
    (\sigma,d))^{-1}$ maps $\delta_{1}$ to $d$. In this case,
    $\pi_{1}(\delta_{1})=E(u)(d)$, so $E(u) \cdot E(u \cdot
    (\sigma,d))^{-1}$ coincides with $\pi_{1}$ on $\delta_{1}$. So
    $E(u) \cdot E(u \cdot (\sigma,d))^{-1}$ coincides with
    $\pi_{1}$ on $\delta_{r}, \ldots, \delta_{1}$. Hence, similar to
    case 1, we get $E(u)(f(\underline{u} \mid (\sigma,d) \mid \underline{E(u \cdot
        (\sigma,d))^{-1}(v)}))=f_{z}(\underline{E(u')(u')} \mid
    E(u)(\sigma,d) \mid \underline{\pi_{1}(v)} )$.

    Recall that $\delta_0$ is a data value that is not $f$-influencing in
    $E(u')(u')$ and does not occur in $\set{\delta_{m}, \ldots, \delta_{1}}$. Let
    $\pi'$ be the permutation that interchanges $\delta_0$ and $E(u)(d)$ and
    doesn't change any other value. Since $d$ is not $f$-influencing in
    $u$, $E(u)(d)$ does not occur in $\set{\delta_{m}, \ldots,
    \delta_{1}}$. Since the $f$-influencing values of $E(u')(u')$ are
    $\delta_{m}, \ldots, \delta_{1}$ and neither $\delta_0$ nor $E(u)(d)$ occur
    in $\set{\delta_{m}, \ldots, \delta_{1}}$, we get the following:
    \begin{align*}
        f(\underline{E(u')(u')} \mid (\sigma,\delta_0) \cdot
        \pi'\pcomp\pi_{1}(v)) &=f(\underline{\pi'\pcomp E(u')(u')} \mid
(\sigma,\delta_0) \cdot
        \pi'\pcomp\pi_{1}(v))
        &&[\text{Lemma~\ref{lem:nonSufIflPermutable}}]\\
        f(\underline{E(u')(u')} \mid (\sigma,\delta_0) \mid
        \underline{\pi'\pcomp\pi_{1}(v))}
        &=f(\underline{\pi'\pcomp E(u')(u')} \mid (\sigma,\delta_0) \mid
        \underline{\pi'\pcomp\pi_{1}(v)})
        &&[\text{Lemma~\ref{lem:rightCongruenceHelper},
        point~\ref{rch:rightToAbstract}}]\\
        E(u)(f(\underline{u} \mid (\sigma,d) \mid \underline{E(u \cdot
            (\sigma,d))^{-1}(v)})) &=f_{z}(\underline{E(u')(u')} \mid
                E(u)(\sigma,d) \mid
                \underline{\pi_{1}(v)})\\
        \pi' \pcomp E(u)(f(\underline{u} \mid (\sigma,d) \mid \underline{E(u \cdot
                    (\sigma,d))^{-1}(v)}))
                    &=\pi'(f_{z}(\underline{E(u')(u')} \mid
                    E(u)(\sigma,d) \mid
                    \underline{\pi_{1}(v)})) && [\text{apply } \pi'
                    \text{ on both sides}]\\
                    &=f_{z}(\underline{\pi'\pcomp E(u')(u')} \mid
                    (\sigma,\delta_0)
                    \mid \underline{\pi' \pcomp \pi_{1}(v)}) &&
                    [\text{Lemma~\ref{lem:factoredOutputInvariant}}]\\
                    &=f_{z}(\underline{E(u')(u')} \mid
(\sigma,\delta_0)
                    \mid \underline{\pi' \pcomp \pi_{1}(v)}) &&
                    [\text{second equality above}]
    \end{align*}
    For $i \in \set{1, \ldots, r}$, $\pi$ maps $\delta_{i}$ to the
    $i$\tsc{th} $f$-influencing value in $E(u')(u')\cdot
    (\sigma,\eta)$ by definition. We will prove that $\pi'\pcomp
    \pi_{1}$ does exactly the same on $\delta_{1}, \ldots,
    \delta_{r}$. For $i \ge 2$, $\pi$ maps $\delta_{i}$ to the
    $i$\tsc{th} $f$-influencing value in $E(u')(u')\cdot
    (\sigma,\eta)$, which is among $\delta_{m}, \ldots,
    \delta_{1}$. By definition, $\pi_{1}$ also maps $\delta_{i}$ to the
    $i$\tsc{th} $f$-influencing value in $E(u')(u')\cdot
    (\sigma,\eta)$, and $\pi'$ doesn't change this value, since
    neither $E(u)(d)$ nor $\delta_0$ are among $\delta_{m}, \ldots,
    \delta_{1}$. The permutation $\pi$ maps $\delta_{1}$ to the
    first $f$-influencing value in $E(u')(u')\cdot
    (\sigma,\eta)$. If this first $f$-influencing value is
    $\delta_{j}$ for  some $j \ge 1$, then, by definition, $\pi_{1}$ also maps $\delta_{1}$ to the
    first $f$-influencing value in $E(u')(u')\cdot
    (\sigma,\eta)$, and $\pi'$ doesn't change this value, since
    neither $E(u)(d)$ nor $\delta_0$ are among $\delta_{m}, \ldots,
    \delta_{1}$. If the first $f$-influencing value in $E(u')(u')\cdot
    (\sigma,\eta)$ is $\eta=\delta_0$, then $\pi$ maps $\delta_{1}$ to
    $\delta_0$. By definition, $\pi_{1}$ maps $\delta_{1}$ to $E(u)(d)$ and
    $\pi'$ maps $E(u)(d)$ to $\delta_0$. Hence, $\pi' \pcomp \pi_{1}$ maps
    $\delta_{1}$ to $\delta_0$. Therefore, for $i \in \set{1, \ldots, r}$,
    both $\pi$ and $\pi'\pcomp \pi_{1}$ map $\delta_{i}$ to the
    $i$\tsc{th} $f$-influencing value in $E(u')(u')\cdot
    (\sigma,\eta)$. Hence, we can apply
    Lemma~\ref{lem:nonPreIflPermutable} to get the next equality.
    \begin{align*}
        f(E(u')(u')\cdot (\sigma,\delta_0) \mid
        \underline{\pi'\pcomp \pi_{1}(v)}) &= f(E(u')(u')\cdot
(\sigma,\delta_0) \mid
        \underline{\pi(v)})\\
        f(\underline{E(u')(u')}\mid (\sigma,\delta_0) \mid
        \underline{\pi'\pcomp \pi_{1}(v)}) &=
        f(\underline{E(u')(u')}\mid (\sigma,\delta_0) \mid
        \underline{\pi(v)}) &&
        [\text{Lemma~\ref{lem:rightCongruenceHelper},
        point~\ref{rch:leftToAbstract}}]
    \end{align*}
    Hence $\pi' \pcomp E(u)(f(\underline{u} \mid (\sigma,d) \mid
    \underline{E(u
\cdot(\sigma,d))^{-1}(v)}))=f_z(\underline{E(u')(u')}\mid
(\sigma,\delta_0)
    \mid \underline{\pi(v)})$, concluding the proof for this case.
\end{proof}

\subsection{Dependency Trees}
\label{sec:dependencyTrees}
\begin{proof}[Proof of Lemma~\ref{lem:extAntTreeInv}]
    Suppose $[v]_{f}^{E}$ is an equivalence class and
    $\theta_{v},\theta$ are as explained in
    Definition~\ref{def:extCompAntTree}. If $d$ is the $i$\tsc{th} $f$-influencing value in $u$ for some $i
\ge 1$, let $\eta=\delta_i$ and let $\eta=\delta_0$ otherwise. Let
$u'$ be an arbitrary data word in $[u]_f$. We have from
    Lemma~\ref{lem:rightCongruences} that $u \cdot
    (\sigma,d)\fequiv E(u')(u')\cdot (\sigma,\eta)$, so
    $\pref(\theta)=[E(u')(u')\cdot(\sigma,\eta)]_{f}=[u \cdot
    (\sigma,d)]_{f}$ as required. We have from
    Lemma~\ref{lem:computeLeftBlocks} that $f(\underline{u} \mid
    (\sigma,d) \mid \underline{E(u \cdot (\sigma,d))^{-1}(v)})$
    is equal to either $E(u)^{-1}(f_{z}(\underline{E(u')(u')} \mid
    (\sigma,\eta) \mid \underline{\pi(v)}))$ or $E(u)^{-1}\pcomp
    \pi'(f_{z}(\underline{E(u')(u')} \mid (\sigma,\eta) \mid
    \underline{\pi(v)}))$. Hence, $f(\underline{u} \mid (\sigma,d)
    \mid \underline{E(u \cdot (\sigma,d))^{-1}(v)})$ and
    $f_{z}(\underline{E(u')(u')} \mid (\sigma,\eta) \mid
    \underline{\pi(v)})$ are isomorphic. Hence, the $i$\tsc{th} left block of
    $f(u \cdot(\sigma,d) \mid \underline{E(u \cdot
    (\sigma,d))^{-1}(v)})$ is the concretization of $z$, the
    $i$\tsc{th} non-right block of
    $f(\underline{E(u')(u')} \mid
    \underline{(\sigma,\eta)} \mid \underline{\pi(v)})$, as
    defined in Definition~\ref{def:extCompAntTree}. We will prove that
    $\val'(\ur(\theta,\bl(\theta,i)))$ is the concretization of $z$,
    which is sufficient to complete the proof.

    Indeed, $\val'(\ur(\theta,\bl(\theta,i))) = \val'(\ur(\theta,z'))$, where
    $z'$ is obtained from $z$ by replacing $j$\tsc{th} left block by $P_{j}$ and
    $k$\tsc{th} middle block by $\bv{\theta,k}$. Since we set $\val'(\bv{\theta,k})$ to be
    the $k$\tsc{th} middle block of $f(\underline{u} \mid
    (\sigma,d) \mid \underline{E(u\cdot (\sigma,d))^{-1}(v)})$,
    $\val'(\ur(\theta,\bl(\theta,i)))$ correctly concretizes the
    middle blocks. Since $\ur(\theta,P_{j}) =
    \ur(\theta_{v},\bl(\theta_v,j))$ and $\theta_{v}$ is a node in the
    original tree $T$, we infer that $\val(\ur(\theta_{v},\bl(\theta_v,j)))$
    is the $j$\tsc{th} left block of $f(u \mid
    \underline{E(u)^{-1}((\sigma,\eta)\cdot \pi(v))})$. Since $\val$
    and $\val'$ differ only in the variables $\bv{\theta,k}$ where
    $\theta$ is newly introduced, we infer that
    $\val'(\ur(\theta_{v},\bl(\theta_v,j)))=\val(\ur(\theta_{v},\bl(\theta_v,j)))$
	is the $j$\tsc{th} left block of $f(u \mid
	\underline{E(u)^{-1}((\sigma,\eta)\cdot\pi(v))})$. From
	Lemma~\ref{lem:computeLeftBlocks}, we infer that the $j$\tsc{th} left
	block of $f(u \mid \underline{E(u)^{-1}((\sigma,\eta)\cdot\pi(v))})$
	is equal to the $j$\tsc{th} left block of $f(u \mid
	\underline{(\sigma,d)\cdot E(u \cdot (\sigma,d))^{-1}(v)})$.
	Hence, $\val'(\ur(\theta,\bl(\theta,j)))$ correctly concretizes the
	left blocks.
\end{proof}
\begin{proof}[Proof of Lemma~\ref{lem:antTreeShortInv}]
  Suppose $T'$ is obtained from $T$ by removing a node $\theta$ and
  making the only child of $\theta$ a child of $\theta$'s parent. If
  the only child of $\theta$ is $\theta \cdot [v]_{f}^{E}$, we will
  prove that for all $i \in [1,B]$, $\ur(\theta\drlst \cdot
  [v]_{f}^{E},\bl(\theta\drlst \cdot [v]_{f}^{E},i))=\ur(\theta \cdot
  [v]_{f}^{E}, \bl(\theta \cdot [v]_{f}^{E},i))$.  This will imply
  that the unrolling of any block description in any leaf remains
  unchanged due to the shortening, so the lemma will be proved. First
  we will prove that
  $\ur(\theta\drlst\cdot[v]_{f}^{E},\bl(\theta,j))=\ur(\theta,\bl(\theta,j))$.
  Indeed, both are obtained from $\bl(\theta,j)$ by replacing every
  occurrence of $P_{k}$ by $\ur(\theta\drlst,\bl(\theta\drlst,k))$.

  We get $\ur(\theta\cdot [v]_{f}^{E},\bl(\theta\cdot [v]_{f}^{E},i))$
  from $\bl(\theta\cdot [v]_{f}^{E},i)$ by replacing every occurrence
  of $P_{j}$ by $\ur(\theta,\bl(\theta,j))$. We will prove that we also get $\ur(\theta\drlst\cdot
  [v]_{f}^{E},\bl(\theta\drlst\cdot [v]_{f}^{E},i))$ from
  $\bl(\theta\cdot [v]_{f}^{E},i)$ by replacing every occurrence
  of $P_{j}$ by
  $\ur(\theta,\bl(\theta,j))$, which is sufficient to prove the lemma.

  Recall that
  $\bl(\theta\drlst\cdot [v]_{f}^{E},i)$ is obtained from
  $\bl(\theta\cdot [v]_{f}^{E},i)$ by replacing every occurrence of
  $P_{j}$ by $\bl(\theta,j)$, as given in Definition~\ref{def:antTreeShort}. Hence, we get $\ur(\theta\drlst\cdot
  [v]_{f}^{E},\bl(\theta\drlst\cdot [v]_{f}^{E},i))$ from
  $\bl(\theta\cdot [v]_{f}^{E},i)$ by first replacing every occurrence
  of $P_{j}$ by $\bl(\theta,j)$, which is then replaced by
  $\ur(\theta\drlst\cdot[v]_{f}^{E},\bl(\theta,j))=
  \ur(\theta,\bl(\theta,j))$. Hence, for all $i \in [1,B]$,
  $\ur(\theta\drlst \cdot [v]_{f}^{E},\bl(\theta\drlst \cdot
  [v]_{f}^{E},i))=\ur(\theta \cdot [v]_{f}^{E}, \bl(\theta \cdot
  [v]_{f}^{E},i))$.
\end{proof}

\begin{proof}[Proof of Lemma~\ref{lem:antTreeExtShortTrimRed}]
    Suppose all leaves in $T$ are labeled with $[u]_{f}$ by $\pref$.
    Then all leaves in $T_{1}$ (and hence in $T_{2}$ and $T_{3}$) are
    labeled by $[u \cdot (\sigma,\eta)]_{f}$. All paths in $T_{2}$
    (and hence in $T_{3}$) are of length at most $|(\Sigma \times
    D)^{*}/\requiv|+1$: if there are longer paths, there will be at
    least $|(\Sigma \times D)^{*}/\requiv|+1$ leaves since each
    internal node has at least two children. However, this is not
    possible since $T_{2}$ has only one leaf for every equivalence
    class of $\requiv$. In $T_{3}$, for any node $\theta$ and any $i
    \in [1,B]$, $\bl(\theta,i)$ will only contain elements from
    $X_{\theta}$ and $\pbl$, as ensured in the trimming process in
    Definition~\ref{def:antTreeTrim}. There are at most $B$ parent
    references, each of which occurs at most once in $\bl(\theta,i)$
    for at most one $i \in [1,B]$. Since every non-parent block is
    replaced by a data word variable in the trimming process, each
    $\bl(\theta,i)$ is of length at most $2B+1$. Each $\bl(\theta,i)$
    has at most $(B+1)$ data word variables and $i \in [1,B]$, so at
    most $(B^{2}+B)$ data word variables are sufficient for the block
    descriptions in $\theta$.  Hence, $T_{3}$ is reduced.
\end{proof}

\begin{proof}[Proof of Lemma~\ref{lem:SSRTTransdInv}]
    Since $S$ is an extension of the SSRT constructed in
    Construction~\ref{const:SSRTIflVal}, the claim about the pointer function
    $\ptr$ comes from Lemma~\ref{lem:SSRTIflValInv}. For the
    We will prove that $(T,\val)$ is complete for $u$ by induction on $|u|$. For
the base case, $|u|=0$ and  we infer that $( ([\epsilon]_{f},\ptr_{\bot}, T_{\bot}),
\val_{\epsilon})$ is complete for $u=\epsilon$ by definition. We inductively
assume that after reading $u$, $S$ reaches the configuration $(
([u]_{f},\ptr,T),\val,|u|)$ such that $\val(\ptr(i))$ is the $i$\tsc{th} $f$-influencing
value in $u$ and $(T,\val)$ is complete for $u$.  Suppose the next symbol read
by the SSRT is $(\sigma,d)$ and $m=|\ifl_f(u)|$.

If $d$ is the $i$\tsc{th} $f$-influencing value in $u$ for some $i
\ge 1$, let $\eta=\delta_i$ and let $\eta=\delta_0$ otherwise. Let
$\pi$ be a permutation tracking influencing values on $E(u')(u') \cdot
(\sigma,\eta)$ as given in Definition~\ref{def:middleBlocks}.
Suppose $T_{1}$ is the $(\sigma,\eta)$ extension of $T$, $T_{2}$
is obtained from $T_{1}$ by shortening it as much as possible and
$T'$ is the trimming of $T_{2}$. Let $\update_{1}$ be the function as
defined in Construction~\ref{const:SSRTTransd}.  If $S$ had the transition
$( ([u]_{f},\ptr,T), \sigma, \phi, ([E(u')(u') \cdot
(\sigma,\eta)]_{f},\ptr',T_{1}),R',\update_{1})$, $S$ would read
$(\sigma,d)$ and reach the configuration $(([E(u')(u') \cdot
(\sigma,\eta)]_{f},\ptr',T_{1}), \val_1, |u|+1)$. We will prove that
$(T_{1},\val_{1})$ is complete for $u \cdot (\sigma,d)$. This can
be inferred from Lemma~\ref{lem:extAntTreeInv} if $\val_{1}$ is the
$(\sigma,d)$ extension of $(T, \val)$. This can be inferred if
$\val_{1}$ is obtained from $\val$ by setting
$\val_{1}(\bv{\theta,k})$ to the $k$\tsc{th} middle block of
$f(\underline{u}\mid (\sigma,d) \mid \underline{E(u \cdot
(\sigma,d))^{-1}(v)})$ for every leaf $\theta=\theta_{v} \cdot
[v]_{f}^{E}$ that is
newly added while extending $T$ to $T_{1}$. This can be inferred
from Lemma~\ref{lem:computeLeftBlocks} if $\val_{1}(\bv{\theta,k})$
is set to $z_{1}$, the $k$\tsc{th} middle block of
$E(u)^{-1}(f_{z}(\underline{E(u')(u')} \mid (\sigma,\eta) \mid
\underline{\pi(v)}))$ if $\eta=\delta_i$ for some $i \in [1,m]$ and $\val_{1}(\bv{\theta,k})$
is set to $z_{2}$, the $k$\tsc{th} middle block of $E(u)^{-1} \pcomp
\pi'(f_{z}(\underline{E(u')(u')} \mid (\sigma,\eta) \mid
\underline{\pi(v)}))$ if $\eta=\delta_{0}$, where $z=|u|-|u'|$ and
$\pi'$ is the permutation that interchanges $\delta_{0}$ and
$E(u)(d)$ and doesn't change any other value. From
the semantics of SSRTs, we infer that the third component in every
triple of $\val_{1}(\bv{\theta,k})$ is $|u|+1$, as required. Hence, it
remains to prove that
$\dropOrig(\val_{1}(\bv{\theta,k}))=\dropOrig(z_{1})$ if
$\eta=\delta_{i}$ for some $i \in [1,m]$ and
$\dropOrig(\val_{1}(\bv{\theta,k}))=\dropOrig(z_{2})$ if
$\eta=\delta_{0}$.

From Lemma~\ref{lem:currOrIfl}, we infer that all
data values in $f_{z}(\underline{E(u')(u')} \mid (\sigma,\eta) \mid
\underline{\pi(v)})$ are among $\set{\delta_{0}, \ldots,
\delta_{m}}$. Hence, we get $z_{1}$ and $z_{2}$ from the
$k$\tsc{th} middle block of $f_{z}(\underline{E(u')(u')} \mid (\sigma,\eta) \mid
\underline{\pi(v)})$ by replacing every occurrence of
$\delta_{j}$ for $j \in [1,m]$ by $E(u)^{-1}(\delta_{j})$ (which is the
$j$\tsc{th} $f$-influencing value in $u$) and replacing every
occurrence of $\delta_{0}$ by $E(u)^{-1}\pcomp\pi'(\delta_{0})$
(which is $d$). This exactly what the update function $\update_{1}$ does to
$\bv{\theta,k}$: it is set to the $k$\tsc{th} middle block of
$f_{z}(\underline{E(u')(u')} \mid (\sigma,\eta) \mid
\underline{\pi(v)})$ and every occurrence of $\delta_{j}$ is replaced by
$\ptr(j)$ (the transition of $S$ then replaces this with
$\val(\ptr(j))$, the $j$\tsc{th} $f$-influencing value in $u$) and
every occurrence of $\delta_{0}$ is replaced by $\curr$ (the
transition of $S$ then replaces this with $d$, the current data
value being read). Hence, $(T_{1},\val_{1})$ is complete for $u \cdot
(\sigma,d)$.

Since $T_2$ is obtained from $T_1$ by shortening it as much as
possible, we infer from Lemma~\ref{lem:antTreeShortInv} that
$(T_2,\val_1)$ is complete for $u \cdot (\sigma,d)$. The actual
transition in $S$ is $( ([u]_{f},\ptr,T), \sigma, \phi, ([E(u')(u')
\cdot (\sigma,\eta)]_{f},\ptr',T'),R',\update)$. After reading $(\sigma,d)$,
$S$ goes to the configuration $( ([u \cdot
(\sigma,\eta)]_{f},\ptr',T'),\val', |u|+1)$ where $\val'$ is the trimming of
$\val_{1}$ (due to the way $\update$ is defined from $\update_{1}$). Since $T'$ is
the trimming of $T_2$, we conclude from
Proposition~\ref{prop:antTreeTrimInv} that $(T',\val')$ is complete
for $u \cdot (\sigma,d)$.
\end{proof}

\section{Technical Details and Proofs of Results in Section~\ref{sec:propTransdSSRT}}
\label{app:TransdSSRT}
\begin{proof}[Proof of Lemma~\ref{lem:ifValsInRegisters}]
    Suppose a data value $d$ is not stored in any of the registers
    after reading $u$. We will prove that $d$ is neither $f$-memorable
    nor $f$-vulnerable in $u$. To prove that
    $d$ is not $f$-memorable in $u$, we will show that for
    any data word $v$ and any safe replacement $d'$ for $d$ in
    $u$, $f(\underline{u[d/d']} \mid v)=f(\underline{u}\mid v)$.
    Indeed, let $\pi$ be the permutation that interchanges $d$ and
    $d'$ and that doesn't change any other value. We have
    $u[d/d']=\pi(u)$. Suppose $S$ reaches the configuration
    $(q,\val)$ after reading $u$. We infer from
    Lemma~\ref{prop:SSRTInvPerm} that $S$ reaches the configuration
    $(q,\pi(\val))$ after reading $\pi(u)$. Since $d$ is not stored in
    any of the registers under the valuation $\val$, $\pi(\val)$
    coincides with $\val$ on all registers. Hence, if $S$ executes a
    sequence of transitions reading a data word $v$ from the
    configuration $(q,\val)$, the same sequence of transitions are
    executed reading $v$ from $(q,\pi(\val))$. Since
    $f(\underline{u[d/d']} \mid v)$ and $f(\underline{u}\mid v)$
    depends only on the sequence of transitions that are executed
    while reading $v$, we infer that $f(\underline{u[d/d']} \mid
    v)=f(\underline{u}\mid v)$.

    Next we will prove that if a data value $d$ is not stored in any
    of the registers after reading $u$, then $d$ is not $f$-vulnerable
    in $u$. Let $u',v$ be data words and $d'$ be a data
    value such that $d$ doesn't occur in $u'$.  Since $d$ is not
    stored in any of the registers after reading $u$ and $d$ doesn't
    occur in $u'$, $d$ is not stored in any of the registers after
    reading $u\cdot u'$. Suppose $d'$ is a safe replacement for $d$ in
    $u\cdot u'\cdot v$. Then $d'$ doesn't occur in $u\cdot u'$ so
    neither $d'$ nor $d$ is stored in any of the registers after reading
    $u \cdot u'$. Since $d'$ doesn't occur in $v$, $v \iso v[d/d']$.
    Hence the SSRT executes the same sequence of transitions for
    reading $u\cdot u'\cdot v$ and for $u \cdot u' \cdot v[d/d']$.
    Hence, the only difference between $f(u\cdot u'\cdot v)$ and
    $f(u\cdot u'\cdot v[d/d'])$ is that at some positions whose origin
    is not in $u \cdot u'$, the first one may contain $d$ and the
    second one may contain $d'$. Since such positions are abstracted
    out, $f(u\cdot u'\mid \underline{v[d/d']}) = f(u\cdot u'\mid
    \underline{v})$. Hence,
    $d$ is not $f$-vulnerable in $u$.
\end{proof}

\begin{proof}[Proof of Lemma~\ref{lem:fEquivFiniteIndex}]
    We will prove that $\sequiv$ refines $\fequiv$. Suppose
    $u_{1},u_{2}$ are data words such that $u_{1}\sequiv u_{2}$ and
    $S$ reaches the configurations $(q,\val_{1})$, $(q,\val_{2})$
    after reading $u_{1}$,$u_{2}$ respectively. Let $\pi$ be a
    permutation such that for every register $r$, $\pi(\val_{2}(r)) =
    \val_{1}(r)$. We can verify by a routine induction on $|u_{2}|$
    that after reading $\pi(u_{2})$, $S$ reaches the configuration
    $(q,\pi(\val_{2}))$.  We infer from
    Lemma~\ref{lem:ifValsInRegisters} that all $f$-influencing values
    of $u_{1}$ are stored in registers in the configuration
    $(q,\val_{1})$ and all $f$-influencing values of $\pi(u_{2})$ are
    stored in registers in the configuration $(q,\pi(\val_{2}))$. The
    valuations $\pi(\val_{2})$ and $\val_{1}$ coincide on all the
    registers. Hence, we can infer from condition~\ref{sequiv:ifl}
    of Definition~\ref{def:sequiv} that
    $\aifl_{f}(\pi(u_{2}))=\aifl_{f}(u_{1})$.
    
    Since $\pi(\val_{2})$ and $\val_{1}$ coincide on all the
    registers, for any data word $v$, the sequence of transitions
    executed when reading $v$ from the configuration $(q,\val_{1})$
    and from $(q,\pi(\val_{2}))$ are the same. Hence,
    $f_{z}(\underline{\pi(u_{2})} \mid v)=f(\underline{u_{1}}\mid v)$, where
    $z=|u_{1}|-|u_{2}|$.

    Let $u, v_{1},v_{2}$ be data words. To finish the proof, we have
    to show that $f(u_{1}\cdot u \mid \underline{v_{1}}) = f(u_{1}
    \cdot u \mid \underline{v_{2}})$ iff $f(\pi(u_{2})\cdot u \mid
    \underline{v_{1}}) = f(\pi(u_{2}) \cdot u \mid
    \underline{v_{2}})$. Any left factor of $f(u_{1} \mid \underline{u
    \cdot v_{1}})$ is of the form $\val_{1}(\chi)$, where
    $\chi_{1}$ is some arrangement of some subset $X_{1}\subseteq X$.
    Since $\val_{1}$ and $\pi(\val_{2})$ coincide on all the registers
    and $\val_{1}(x)=\epsilon$ iff $\pi(\val_{2})(x)=\epsilon$ for all
    data word variables $x \in X$ (by condition~\ref{sequiv:varEmpty}
    of Definition~\ref{def:sequiv}), it can be routinely verified that
    $f(u_{1} \mid \underline{u \cdot v_{1}})$ and $f(\pi(u_{2})\mid
    \underline{u \cdot v_{1}})$ have the same number of left blocks
    and right blocks. If the $i$\tsc{th} left block of $f(u_{1} \mid
    \underline{u \cdot v_{1}})$ is $\val_{1}(\chi)$, then the
    $i$\tsc{th} left block of $f(\pi(u_{2}) \mid
    \underline{u \cdot v_{1}})$ is $\pi(\val_{2})(\chi)$. We will
    assume that $f(u_{1}\cdot u \mid \underline{v_{1}}) \ne f(u_{1}
    \cdot u \mid \underline{v_{2}})$ and show that $f(\pi(u_{2})\cdot u \mid
    \underline{v_{1}}) \ne f(\pi(u_{2}) \cdot u \mid
    \underline{v_{2}})$. The proof of the converse direction is
    symmetric. It is sufficient to prove that either $f(\pi(u_{2})
    \mid \underline{u \cdot v_{1}}) \ne f(\pi(u_{2}) \mid \underline{u
    \cdot v_{2}})$ or $f(\underline{\pi(u_{2})} \mid u \mid
    \underline{v_{1}}) \ne f(\underline{\pi(u_{2})} \mid u \mid
    \underline{v_{2}})$; we can infer from the contrapositive of
    point~\ref{rch:middleToRight} or point~\ref{rch:leftToAbstract} of
    Lemma~\ref{lem:rightCongruenceHelper} respectively that
    $f(\pi(u_{2})\cdot u \mid \underline{v_{1}}) \ne f(\pi(u_{2})
    \cdot u \mid \underline{v_{2}})$. Since $f(u_{1}\cdot u \mid
    \underline{v_{1}}) \ne f(u_{1} \cdot u \mid \underline{v_{2}})$,
    we infer from the contrapositive of
    Lemma~\ref{lem:facOpEqualityFromParts} that either $f(u_{1} \mid
    \underline{u \cdot v_{1}}) \ne f(u_{1} \mid \underline{u \cdot
    v_{2}})$ or $f(\underline{u_{1}} \mid u \mid \underline{v_{1}})
    \ne f(\underline{u_{1}} \mid u \mid \underline{v_{2}})$.
    
    Case 1: $f(u_{1} \mid \underline{u \cdot v_{1}}) \ne f(u_{1} \mid
    \underline{u \cdot v_{2}})$. If the number of left blocks in
    $f(u_{1}\mid \underline{u \cdot v_{1}})$ is different from the
    number of left blocks in $f(u_{1}\mid \underline{u \cdot v_{2}})$,
    then the number of left blocks in
    $f(\pi(u_{2})\mid \underline{u \cdot v_{1}})$ is different from the
    number of left blocks in $f(\pi(u_{2})\mid \underline{u \cdot
    v_{2}})$ and we are done. Suppose $f(u_{1}\mid \underline{u \cdot
    v_{1}})$ and $f(u_{1}\mid \underline{u \cdot v_{2}})$ have the
    same number of left blocks but the $i$\tsc{th} left blocks are
    different. Suppose the $i$\tsc{th} left block of $f(u_{1}\mid
    \underline{u \cdot v_{1}})$ is $\val_{1}(\chi_{1})$ and the
    $i$\tsc{th} left block of $f(u_{1}\mid \underline{u \cdot v_{2}})$
    is $\val_{1}(\chi_{2})$, where $\chi_{1},\chi_{2}$ are some
    arrangements of some subsets $X_{1},X_{2} \subseteq X$
    respectively. The $i$\tsc{th} left block of $f(\pi(u_{2})\mid
    \underline{u \cdot v_{1}})$ is $\pi(\val_{2})(\chi_{1})$ and the
    $i$\tsc{th} left block of $f(\pi(u_{2})\mid \underline{u \cdot v_{2}})$
    is $\pi(\val_{2})(\chi_{2})$. Since $\val_{1}(\chi_{1}) \ne
    \val_{1}(\chi_{2})$, we infer from
    condition~\ref{sequiv:varArrEquiv} of Definition~\ref{def:sequiv} that
    $\pi(\val_{2})(\chi_{1}) \ne \pi(\val_{2})(\chi_{2})$. Hence, the
    $i$\tsc{th} left blocks of $f(\pi(u_{2})\mid \underline{u \cdot
    v_{1}})$ and $f(\pi(u_{2})\mid \underline{u \cdot v_{2}})$ are
    different and we are done.

    Case 2: $f(\underline{u_{1}} \mid u \mid \underline{v_{1}}) \ne
    f(\underline{u_{1}} \mid u \mid \underline{v_{2}})$. As we have
    seen in the second paragraph of this proof,
    $f_{z}(\underline{\pi(u_{2})} \mid u \cdot
    v_{1})=f(\underline{u_{1}} \mid u \cdot v_{1})$ and
    $f_{z}(\underline{\pi(u_{2})} \mid u \cdot
    v_{2})=f(\underline{u_{1}} \mid u \cdot v_{2})$. We infer from
    point~\ref{rch:rightToAbstract} of
    Lemma~\ref{lem:rightCongruenceHelper} that
    $f_{z}(\underline{\pi(u_{2})} \mid u \mid
    \underline{v_{1}})=f(\underline{u_{1}} \mid u \mid
    \underline{v_{1}})$ and $f_{z}(\underline{\pi(u_{2})} \mid u \mid
    \underline{v_{2}})=f(\underline{u_{1}} \mid u \mid
    \underline{v_{2}})$. Since $f(\underline{u_{1}} \mid u \mid \underline{v_{1}}) \ne
    f(\underline{u_{1}} \mid u \mid \underline{v_{2}})$, $f_{z}(\underline{\pi(u_{2})} \mid u \mid
    \underline{v_{1}}) \ne f_{z}(\underline{\pi(u_{2})} \mid u \mid
    \underline{v_{2}})$, hence $f(\underline{\pi(u_{2})} \mid u \mid
    \underline{v_{1}}) \ne f(\underline{\pi(u_{2})} \mid u \mid
    \underline{v_{2}})$ and we are done.
\end{proof}

\section{Proofs with Lengthy Case Analyses}
\label{app:Lengthy}
\begin{lemma}
    \label{lem:iflValuesPreserved}
    Suppose $f$ is a transduction that is invariant under permutations
    and without data peeking, $u$ is a data word and $e$ is a data value.
    If $d$ is a data value that is not $f$-influencing in $u$ and $d'$
    is a safe replacement for $d$ in $u$, then $e$ is $f$-memorable
    (resp.~$f$-vulnerable) in $u$ iff $e$ is
    $f$-memorable (resp.~$f$-vulnerable) in
    $u[d/d']$.
\end{lemma}
\begin{proof}
    The idea for the proof is the following. If a data value
    $e'$ and data word $v$ certify that $e$ is $f$-memorable
    in $u$, then some permutations can be applied on $e'$ and
    $v$ to certify that $e$ is $f$-memorable in $u[d/d']$.
    Similar strategies work for the converse direction and for
    $f$-vulnerable values.

    Suppose $e=d$. We have to prove that $d$ is not
    $f$-influencing in $u[d/d']$. Since $d$ doesn't occur in
    $u[d/d']$, we get $u[d/d'][d/d''] = u[d/d']$ for any data value
    $d''$. Hence, $f(\underline{u[d/d'][d/d'']} \mid v) =
    f(\underline{u[d/d']} \mid v)$ for all data words $v$, so $d$ is
    not $f$-memorable in $u[d/d']$. Since $d$ doesn't occur
    in $u[d/d']$, $d$ is not $f$-vulnerable in $u[d/d']$, as
    proved in Lemma~\ref{lem:prefIflOccurs}.

    Suppose $e \ne d$. First we will prove the statement about
    $f$-memorable data values. First we will assume that $e$
    is $f$-memorable in $u$ and prove that $e$ is $f$-memorable
    in $u[d/d']$. There exists a safe replacement $e'$ for
    $e$ in $u$ and a data word $v$ such that
    \begin{align}
        f(\underline{u[e/e']} \mid v) &\ne f(\underline{u} \mid v)
        \label{eq:eIsMemorable}
    \end{align}
    Let $e_{1} \not\in \data(u \cdot v, *) \cup
    \set{d,d',e,e'}$ be a fresh data value and $\pi_{1}$ be the
    permutation that interchanges $e'$ and $e_{1}$ and doesn't change
    any other data value. We apply $\pi_{1}$ to both sides of
    \eqref{eq:eIsMemorable} to get $\pi_{1}(f(\underline{u[e/e']} \mid
    v)) \ne \pi_{1}(f(\underline{u} \mid v))$. From
    Lemma~\ref{lem:factoredOutputInvariant}, we then infer that
    \begin{align}
        f(\underline{\pi_{1}(u[e/e'])} \mid \pi_{1}(v)) \ne
        f(\underline{\pi_{1}(u)} \mid \pi_{1}(v)) \enspace.
            \label{eq:eIsStillMemorable}
    \end{align}
    Since, $e'$ is a safe replacement for $e$ in $u$, $e'$ doesn't
    occur in $u$. Hence, $\pi_{1}(u[e/e']) = u[e/e_{1}]$ and
    $\pi_{1}(u) = u$. Using these in \eqref{eq:eIsStillMemorable}, we
    get
    \begin{align}
        f(\underline{u[e/e_{1}]} \mid \pi_{1}(v)) &\ne
        f(\underline{u} \mid \pi_{1}(v)) \enspace.
            \label{eq:eIsMemorableDueToe1}
    \end{align}
    Let $\pi_{2}$ be the permutation that interchanges $d$ and $d'$ and
    doesn't change any other data value. We apply $\pi_{2}$ to both sides
    of \eqref{eq:eIsMemorableDueToe1} to get
    $\pi_{2}(f(\underline{u[e/e_{1}]} \mid \pi_{1}(v))) \ne
    \pi_{2}(f(\underline{u} \mid \pi_{1}(v)))$. From
    Lemma~\ref{lem:factoredOutputInvariant}, we then infer that
    $f(\underline{\pi_{2}(u[e/e_{1}])} \mid \pi_{2}(\pi_{1}(v))) \ne
    f(\underline{\pi_{2}(u)} \mid \pi_{2}(\pi_{1}(v)))$. Since
    $d'$ is a safe replacement for $d$ in $u$, $d'$ doesn't occur in
    $u$. By choice, $d' \ne e_{1}$. Hence,
    $\pi_{2}(u[e/e_{1}])=u[e/e_{1}][d/d']=u[d/d'][e/e_{1}]$ and
    $\pi_{2}(u)=u[d/d']$. Using these in the last inequality, we get
    $f(\underline{u[d/d'][e/e_{1}]} \mid \pi_{2}(\pi_{1}(v))) \ne
    f(\underline{u[d/d']} \mid \pi_{2}(\pi_{1}(v)))$. This implies that
    $e$ is $f$-memorable in $u[d/d']$.

    For the converse direction, we will first prove that $d'$ is not
    $f$-memorable in $u[d/d']$. Suppose for the sake of
    contradiction that $d'$ is $f$-memorable in $u[d/d']$.
    Then there exists a data word $v$ and a data value $d''$ that is
    a safe replacement for $d'$ in $u[d/d']$ such that
    $f(\underline{u[d/d'][d'/d'']} \mid v) \ne f(\underline{u[d/d']}
    \mid v)$, so $f(\underline{u[d/d'']} \mid v) \ne
    f(\underline{u[d/d']} \mid v)$.  Now we apply the permutation $\pi_3$ that
    interchanges $d$ and $d''$ on both sides of this inequality and
    Lemma~\ref{lem:factoredOutputInvariant} implies that
    $f(\underline{u} \mid \pi_3(v)) \ne f(\underline{u[d/d']} \mid
    \pi_3(v))$. This shows that $d$ is $f$-memorable in $u$,
    a contradiction. Hence, $d'$ is not $f$-memorable in
    $u[d/d']$. Now, we have that $d'$ is not $f$-memorable in
    $u[d/d']$ and $d$ is a safe replacement for $d'$ in $u[d/d']$ and
    we have to prove that if $e$ is $f$-memorable in
    $u[d/d']$, then $e$ is $f$-memorable in $u$, which is
    same as $u[d/d'][d'/d]$. This is similar to proving that if $e$ is
    $f$-memorable in $u$, then $e$ is $f$-memorable
    in $u[d/d']$, which we have already proved.

    Next we will prove the statement about $f$-vulnerable data
    values. We have already proved the statement for $e=d$, so assume
    that $e \ne d$. First assume that $e$ doesn't occur in $u$. Then
    $e$ is not $f$-vulnerable in $u$. The value $e$ is also
    not $f$-vulnerable in $u[d/d']$ in the case where $d'\ne
    e$, since $e$ doesn't occur in $u[d/d']$. We will prove that $e$
    is not $f$-vulnerable in $u[d/e]$. Suppose for the sake of
    contradiction that $e$ is $f$-vulnerable in $u[d/e]$.
    There exist data words $u',v$ such that $e$ does not occur in
    $u'$  and there exists a data value $e'$ that
            is a safe replacement for $e$ in $u[d/e] \cdot u' \cdot v$ such
            that $f(u[d/e] \cdot u' \mid \underline{v}) \ne f(u[d/e] \cdot u' \mid
            \underline{v[e/e']})$.
    Now we apply the permutation $\pi$ that
    interchanges $d$ and $e$ on both sides of this inequality and
    Lemma~\ref{lem:factoredOutputInvariant} implies that $f(u
    \cdot\pi(u')\mid \underline{\pi(v)}) \ne f(u \cdot \pi(u') \mid
    \underline{\pi(v[e/e'])})$. We have $\pi(v[e/e']) = \pi(v)[d/e']$,
    so $f(u \cdot \pi(u') \mid \underline{\pi(v)}) \ne f(u \cdot
    \pi(u') \mid \underline{\pi(v)[d/e']})$. Since $e$ doesn't occur
    in $u'$, $d$ doesn't occur in $\pi(u')$. This implies that $d$ is
    $f$-vulnerable in $u$, a contradiction. So $e$ is not
    $f$-vulnerable in $u[d/e]$.
    
    Next we will assume that $e$ occurs in $u$. First we will assume
    that $e$ is $f$-vulnerable in $u$ and prove that $e$ is
    $f$-vulnerable in $u[d/d']$. Suppose that $e$ is
    $f$-vulnerable in $u$. So there
    exist data words $u',v$ such that $e$ doesn't occur in
    $u'$ and there exists a data value $e'$ that
            is a safe replacement for $e$ in $u \cdot u' \cdot v$ such
            that $f(u \cdot u' \mid \underline{v}) \ne f(u \cdot u' \mid
            \underline{v[e/e']})$.
    Let $e_{1} \not\in
    \data(u \cdot u' \cdot v, *) \cup \set{d,d',e,e'}$ be a fresh data
    value. The values $e',e_{1}$ don't occur in $u \cdot u' \cdot v$,
    so we can apply the permutation that interchanges $e'$ and $e_{1}$
    to both sides of the last inequality and
    Lemma~\ref{lem:factoredOutputInvariant} implies that $f(u \cdot u'
    \mid \underline{v}) \ne f(u \cdot u' \mid
    \underline{v[e/e_{1}]})$. Now we apply the permutation $\pi$ that
    interchanges $d$ and $d'$ to both sides of the last inequality and
    from Lemma~\ref{lem:factoredOutputInvariant}, we get that
    $f(u[d/d'] \cdot \pi(u')\mid \underline{\pi(v)}) \ne f(u[d/d']
    \cdot \pi(u') \mid \underline{\pi(v[e/e1])})$.  The value $d'$
    doesn't occur in $u$ (since $d'$ is a safe replacement for $d$ in
    $u$) but $e$ does, so $e\ne d'$.  We also have $d\ne e$, $d \ne
    e_{1}$ and $d'\ne e_{1}$, so $\set{d,d'} \cap \set{e,e_{1}} =
    \emptyset$.  Hence, $\pi(v[e/e_{1}]) = \pi(v)[e/e_{1}]$.  So we
    get $f(u[d/d'] \cdot \pi(u') \mid \underline{\pi(v)}) \ne
    f(u[d/d'] \cdot \pi(u') \mid \underline{\pi(v)[e/e1]})$,
    demonstrating that $e$ is a $f$-vulnerable value in
    $u[d/d']$ (note that since $e$ doesn't occur in $u'$, it doesn't
    occur in $\pi(u')$ also). Hence we have shown that
    when $e \ne d$, if $e$ is $f$-influencing in $u$, then $e$ is
    $f$-influencing in $u[d/d']$.

    For the converse direction, we will first prove that $d'$ is not
    $f$-vulnerable in $u[d/d']$. We have already proved that
    if $e$ doesn't occur in $u$, then $e$ is not $f$-vulnerable
    in $u[d/e]$. Since $d'$ doesn't occur in $u$, we can
    put $e=d'$ to conclude that $d'$ is not $f$-vulnerable in
    $u[d/d']$. Now, we have that $d'$ is not $f$-vulnerable in
    $u[d/d']$ and $d$ is a safe replacement for $d'$ in $u[d/d']$ and
    we have to prove that if $e$ is $f$-vulnerable in
    $u[d/d']$, then $e$ is $f$-vulnerable in $u$, which is
    same as $u[d/d'][d'/d]$. This is similar to proving that if $e$ is
    $f$-vulnerable in $u$, then $e$ is $f$-vulnerable
    in $u[d/d']$. Hence the proof is complete.
\end{proof}

\begin{lemma}
    \label{lem:allNewValuesSufIfl}
    Suppose $f$ is a transduction that is invariant under
    permutations, $\sigma \in \Sigma$ is a letter and $u$ is a data
    string.  If $d,e$ are data values, neither of which are
    $f$-influencing in $u$, then $d$ is $f$-memorable in $u
    \cdot (\sigma,d)$ iff $e$ is $f$-memorable in $u \cdot
    (\sigma,e)$. In addition, for any data value
    $\delta\notin\set{d,e}$, $\delta$ is $f$-memorable in
    $u \cdot (\sigma,d)$ iff $\delta$ is $f$-memorable in
    $u \cdot (\sigma,e)$.
\end{lemma}
\begin{proof}
    We will assume that $d$ is $f$-memorable in $u \cdot
    (\sigma,d)$ and prove that $e$ is $f$-memorable in $u \cdot
    (\sigma,e)$. The proof of the other direction is similar. Let
    $\pi$ be the permutation that interchanges $d$ and $e$ and doesn't
    change any other value. Since $d$ is $f$-memorable in
    $u\cdot (\sigma,d)$, there exist a data word $v$ and a data
    value $d'$ that is a safe replacement for $d$ in $u\cdot (\sigma,d)$ satisfying
    the next inequality. Let
    $\pi'$ be the permutation that interchanges $d'$ and $e$ and doesn't
    change any other value.
    \begin{align}
        f( \underline{(u\cdot (\sigma,d))[d/d']} \mid v) &\ne
        f( \underline{u\cdot (\sigma,d)} \mid v) \nonumber &&
        [\text{Definition~\ref{def:fMemorable}}]\\
        \pi(f( \underline{(u\cdot (\sigma,d))[d/d']} \mid v)) &\ne
        \pi(f( \underline{u\cdot (\sigma,d)} \mid v)) \nonumber &&
        [\text{apply }\pi\text{ to both sides}]\nonumber\\
        f( \underline{\pi((u\cdot (\sigma,d))[d/d'])} \mid \pi(v)) &\ne
        f( \underline{\pi(u\cdot (\sigma,d))} \mid \pi(v)) &&
        [\text{Lemma~\ref{lem:factoredOutputInvariant}}]\label{eq:sufIfl1}\\
        f(\underline{\pi(u)} \mid (\sigma,e)\cdot \pi(v)) &=
        f(\underline{u} \mid (\sigma,e) \cdot \pi(v)) \nonumber
        &&[\text{Lemma~\ref{lem:nonSufIflPermutable}}]\\
        f(\underline{\pi(u)\cdot (\sigma,e)} \mid \pi(v)) &=
        f(\underline{u \cdot (\sigma,e)} \mid \pi(v)) \nonumber
        &&[\text{Lemma~\ref{lem:rightCongruenceHelper},
        point~\ref{rch:middleToLeft}}]\\
        f(\underline{\pi(u\cdot (\sigma,d))} \mid \pi(v)) &=
        f(\underline{u \cdot (\sigma,e)} \mid \pi(v))
        \label{eq:sufIfl2}\\
        f(\underline{u} \mid (\sigma,d') \cdot \pi(v)) &=
        f(\underline{\pi(u)} \mid (\sigma,d') \cdot \pi(v)) \nonumber
        &&[\text{Lemma~\ref{lem:nonSufIflPermutable}}]\\
        f(\underline{\pi'(u)} \mid (\sigma,d') \cdot \pi(v)) &=
        f(\underline{\pi'\pcomp\pi(u)} \mid (\sigma,d') \cdot \pi(v)) \nonumber
        &&[\text{Lemma~\ref{lem:nonSufIflPermutable}}]\\
        f(\underline{u[e/d']} \mid (\sigma,d') \cdot \pi(v)) &=
        f(\underline{\pi(u[d/d'])} \mid (\sigma,d') \cdot \pi(v)) \nonumber
        &&[d' \notin \data(u,*)]\\
        f(\underline{(u \cdot (\sigma,e))[e/d']} \mid \pi(v)) &=
        f(\underline{\pi( (u \cdot (\sigma,d))[d/d'])} \mid \pi(v))
        &&[d' \notin \data(u,*)]\label{eq:sufIfl3}\\
        f(\underline{(u \cdot (\sigma,e))[e/d']} \mid \pi(v)) &\ne
        f(\underline{u \cdot (\sigma,e)} \mid \pi(v)) \nonumber
        &&[\text{\eqref{eq:sufIfl1},\eqref{eq:sufIfl2}, \eqref{eq:sufIfl3}}]
    \end{align}
    From the last inequality above, we conclude that $e$ is $f$-memorable
    in $u \cdot (\sigma,e)$.

    Next we will assume that $\delta$ is $f$-memorable in $u
    \cdot (\sigma,d)$ and prove that $\delta$ is $f$-memorable
    in $u \cdot (\sigma,e)$. The proof of the other
    direction is similar. Since $\delta$ is $f$-memorable in
    $u \cdot (\sigma,d)$, there exists a data value $\delta'$ that is
    safe for replacing $\delta$ in $u \cdot (\sigma,d)$ and a data
    word $v$ such that
    $f(\underline{(u\cdot(\sigma,d))[\delta/\delta']} \mid v) \ne
    f(\underline{u \cdot (\sigma,d)} \mid v)$. Let $\delta''$ be a
    data value that is a safe replacement for $\delta$ in $u \cdot
    (\sigma,d) \cdot (\sigma,e)$. Let $\pi_{1}$ be the permutation
    that interchanges $\delta'$ and $\delta''$ and doesn't change any
    other value. Let $\pi_{2}$ be the permutation that interchanges
    $\delta$ and $\delta''$ and doesn't change any other value.
    \begin{align}
        f(\underline{(u\cdot(\sigma,d))[\delta/\delta']} \mid v) &\ne
        f(\underline{u \cdot (\sigma,d)} \mid v) \nonumber\\
        \pi_{1}(f(\underline{(u\cdot(\sigma,d))[\delta/\delta']} \mid
        v)) &\ne
        \pi_{1}(f(\underline{u \cdot (\sigma,d)} \mid v)) &&
        [\text{apply } \pi_{1} \text{ on both sides}] \nonumber\\
        f(\underline{\pi_{1}((u\cdot(\sigma,d))[\delta/\delta'])} \mid
        \pi_{1}(v)) &\ne
        f(\underline{\pi_{1}(u \cdot (\sigma,d))} \mid \pi_{1}(v)) &&
[\text{Lemma~\ref{lem:factoredOutputInvariant}}]\nonumber\\
        f(\underline{(u\cdot(\sigma,d))[\delta/\delta'']} \mid
        \pi_{1}(v)) &\ne
        f(\underline{u \cdot (\sigma,d)} \mid \pi_{1}(v)) &&
        [\delta',\delta'' \notin \data(u\cdot (\sigma,d),*)]
        \nonumber\\
        \pi(f(\underline{(u\cdot(\sigma,d))[\delta/\delta'']} \mid
        \pi_{1}(v))) &\ne
        \pi(f(\underline{u \cdot (\sigma,d)} \mid \pi_{1}(v))) &&
        [\text{apply } \pi \text{ on both sides}]
        \nonumber\\
        f(\underline{\pi(u[\delta/\delta''])\cdot(\sigma,e)} \mid
        \pi\pcomp \pi_{1}(v)) &\ne
        f(\underline{\pi(u) \cdot (\sigma,e)} \mid \pi\pcomp
\pi_{1}(v)) &&[\text{Lemma~\ref{lem:factoredOutputInvariant}}]
        \label{eq:sufIfl4}\\
        f(\underline{u}\mid (\sigma,d) \cdot \pi_{1}(v)) &=
        f(\underline{\pi(u)} \mid (\sigma,d) \cdot (\pi_{1}(v)))
        \nonumber && [\text{Lemma~\ref{lem:nonSufIflPermutable}}]\\
        \pi(f(\underline{u}\mid (\sigma,d) \cdot \pi_{1}(v))) &=
        \pi(f(\underline{\pi(u)} \mid (\sigma,d) \cdot \pi_{1}(v)))
        \nonumber && [\text{apply } \pi \text{ on both sides}]\\
        f(\underline{\pi(u)}\mid (\sigma,e) \cdot \pi \pcomp \pi_{1}(v)) &=
        f(\underline{u} \mid (\sigma,e) \cdot \pi \pcomp \pi_{1}(v))
        &&[\text{Lemma~\ref{lem:factoredOutputInvariant}}]\nonumber\\
        f(\underline{\pi(u)  \cdot(\sigma,e)}\mid \pi \pcomp \pi_{1}(v)) &=
        f(\underline{u  \cdot (\sigma,e)} \mid\pi \pcomp \pi_{1}(v))
        \label{eq:sufIfl5} &&
        [\text{Lemma~\ref{lem:rightCongruenceHelper},
        point~\ref{rch:middleToLeft}}]\\
        d,e &\notin \data(\ifl_f(\pi_2(u)),
*)&&[\set{d,e}\cap\set{\delta,\delta''}=\emptyset, \text{
Lemma~\ref{lem:iflInvUnderPerm}}]\nonumber\\
        f(\underline{\pi_2(u)} \mid (\sigma,d)\cdot \pi_1(v)) &=
f(\underline{\pi\pcomp \pi_2(u)} \mid (\sigma,d)\cdot \pi_1(v)) &&
[\text{Lemma~\ref{lem:nonSufIflPermutable}}]\nonumber\\
        \pi(f(\underline{\pi_2(u)} \mid (\sigma,d)\cdot \pi_1(v))) &=
\pi(f(\underline{\pi\pcomp \pi_2(u)} \mid (\sigma,d)\cdot \pi_1(v))) &&
[\text{apply }\pi\text{ on both sides}]\nonumber\\
        f(\underline{\pi\pcomp \pi_2(u)} \mid (\sigma,e)\cdot \pi\pcomp \pi_1(v)) &=
f(\underline{\pi_2(u)} \mid (\sigma,e)\cdot \pi\pcomp \pi_1(v)) &&
[\text{Lemma~\ref{lem:factoredOutputInvariant}}]\nonumber\\
        f(\underline{\pi(u[\delta/\delta''])} \mid (\sigma,e)\cdot \pi\pcomp \pi_1(v)) &=
f(\underline{u[\delta/\delta'']} \mid (\sigma,e)\cdot \pi\pcomp \pi_1(v)) &&
[\delta'' \notin \data(u,*)]\nonumber\\
        f(\underline{\pi(u[\delta/\delta''])\cdot (\sigma,e)} \mid \pi\pcomp \pi_1(v)) &=
f(\underline{u[\delta/\delta'']\cdot (\sigma,e)} \mid \pi\pcomp \pi_1(v)) &&
[\text{Lemma~\ref{lem:rightCongruenceHelper},
point~\ref{rch:middleToLeft}}]\label{eq:sufIfl6}\\
f(\underline{u[\delta/\delta'']\cdot (\sigma,e)} \mid
\pi\pcomp \pi_1(v)) &\ne f(\underline{u  \cdot (\sigma,e)} \mid\pi \pcomp
\pi_{1}(v)) && [\eqref{eq:sufIfl4}, \eqref{eq:sufIfl5},
\eqref{eq:sufIfl6}]\nonumber\\
f(\underline{(u\cdot (\sigma,e))[\delta/\delta'']} \mid
\pi\pcomp \pi_1(v)) &\ne f(\underline{u  \cdot (\sigma,e)} \mid\pi \pcomp
\pi_{1}(v)) && [\delta\ne e]\nonumber
    \end{align}
The last inequality above certifies that $\delta$ is $f$-memorable
in $u \cdot (\sigma,e)$.
\end{proof}

\begin{lemma}
    \label{lem:allNewValuesPrefIfl}
    Suppose $f$ is a transduction that is invariant under
    permutations, $\sigma \in \Sigma$ is a letter and $u$ is a data
    string.  If $d,e$ are data values, neither of which are
    $f$-influencing in $u$, then $d$ is $f$-vulnerable in $u
    \cdot (\sigma,d)$ iff $e$ is $f$-vulnerable in $u \cdot
    (\sigma,e)$. In addition, for any data value
    $\delta\notin\set{d,e}$, $\delta$ is $f$-vulnerable in
    $u \cdot (\sigma,d)$ iff $\delta$ is $f$-vulnerable in
    $u \cdot (\sigma,e)$.
\end{lemma}
\begin{proof}
    We will assume that $d$ is $f$-vulnerable in $u \cdot
    (\sigma,d)$ and prove that $e$ is $f$-vulnerable in $u
    \cdot (\sigma,e)$. The proof of the other direction is similar.
    Let $\pi$ be the permutation that interchanges $d$ and $e$ and
    doesn't change any other value. Since $d$ is $f$-vulnerable
    in $u \cdot (\sigma,d)$, we infer from
    Definition~\ref{def:fMemorable} that there exist data words
    $u',v$ and a data value $d'$ such that $d$ doesn't occur in
    $u'$, $d'$ is a safe replacement for $d$ in
            $u\cdot (\sigma,d) \cdot
            u'\cdot v$ and $f(u\cdot (\sigma,d)\cdot u' \mid
            \underline{v[d/d']}) \ne f(u\cdot (\sigma,d)\cdot u' \mid \underline{v})$.
    Applying the contrapositive of Lemma~\ref{lem:facOpEqualityFromParts} to the
    above inequality, we infer that at least one of the
    following inequalities are true.
    \begin{align*}
        f(\underline{u} \mid (\sigma,d) \cdot u' \mid
        \underline{v[d/d']}) &\ne f(\underline{u} \mid (\sigma,d) \cdot u' \mid
        \underline{v})\\
        f(u \mid \underline{(\sigma,d) \cdot u' \cdot v[d/d']}) &\ne f(u
        \mid \underline{(\sigma,d) \cdot u' \cdot v})
    \end{align*}
    Each of the above inequalities is taken up in one of the following
    cases. Let $\pi$ be the permutation that interchanges $d$ and $e$
    and doesn't change any other value. Let $d''$ be a data value such
    that $d'' \notin \data(u\cdot u'\cdot v,*) \cup
    \set{d,e,d',\pi(d),\pi(d'), \pi(e), \pi(e')}$. Let $\pi'$  be the
    permutation that interchanges $d'$ and $d''$ and doesn't change
    any other value.

    Case 1:
    \begin{align}
        f(\underline{u} \mid (\sigma,d) \cdot u' \mid
        \underline{v[d/d']}) &\ne f(\underline{u} \mid (\sigma,d) \cdot u' \mid
        \underline{v})\nonumber\\
        \pi'(f(\underline{u} \mid (\sigma,d) \cdot u' \mid
        \underline{v[d/d']})) &\ne \pi'(f(\underline{u} \mid (\sigma,d) \cdot u' \mid
        \underline{v})) && [{\text{apply } \pi' \text{ to both sides}}]\nonumber \\
        f(\underline{u} \mid (\sigma,d) \cdot u' \mid
        \underline{v[d/d'']}) &\ne f(\underline{u} \mid (\sigma,d) \cdot u' \mid
        \underline{v})&&
        [\text{Lemma~\ref{lem:factoredOutputInvariant}, }d',d'' \notin
        \data(u\cdot u' \cdot (\sigma,d) \cdot v,*)]\label{eq:allNewValuesPrefIfl1}\\
        f(\underline{u} \mid (\sigma,d) \cdot u' \cdot v[d/d'']) &=
        f(\underline{\pi(u)} \mid (\sigma,d) \cdot u' \cdot v[d/d''])
        && [\text{Lemma~\ref{lem:nonSufIflPermutable}}] \nonumber\\
        f(\underline{u} \mid (\sigma,d) \cdot u' \mid
        \underline{v[d/d'']}) &=
        f(\underline{\pi(u)} \mid (\sigma,d) \cdot u' \mid
        \underline{v[d/d'']})
        && [\text{point~\ref{rch:rightToAbstract} of Lemma~\ref{lem:rightCongruenceHelper}}]
        \label{eq:allNewValuesPrefIfl2}\\
        f(\underline{u} \mid (\sigma,d) \cdot u' \cdot v) &=
        f(\underline{\pi(u)} \mid (\sigma,d) \cdot u' \cdot v)
        && [\text{Lemma~\ref{lem:nonSufIflPermutable}}]
        \nonumber\\
        f(\underline{u} \mid (\sigma,d) \cdot u' \mid \underline{v}) &=
        f(\underline{\pi(u)} \mid (\sigma,d) \cdot u' \mid
        \underline{v})
        && [\text{point~\ref{rch:rightToAbstract} of Lemma~\ref{lem:rightCongruenceHelper}}]
        \label{eq:allNewValuesPrefIfl3}\\
        f(\underline{\pi(u)} \mid (\sigma,d) \cdot u' \mid
        \underline{v[d/d'']}) &\ne f(\underline{\pi(u)} \mid (\sigma,d) \cdot u' \mid
        \underline{v}) &&
        [\text{\eqref{eq:allNewValuesPrefIfl1},
        \eqref{eq:allNewValuesPrefIfl2},
    \eqref{eq:allNewValuesPrefIfl3}}] \nonumber\\
        \pi(f(\underline{\pi(u)} \mid (\sigma,d) \cdot u' \mid
        \underline{v[d/d'']})) &\ne \pi(f(\underline{\pi(u)} \mid (\sigma,d) \cdot u' \mid
        \underline{v})) && [\text{apply } \pi \text{ on both sides}]
        \nonumber\\
        f(\underline{u} \mid (\sigma,e) \cdot \pi(u') \mid
        \underline{\pi(v)[e/d'']}) &\ne f(\underline{u} \mid
        (\sigma,e) \cdot \pi(u') \mid
        \underline{\pi(v)}) &&
        [\text{Lemma~\ref{lem:factoredOutputInvariant},
        }\pi(\pi(u))=u, \pi(v[d/d''])=\pi(v)[e/d'']]\nonumber\\
        f(u \cdot (\sigma,e) \cdot \pi(u') \mid
        \underline{\pi(v)[e/d'']}) &\ne f(u \cdot
        (\sigma,e) \cdot \pi(u') \mid
        \underline{\pi(v)}) &&
        [\text{contrapositive of
        Lemma~\ref{lem:rightCongruenceHelper},
    point~\ref{rch:leftToAbstract}}]\nonumber
    \end{align}

    Case 2:
    \begin{align*}
        f(u \mid \underline{(\sigma,d) \cdot u' \cdot
        v[d/d']}) &\ne f(u \mid \underline{(\sigma,d) \cdot u' \cdot
        v})\\
        \pi'(f(u \mid \underline{(\sigma,d) \cdot u' \cdot
        v[d/d']})) &\ne \pi'(f(u \mid \underline{(\sigma,d) \cdot u' \cdot
        v})) && [\text{apply } \pi' \text{ on both sides}]\\
        f(u \mid \underline{(\sigma,d) \cdot u' \cdot
        v[d/d'']}) &\ne f(u \mid \underline{(\sigma,d) \cdot u' \cdot
        v})&&
        [\text{Lemma~\ref{lem:factoredOutputInvariant}, }d',d'' \notin
        \data(u\cdot u' \cdot (\sigma,d) \cdot v,*)]\\
        f(u \mid \underline{\pi((\sigma,d) \cdot u' \cdot
        v[d/d''])}) &\ne f(u \mid \underline{\pi((\sigma,d) \cdot u' \cdot
        v)}) && [\text{Lemma~\ref{lem:nonPreIflPermutable}}]\\
        f(u \mid \underline{(\sigma,e) \cdot \pi(u') \cdot
        \pi(v)[e/d'']}) &\ne f(u \mid \underline{(\sigma,e) \cdot
            \pi(u') \cdot
        \pi(v)})\\
        f(u \cdot (\sigma,e) \cdot \pi(u') \mid
        \underline{\pi(v)[e/d'']}) &\ne f(u \cdot (\sigma,e) \cdot
            \pi(u') \mid
            \underline{\pi(v)})&&[\text{contrapositive of
            Lemma~\ref{lem:rightCongruenceHelper},
            point~\ref{rch:middleToRight}}]
    \end{align*}

    Since $d$ doesn't occur in $u'$, $e$ doesn't occur in
    $\pi(u')$. The last inequalities in each of the above cases certify that
    $e$ is $f$-vulnerable in $u \cdot (\sigma,e)$.

    Next we will assume that $\delta$ is $f$-vulnerable in
    $u \cdot (\sigma,d)$ and prove that $\delta$ is $f$-vulnerable
    in $u \cdot (\sigma,e)$. The proof of the other
    direction is similar.
    Since $\delta$ is $f$-vulnerable
    in $u \cdot (\sigma,d)$, we infer from
    Definition~\ref{def:fMemorable} that there exist data words
    $u',v$ and a data value $\delta'$ such that $\delta$ doesn't occur
    in $u'$, $\delta'$ is a safe replacement for $\delta$ in
            $u\cdot (\sigma,d) \cdot
            u'\cdot v$ and $f(u\cdot (\sigma,d)\cdot u' \mid
            \underline{v[\delta/\delta']}) \ne f(u\cdot (\sigma,d)\cdot u' \mid \underline{v})$.
    Applying the contrapositive of Lemma~\ref{lem:facOpEqualityFromParts} to the
    above inequality, we infer that at least one of the
    following inequalities are true.
    \begin{align*}
        f(\underline{u} \mid (\sigma,d) \cdot u' \mid
        \underline{v[\delta/\delta']}) &\ne f(\underline{u} \mid (\sigma,d) \cdot u' \mid
        \underline{v})\\
        f(u \mid \underline{(\sigma,d) \cdot u' \cdot v[\delta/\delta']}) &\ne f(u
        \mid \underline{(\sigma,d) \cdot u' \cdot v})
    \end{align*}
    Each of the above inequalities is taken up in one of the following
    cases. Let $\pi$ be the permutation that interchanges $d$ and $e$
    and doesn't change any other value. Let $\delta''$ be a data value such
    that $\delta'' \notin \data(u\cdot u'\cdot v,*) \cup
    \set{d,e,\delta'}$. Let $\pi'$  be the
    permutation that interchanges $\delta'$ and $\delta''$ and doesn't change
    any other value.

    Case 1:
    \begin{align}
        f(\underline{u} \mid (\sigma,d) \cdot u' \mid
        \underline{v[\delta/\delta']}) &\ne f(\underline{u} \mid (\sigma,d) \cdot u' \mid
        \underline{v})\nonumber\\
        \pi'(f(\underline{u} \mid (\sigma,d) \cdot u' \mid
        \underline{v[\delta/\delta']})) &\ne \pi'(f(\underline{u} \mid (\sigma,d) \cdot u' \mid
        \underline{v})) && [{\text{apply } \pi' \text{ to both sides}}]\nonumber \\
        f(\underline{u} \mid (\sigma,d) \cdot u' \mid
        \underline{v[\delta/\delta'']}) &\ne f(\underline{u} \mid (\sigma,d) \cdot u' \mid
        \underline{v})&&
        [\text{Lemma~\ref{lem:factoredOutputInvariant},
        }\delta',\delta'' \notin
        \data(u\cdot u' \cdot (\sigma,d) \cdot v,*)]\label{eq:allNewValuesPrefIfl7}\\
        f(\underline{u} \mid (\sigma,d) \cdot u' \cdot v[\delta/\delta'']) &=
        f(\underline{\pi(u)} \mid (\sigma,d) \cdot u' \cdot
        v[\delta/\delta''])
        && [\text{Lemma~\ref{lem:nonSufIflPermutable}}] \nonumber\\
        f(\underline{u} \mid (\sigma,d) \cdot u' \mid
        \underline{v[\delta/\delta'']}) &=
        f(\underline{\pi(u)} \mid (\sigma,d) \cdot u' \mid
        \underline{v[\delta/\delta'']})
        && [\text{point~\ref{rch:rightToAbstract} of Lemma~\ref{lem:rightCongruenceHelper}}]
        \label{eq:allNewValuesPrefIfl8}\\
        f(\underline{u} \mid (\sigma,d) \cdot u' \cdot v) &=
        f(\underline{\pi(u)} \mid (\sigma,d) \cdot u' \cdot v)
        && [\text{Lemma~\ref{lem:nonSufIflPermutable}}]
        \nonumber\\
        f(\underline{u} \mid (\sigma,d) \cdot u' \mid \underline{v}) &=
        f(\underline{\pi(u)} \mid (\sigma,d) \cdot u' \mid
        \underline{v})
        && [\text{point~\ref{rch:rightToAbstract} of Lemma~\ref{lem:rightCongruenceHelper}}]
        \label{eq:allNewValuesPrefIfl9}\\
        f(\underline{\pi(u)} \mid (\sigma,d) \cdot u' \mid
        \underline{v[\delta/\delta'']}) &\ne f(\underline{\pi(u)} \mid (\sigma,d) \cdot u' \mid
        \underline{v}) &&
        [\text{\eqref{eq:allNewValuesPrefIfl7},
        \eqref{eq:allNewValuesPrefIfl8},
    \eqref{eq:allNewValuesPrefIfl9}}] \nonumber\\
        \pi(f(\underline{\pi(u)} \mid (\sigma,d) \cdot u' \mid
        \underline{v[\delta/\delta'']})) &\ne \pi(f(\underline{\pi(u)} \mid (\sigma,d) \cdot u' \mid
        \underline{v})) && [\text{apply } \pi \text{ on both sides}]
        \nonumber\\
        f(\underline{u} \mid (\sigma,e) \cdot \pi(u') \mid
        \underline{\pi(v)[\delta/\delta'']}) &\ne f(\underline{u} \mid
        (\sigma,e) \cdot \pi(u') \mid
        \underline{\pi(v)}) &&
        [\text{Lemma~\ref{lem:factoredOutputInvariant},
        }\pi(\pi(u))=u, \set{d,e} \cap \set{\delta,\delta''}=\emptyset]\nonumber\\
        f(u \cdot (\sigma,e) \cdot \pi(u') \mid
        \underline{\pi(v)[\delta/\delta'']}) &\ne f(u \cdot
        (\sigma,e) \cdot \pi(u') \mid
        \underline{\pi(v)}) &&
        [\text{contrapositive of
        Lemma~\ref{lem:rightCongruenceHelper},
    point~\ref{rch:leftToAbstract}}]\nonumber
    \end{align}

    Case 2:
    \begin{align*}
        f(u \mid \underline{(\sigma,d) \cdot u' \cdot
        v[\delta/\delta']}) &\ne f(u \mid \underline{(\sigma,d) \cdot u' \cdot
        v})\\
        \pi'(f(u \mid \underline{(\sigma,d) \cdot u' \cdot
        v[\delta/\delta']})) &\ne \pi'(f(u \mid \underline{(\sigma,d) \cdot u' \cdot
        v})) && [\text{apply } \pi' \text{ on both sides}]\\
        f(u \mid \underline{(\sigma,d) \cdot u' \cdot
        v[\delta/\delta'']}) &\ne f(u \mid \underline{(\sigma,d) \cdot u' \cdot
        v})&&
        [\text{Lemma~\ref{lem:factoredOutputInvariant},
        }\delta',\delta'' \notin
        \data(u\cdot u' \cdot (\sigma,d) \cdot v,*)]\\
        f(u \mid \underline{\pi((\sigma,d) \cdot u' \cdot
        v[\delta/\delta''])}) &\ne f(u \mid \underline{\pi((\sigma,d) \cdot u' \cdot
        v)}) && [\text{Lemma~\ref{lem:nonPreIflPermutable}}]\\
        f(u \mid \underline{(\sigma,e) \cdot \pi(u') \cdot
        \pi(v)[\delta/\delta'']}) &\ne f(u \mid \underline{(\sigma,e) \cdot
            \pi(u') \cdot
        \pi(v)})\\
        f(u \cdot (\sigma,e) \cdot \pi(u') \mid
        \underline{\pi(v)[\delta/\delta'']}) &\ne f(u \cdot (\sigma,e) \cdot
            \pi(u') \mid
            \underline{\pi(v)})&&[\text{contrapositive of
            Lemma~\ref{lem:rightCongruenceHelper},
            point~\ref{rch:middleToRight}}]
    \end{align*}
    Since $\delta$ doesn't occur in $u'$, $\delta$ doesn't occur in
    $\pi(u')$. The last inequalities in each of the above cases certify that
    $\delta$ is $f$-vulnerable in $u \cdot (\sigma,e)$.
\end{proof}

\begin{proof}[Proof of Lemma~\ref{lem:rightCongruences}]
 Since $u_{1} \fequiv u_{2}$, there exists a
    permutation $\pi$ satisfying the conditions of
    Definition~\ref{def:fEquivalence}.  Let $z=|u_{1}| - |u_{2}|$.

    \textbf{Proof of 1.} Suppose $d_{1}^{i}$ is $f$-memorable
    in $u_{1} \cdot (\sigma,d_{1}^{j})$. There exist a data word $v$
    and a safe replacement $d'$ for $d_{1}^{i}$ in $u_{1}
    \cdot (\sigma,d_{1}^{j})$ such
    that $f(\underline{ (u_{1} \cdot
    (\sigma,d_{1}^{j}))[d_{1}^{i}/d']} \mid v) \ne f(\underline{
    u_{1} \cdot (\sigma,d_{1}^{j})} \mid v)$. Let $d''$ be a
    data value that is a safe replacement for $d_{1}^{i}$ in
    $u_{1}\cdot (\sigma,d_{1}^{j}) \cdot v \cdot
    \pi(u_{2})$. Let $\pi_{1}$ be the permutation that
    interchanges $d'$ and $d''$ and doesn't change
    any other value.  Let $\pi_{2}$ be the permutation that
    interchanges $d_{1}^{i}$ and $d''$ and doesn't change
    any other value.
    \begin{align}
        f(\underline{(u_{1} \cdot
        (\sigma,d_{1}^{j}))[d_{1}^{i}/d']} \mid v) &\ne
        f(\underline{u_{1} \cdot (\sigma,d_{1}^{j})} \mid v) \nonumber
        &&[\text{Definition~\ref{def:fMemorable}}]\\
        \pi_{1}(f(\underline{(u_{1} \cdot
        (\sigma,d_{1}^{j}))[d_{1}^{i}/d']} \mid v)) &\ne
        \pi_{1}(f(\underline{u_{1} \cdot (\sigma,d_{1}^{j})} \mid v)) \nonumber
        &&[\text{apply }\pi_{1}\text{ on both sides}]\\
        f(\underline{\pi_{1}((u_{1} \cdot
        (\sigma,d_{1}^{j}))[d_{1}^{i}/d'])} \mid \pi_{1}(v)) &\ne
        f(\underline{\pi_{1}(u_{1} \cdot (\sigma,d_{1}^{j}))} \mid
        \pi_{1}(v)) \nonumber
        &&[\text{Lemma~\ref{lem:factoredOutputInvariant}}]\\
        f(\underline{(u_{1} \cdot
        (\sigma,d_{1}^{j}))[d_{1}^{i}/d'']} \mid \pi_{1}(v)) &\ne
        f(\underline{u_{1} \cdot (\sigma,d_{1}^{j})} \mid
        \pi_{1}(v)) \label{eq:rightCongruences4}
        &&[\set{d',d''} \notin \data(u_{1} \cdot
        (\sigma,d_{1}^{j}), *)]\\
        f(\underline{u_{1}} \mid (\sigma,d_{1}^{j}) \cdot
        \pi_{1}(v)) &=
        f_{z}(\underline{\pi(u_{2})} \mid (\sigma,d_{1}^{j}) \cdot
        \pi_{1}(v))\nonumber && [\text{Definition~\ref{def:fEquivalence}}]\\
        f(\underline{u_{1} \cdot(\sigma,d_{1}^{j})} \mid
        \pi_{1}(v)) &=
        f_{z}(\underline{\pi(u_{2}) \cdot(\sigma,d_{1}^{j})} \mid 
        \pi_{1}(v))\label{eq:rightCongruences5} &&
        [\text{Lemma~\ref{lem:rightCongruenceHelper},
        point~\ref{rch:middleToLeft}}]\\
        f(\underline{u_{1}} \mid (\sigma,d_{1}^{j}) \cdot
        \pi_{2}^{-1}\pcomp \pi_{1}(v)) &= f_{z}(\underline{\pi(u_{2})} \mid (\sigma,d_{1}^{j}) \cdot
        \pi_{2}^{-1}\pcomp \pi_{1}(v)) \nonumber &&
        [\text{Definition~\ref{def:fEquivalence}}]\\
        \pi_{2}(f(\underline{u_{1}} \mid (\sigma,d_{1}^{j}) \cdot
        \pi_{2}^{-1}\pcomp \pi_{1}(v))) &= \pi_{2}(f_{z}(\underline{\pi(u_{2})} \mid (\sigma,d_{1}^{j}) \cdot
        \pi_{2}^{-1}\pcomp \pi_{1}(v))) \nonumber && [\text{apply
        }\pi_{2}\text{ on both sides}]\\
        f(\underline{\pi_{2}(u_{1})} \mid \pi_{2}((\sigma,d_{1}^{j}) \cdot
        \pi_{2}^{-1}\pcomp \pi_{1}(v))) &=
        f_{z}(\underline{\pi_{2}(\pi(u_{2}))} \mid \pi_{2}((\sigma,d_{1}^{j}) \cdot
        \pi_{2}^{-1}\pcomp \pi_{1}(v))) \nonumber &&
        [\text{Lemma~\ref{lem:factoredOutputInvariant}}]\\
        f(\underline{u_{1}[d_{1}^{i}/d'']} \mid
        (\sigma,d_{1}^{j})[d_{1}^{i}/d''] \cdot \pi_{1}(v)) &=
        f_{z}(\underline{\pi(u_{2})[d_{1}^{i}/d'']} \mid
        (\sigma,d_{1}^{j})[d_{1}^{i}/d''] \cdot \pi_{1}(v)) \nonumber &&
        [d'' \notin \data(u_{1}\cdot (\sigma,d_{1}^{j})\cdot
        \pi(u_{2}),*)]\\
        f(\underline{(u_{1}\cdot (\sigma,d_{1}^{j}))[d_{1}^{i}/d''] } \mid
        \pi_{1}(v)) &=
        f_{z}(\underline{(\pi(u_{2}) \cdot (\sigma,d_{1}^{j}))[d_{1}^{i}/d'']} \mid
        \pi_{1}(v)) \label{eq:rightCongruences6} &&
        [\text{Lemma~\ref{lem:rightCongruenceHelper},
        point~\ref{rch:middleToLeft}}]\\
        f_{z}(\underline{(\pi(u_{2}) \cdot (\sigma,d_{1}^{j}))[d_{1}^{i}/d'']} \mid
        \pi_{1}(v)) &\ne f_{z}(\underline{\pi(u_{2}) \cdot(\sigma,d_{1}^{j})} \mid 
        \pi_{1}(v)) \nonumber && [\text{\eqref{eq:rightCongruences4},
        \eqref{eq:rightCongruences5}, \eqref{eq:rightCongruences6}}]\\
        f(\underline{(\pi(u_{2}) \cdot (\sigma,d_{1}^{j}))[d_{1}^{i}/d'']} \mid
        \pi_{1}(v)) &\ne f(\underline{\pi(u_{2}) \cdot(\sigma,d_{1}^{j})} \mid 
        \pi_{1}(v)) \nonumber
    \end{align}
    Since $d''$ is a safe replacement for $d_{1}^{i}$ in
    $\pi(u_{2}) \cdot
    (\sigma,d_{1}^{j})$, the last inequality above certifies that
    $d_{1}^{i}$ is $f$-memorable in $\pi(u_{2}) \cdot
    (\sigma,d_{1}^{j})$. Since $\pi(u_{2}) \cdot
    (\sigma,d_{1}^{j})=\pi(u_{2} \cdot (\sigma,\pi^{-1}(d_{1}^{j})))$,
    we infer that $d_{1}^{i}$ is $f$-memorable in $\pi(u_{2}
    \cdot (\sigma,\pi^{-1}(d_{1}^{j})))$. From
    Lemma~\ref{lem:iflInvUnderPerm}, we infer that
    $\pi^{-1}(d_{1}^{i})$ is $f$-memorable in $u_{2} \cdot
    (\sigma,\pi^{-1}(d_{1}^{j}))$.

    Case 1: $(d_{1}^{j},d_{2}^{j}) \in \set{(d_{1}^{k},d_{2}^{k}) \mid 1
    \le k \le m}$. In this case, $\pi^{-1}(d_{1}^{j})=d_{2}^{j}$. So
    $\pi^{-1}(d_{1}^{i})$ is $f$-memorable in $u_{2} \cdot
    (\sigma,d_{2}^{j})$. Since $d_{1}^{i}$ is $f$-memorable
    in $u_{1} \cdot (\sigma,d_{1}^{j})$, we infer from
    Lemma~\ref{lem:iflValsMonotonic} that $d_{1}^{i}$ is $f$-memorable
    in $u_{1}$ or $d_{1}^{i}=d_{1}^{j}$. Either way,
    $d_{1}^{i}\in \set{d_{1}^{k} \mid 1 \le k \le m}$, so
    $(d_{1}^{i},d_{2}^{i}) \in \set{(d_{1}^{k},d_{2}^{k}) \mid 1
    \le k \le m}$. Hence $\pi^{-1}(d_{1}^{i})=d_{2}^{i}$, so
    $d_{2}^{i}$ is $f$-memorable in $u_{2} \cdot
    (\sigma,d_{2}^{j})$.

    Case 2: $(d_{1}^{j},d_{2}^{j})=(d_{1}^{0},d_{2}^{0})$. Since
    $d_{1}^{j}=d_{1}^{0}$ is not $f$-influencing in $u_{1}$,
    $\pi^{-1}(d_{1}^{j})$ is not $f$-influencing in $u_{2}$. From the
    hypothesis of this lemma, $d_{2}^{j}=d_{2}^{0}$ is not $f$-influencing
    in $u_{2}$. If $(d_{1}^{i},d_{2}^{i})\in
    \set{(d_{1}^{k},d_{2}^{k}) \mid 1 \le k \le m}$, then
    $\pi^{-1}(d_{1}^{i})=d_{2}^{i}$. So $d_{2}^{i}$ is $f$-memorable
    in $u_{2} \cdot (\sigma,\pi^{-1}(d_{1}^{j}))$.  From
    Lemma~\ref{lem:allNewValuesSufIfl}, we conclude that $d_{2}^{i}$
    is $f$-memorable in $u_{2} \cdot (\sigma,d_{2}^{j})$. The
    other possibility is that $(d_{1}^{i},d_{2}^{i})=(d_{1}^{0},d_{2}^{0}) =
    (d_{1}^{j},d_{2}^{j})$.  Since $d_{1}^{i}=d_{1}^{0}$ is not
    $f$-influencing in $u_{1}$, $\pi^{-1}(d_{1}^{i})$ is not
    $f$-influencing in $u_{2}$. Since
    $\pi^{-1}(d_{1}^{i})=\pi^{-1}(d_{1}^{j})$ is $f$-memorable
    in $u_{2} \cdot (\sigma,\pi^{-1}(d_{1}^{j}))$, from
    Lemma~\ref{lem:allNewValuesSufIfl}, we conclude that
    $d_{2}^{i}=d_{2}^{j}$ is $f$-memorable in $u_{2} \cdot
    (\sigma,d_{2}^{j})$. If $d_{2}^{i}$ is $f$-memorable in
    $u_{2} \cdot (\sigma,d_{2}^{j})$, we can prove that $d_{1}^{i}$ is
    $f$-memorable in $u_{1} \cdot (\sigma,d_{1}^{j})$ with a
    similar proof.

    Suppose $d_{1}^{i}$ is $f$-vulnerable in $u_{1} \cdot
    (\sigma,d_{1}^{j})$. We infer from Definition~\ref{def:fMemorable}
    that there exist data words $u',v$ and a data value $d'$
    such that $d_{1}^{i}$ doesn't occur in $u'$, $d'$ is a safe replacement for $d_{1}^{i}$ in
            $u_{1}\cdot (\sigma,d_{1}^{j}) \cdot
            u'\cdot v$ and $f(u_{1}\cdot (\sigma,d_{1}^{j})\cdot u' \mid
            \underline{v[d_{1}^{i}/d']}) \ne f(u_{1}\cdot
            (\sigma,d_{1}^{j})\cdot u' \mid \underline{v})$.
    Let $d''$ be a data value that is a safe replacement for
    $d_{1}^{i}$ in $u_{1}\cdot (\sigma,d_{1}^{j}) \cdot u' \cdot v \cdot
    \pi(u_{2})$. Let $\pi_{1}$ be the permutation that interchanges
    $d'$ and $d''$ and doesn't change any other value.

    \begin{align*}
        f(u_{1}\cdot (\sigma,d_{1}^{j})\cdot u' \mid
            \underline{v[d_{1}^{i}/d']}) &\ne f(u_{1}\cdot
            (\sigma,d_{1}^{j})\cdot u' \mid \underline{v})\\
            \pi_{1}(f(u_{1}\cdot (\sigma,d_{1}^{j})\cdot u' \mid
            \underline{v[d_{1}^{i}/d']})) &\ne \pi_{1}(f(u_{1}\cdot
            (\sigma,d_{1}^{j})\cdot u' \mid \underline{v}))
            &&[\text{apply }\pi_{1}\text{ on both sides}]\\
        f(u_{1}\cdot (\sigma,d_{1}^{j})\cdot u' \mid
            \underline{v[d_{1}^{i}/d'']}) &\ne f(u_{1}\cdot
            (\sigma,d_{1}^{j})\cdot u' \mid \underline{v}) &&
            [\text{Lemma~\ref{lem:factoredOutputInvariant},
            }d',d'' \notin \data(u_{1}\cdot
        (\sigma,d_{1}^{j})\cdot u'\cdot v,*)]\\
        f(\pi(u_{2})\cdot (\sigma,d_{1}^{j})\cdot u' \mid
            \underline{v[d_{1}^{i}/d'']}) &\ne f(\pi(u_{2})\cdot
            (\sigma,d_{1}^{j})\cdot u' \mid \underline{v}) &&
            [\text{last condition on }\pi\text{ in
            Definition~\ref{def:fEquivalence}}]
    \end{align*}
    The last inequality above implies that $d_{1}^{i}$
    is $f$-vulnerable in $\pi(u_{2}) \cdot (\sigma,d_{1}^{j})$.
    Since, $\pi(u_{2}) \cdot (\sigma,d_{1}^{j}) = \pi(u_{2} \cdot
    (\sigma,\pi^{-1}(d_{1}^{j})))$, $d_{1}^{i}$ is $f$-vulnerable
    in $\pi(u_{2} \cdot (\sigma,\pi^{-1}(d_{1}^{j})))$. From
    Lemma~\ref{lem:iflInvUnderPerm}, we infer that
    $\pi^{-1}(d_{1}^{i})$ is $f$-vulnerable in $u_{2} \cdot
    (\sigma,\pi^{-1}(d_{1}^{j}))$.
    
    Case 1: $(d_{1}^{j},d_{2}^{j}) \in
    \set{(d_{1}^{k},d_{2}^{k}) \mid 1 \le k \le m}$. In this case,
    $\pi^{-1}(d_{1}^{j})=d_{2}^{j}$ (since $\pi$ maps
    $\ifl_{f}(u_{2})$ to $\ifl_{f}(u_{1})$), so
    $\pi^{-1}(d_{1}^{i})$ is
    $f$-vulnerable in $u_{2} \cdot (\sigma,d_{2}^{j})$. Since
    $d_{1}^{i}$ is $f$-vulnerable in $u_{1} \cdot
    (\sigma,d_{1}^{j})$, we infer from Lemma~\ref{lem:iflValsMonotonic}
    that $d_{1}^{i}$ is $f$-vulnerable in $u_{1}$ or
    $d_{1}^{i}=d_{1}^{j}$. Either way, $d_{1}^{i} \in
    \set{d_{1}^{k} \mid 1 \le k \le m}$, so
    $(d_{1}^{i},d_{2}^{i}) \in \set{(d_{1}^{k},d_{2}^{k}) \mid 1 \le
    k \le m}$. Hence, $\pi^{-1}(d_{1}^{i})=d_{2}^{i}$, so
    $d_{2}^{i}$ is $f$-vulnerable in $u_{2} \cdot (\sigma,d_{2}^{j})$.

    Case 2: $(d_{1}^{j},d_{2}^{j})=(d_{1}^{0},d_{2}^{0})$. In this
    case, $d_{2}^{j}=d_{2}^{0}$ is not $f$-influencing in $u_{2}$, and
    $\pi^{-1}(d_{1}^{j})=\pi^{-1}(d_{1}^{0})$ is not $f$-influencing in
    $u_{2}$ (since $d_{1}^{0}$ is not $f$-influencing in $u_{1}$). If
    $(d_{1}^{i},d_{2}^{i})\in \set{(d_{1}^{k},d_{2}^{k}) \mid 1 \le
    k \le m}$, then $\pi^{-1}(d_{1}^{i})=d_{2}^{i}$. So
    $d_{2}^{i}$ is $f$-vulnerable in $u_{2} \cdot
    (\sigma,\pi^{-1}(d_{1}^{j}))$. From
    Lemma~\ref{lem:allNewValuesPrefIfl}, we infer that $d_{2}^{i}$ is
    $f$-vulnerable in $u_{2} \cdot (\sigma,d_{2}^{j})$. The
    other possibility is that
    $(d_{1}^{i},d_{2}^{i})=(d_{1}^{0},d_{2}^{0}) = (d_{1}^{j},d_{2}^{j})$.
    Since $d_{1}^{i}=d_{1}^{0}$ is not $f$-influencing in $u_{1}$,
    $\pi^{-1}(d_{1}^{i})$ is not $f$-influencing in $u_{2}$. Since
    $\pi^{-1}(d_{1}^{i})=\pi^{-1}(d_{1}^{j})$ is $f$-vulnerable
    in $u_{2} \cdot (\sigma,\pi^{-1}(d_{1}^{j}))$, from
    Lemma~\ref{lem:allNewValuesPrefIfl}, we conclude that
    $d_{2}^{i}=d_{2}^{j}$ is $f$-vulnerable
    in $u_{2} \cdot (\sigma,d_{2}^{j})$.
    If $d_{2}^{i}$ if
    $f$-vulnerable in $u_{2} \cdot (\sigma,d_{2}^{j})$, we can
    prove that $d_{1}^{i}$ is $f$-vulnerable in $u_{1} \cdot
    (\sigma,d_{1}^{j})$ with a similar proof.

    \textbf{Proof of 2.} Let $\pi'$ be the permutation that
    interchanges $d_{1}^{0}$ and $\pi(d_{2}^{0})$ and doesn't change any other
    value. To prove that $u_{1} \cdot (\sigma,d_{1}^{j}) \fequiv
    u_{2} \cdot (\sigma,d_{2}^{j})$, we will prove that the permutation
    $\pi' \pcomp \pi$ satisfies all the conditions of
    Definition~\ref{def:fEquivalence}. Note that
    $\pi'\pcomp \pi(d_{2}^{j})=d_{1}^{j}$. From
    Lemma~\ref{lem:iflValsMonotonic}, we infer that $f$-influencing
    values in $u_{1} \cdot (\sigma,d_{1}^{j})$ are among
    $\set{d_{1}^{k} \mid 1 \le k \le m} \cup \set{d_{1}^{j}}$ and that
    $f$-influencing values in $u_{2} \cdot (\sigma,d_{2}^{j})$ are
    among $\set{d_{2}^{k} \mid 1 \le k \le m} \cup \set{d_{2}^{j}}$. We
    infer from point 1 of this lemma that $d_{1}^{j}$ is $f$-memorable
    (resp.~$f$-vulnerable) in $u_{1} \cdot
    (\sigma,d_{1}^{j})$ iff $d_{2}^{j}$ is $f$-memorable
    (resp.~$f$-vulnerable) in $u_{2} \cdot (\sigma,d_{2}^{j})$.
    We also infer from point 1 of this lemma that for
    $(d_{1}^{i},d_{2}^{i}) \in \set{(d_{1}^{k},d_{2}^{k}) \mid 1 \le
    k \le m}$, $d_{1}^{i}$ is $f$-memorable
    (resp.~$f$-vulnerable) in $u_{1} \cdot (\sigma,d_{1}^{j})$
    iff $d_{2}^{i}$ is $f$-memorable (resp.~$f$-vulnerable) in $u_{2} \cdot (\sigma,d_{2}^{j})$. Since,
    $\pi'\pcomp\pi(d_{2}^{i})=d_{1}^{i}$ and
    $\pi'\pcomp\pi(d_{2}^{j})=d_{1}^{j}$, we infer that
    $\aifl_{f}(\pi'\pcomp\pi(u_{2}\cdot
    (\sigma,d_{2}^{j})))=\aifl_{f}(u_{1} \cdot (\sigma,d_{1}^{j}))$.

    Let $v$ be an arbitrary data word. Since $d_{2}^{0}$ is not
    $f$-influencing in $u_{2}$, $\pi(d_{2}^{0})$ is not $f$-influencing in
    $u_{1}$.
    \begin{align*}
        f_{z}(\underline{\pi(u_{2})} \mid (\sigma,d_{1}^{j}) \cdot v)
        &= f(\underline{u_{1}} \mid (\sigma,d_{1}^{j}) \cdot v) &&
        [\text{first condition on }\pi\text{ in
        Definition~\ref{def:fEquivalence}}]\\
        f_{z}(\underline{\pi' \pcomp \pi(u_{2})} \mid (\sigma,d_{1}^{j}) \cdot v)
        &= f(\underline{u_{1}} \mid (\sigma,d_{1}^{j}) \cdot v) &&
        [\text{Lemma~\ref{lem:nonSufIflPermutable}}]\\
        f_{z}(\underline{\pi' \pcomp \pi(u_{2}) \cdot (\sigma,d_{1}^{j})} \mid v)
        &= f(\underline{u_{1} \cdot (\sigma,d_{1}^{j})} \mid v) &&
        [\text{Lemma~\ref{lem:rightCongruenceHelper},
        point~\ref{rch:middleToLeft}}]\\
        f_{z}(\underline{\pi' \pcomp \pi(u_{2} \cdot (\sigma,d_{2}^{j}))} \mid v)
        &= f(\underline{u_{1} \cdot (\sigma,d_{1}^{j})} \mid v)
    \end{align*}
    Since the last inequality above holds for any data word $v$, it
    proves the first condition of Definition~\ref{def:fEquivalence}.

    For the last condition of Definition~\ref{def:fEquivalence},
    suppose $u, v_{1}, v_{2}$ are arbitrary data values and
    $f(u_{1} \cdot (\sigma,d_{1}^{j}) \cdot u \mid
    \underline{v_{1}}) = f(u_{1} \cdot (\sigma,d_{1}^{j}) \cdot u \mid
    \underline{v_{2}})$. Since, $u_{1}\fequiv u_{2}$ and $\pi$
    satisfies all the conditions of Definition~\ref{def:fEquivalence},
    we infer that $f(\pi(u_{2}) \cdot (\sigma,d_{1}^{j}) \cdot u \mid
    \underline{v_{1}}) = f(\pi(u_{2}) \cdot (\sigma,d_{1}^{j}) \cdot u \mid
    \underline{v_{2}})$.
    \begin{align}
        f(\pi(u_{2}) \cdot (\sigma,d_{1}^{j}) \cdot u \mid
    \underline{v_{1}}) &= f(\pi(u_{2}) \cdot (\sigma,d_{1}^{j}) \cdot u \mid
    \underline{v_{2}})\nonumber\\
        f(\pi(u_{2}) \mid \underline{(\sigma,d_{1}^{j}) \cdot u \cdot
    v_{1}}) &= f(\pi(u_{2}) \mid \underline{(\sigma,d_{1}^{j}) \cdot u
    \cdot v_{2}})&&[\text{Lemma~\ref{lem:rightCongruenceHelper},
    point~\ref{rch:middleToRight}}]\nonumber\\
        f(\pi(u_{2}) \mid \underline{\pi'((\sigma,d_{1}^{j}) \cdot u \cdot
    v_{1})}) &= f(\pi(u_{2}) \mid \underline{\pi'((\sigma,d_{1}^{j}) \cdot u
    \cdot v_{2})})&&[\text{Lemma~\ref{lem:nonPreIflPermutable},
    }\pi(d_{2}^{0}),d_{1}^{0} \notin \data(\aifl_{f}(\pi(u_{2}),*))]\nonumber\\
        \pi'(f(\pi(u_{2}) \mid \underline{\pi'((\sigma,d_{1}^{j}) \cdot u \cdot
    v_{1})})) &= \pi'(f(\pi(u_{2}) \mid \underline{\pi'((\sigma,d_{1}^{j}) \cdot u
    \cdot v_{2})}))&&[\text{apply }\pi' \text{ on both sides}]\nonumber\\
        f(\pi'\pcomp \pi(u_{2}) \mid \underline{(\sigma,d_{1}^{j}) \cdot u \cdot
    v_{1}}) &= f(\pi'\pcomp\pi(u_{2}) \mid \underline{(\sigma,d_{1}^{j}) \cdot u
    \cdot
    v_{2}})&&[\text{Lemma~\ref{lem:factoredOutputInvariant}}]\label{eq:rc1}\\
        f(\pi(u_{2}) \cdot (\sigma,d_{1}^{j}) \cdot u \mid
    \underline{v_{1}}) &= f(\pi(u_{2}) \cdot (\sigma,d_{1}^{j}) \cdot u \mid
    \underline{v_{2}})\nonumber\\
    f(\underline{\pi(u_{2})} \mid (\sigma,d_{1}^{j}) \cdot u \mid
    \underline{v_{1}}) &= f(\underline{\pi(u_{2})} \mid(\sigma,d_{1}^{j}) \cdot u \mid
    \underline{v_{2}})&&[\text{Lemma~\ref{lem:rightCongruenceHelper},
    point~\ref{rch:leftToAbstract}}]\label{eq:rc2}\\
    f(\underline{\pi'\pcomp\pi(u_{2})} \mid (\sigma,d_{1}^{j}) \cdot u
    \cdot v_{1}) &= f(\underline{\pi(u_{2})} \mid (\sigma,d_{1}^{j}) \cdot u
    \cdot v_{1})&&[\text{Lemma~\ref{lem:nonSufIflPermutable},
    }\pi(d_{2}^{0}),d_{1}^{0} \notin \data(\aifl_{f}(\pi(u_{2}),*))]\nonumber\\
    f(\underline{\pi'\pcomp\pi(u_{2})} \mid (\sigma,d_{1}^{j}) \cdot u
    \mid \underline{v_{1}}) &= f(\underline{\pi(u_{2})} \mid (\sigma,d_{1}^{j}) \cdot u
    \mid
    \underline{v_{1}})&&[\text{Lemma~\ref{lem:rightCongruenceHelper},
    point~\ref{rch:rightToAbstract}}]\label{eq:rc3}\\
    f(\underline{\pi'\pcomp\pi(u_{2})} \mid (\sigma,d_{1}^{j}) \cdot u
    \cdot v_{2}) &= f(\underline{\pi(u_{2})} \mid (\sigma,d_{1}^{j}) \cdot u
    \cdot v_{2})&&[\text{Lemma~\ref{lem:nonSufIflPermutable},
    }\pi(d_{2}^{0}),d_{1}^{0} \notin \data(\aifl_{f}(\pi(u_{2}),*))]\nonumber\\
    f(\underline{\pi'\pcomp\pi(u_{2})} \mid (\sigma,d_{1}^{j}) \cdot u
    \mid \underline{v_{2}}) &= f(\underline{\pi(u_{2})} \mid (\sigma,d_{1}^{j}) \cdot u
    \mid
    \underline{v_{2}})&&[\text{Lemma~\ref{lem:rightCongruenceHelper},
    point~\ref{rch:rightToAbstract}}]\label{eq:rc4}\\
    f(\underline{\pi'\pcomp\pi(u_{2})} \mid (\sigma,d_{1}^{j}) \cdot u \mid
    \underline{v_{1}}) &= f(\underline{\pi'\pcomp\pi(u_{2})} \mid(\sigma,d_{1}^{j}) \cdot u \mid
    \underline{v_{2}})&&[\text{\eqref{eq:rc2},\eqref{eq:rc3},\eqref{eq:rc4}}]\label{eq:rc5}\\
    f(\pi'\pcomp \pi(u_{2}) \cdot (\sigma,d_{1}^{j}) \cdot u \mid
    \underline{v_{1}}) &= f(\pi'\pcomp\pi(u_{2}) \cdot (\sigma,d_{1}^{j}) \cdot u \mid
    \underline{v_{2}})\nonumber&&[\text{\eqref{eq:rc1},\eqref{eq:rc5},
    Lemma~\ref{lem:facOpEqualityFromParts}}]\\
    f(\pi'\pcomp \pi(u_{2}\cdot (\sigma,d_{2}^{j}))  \cdot u \mid
    \underline{v_{1}}) &= f(\pi'\pcomp\pi(u_{2} \cdot (\sigma,d_{2}^{j})) \cdot u \mid
    \underline{v_{2}})\nonumber
    \end{align}
    Hence, if $f(u_{1} \cdot (\sigma,d_{1}^{j}) \cdot u \mid
    \underline{v_{1}}) = f(u_{1} \cdot (\sigma,d_{1}^{j}) \cdot u \mid
    \underline{v_{2}})$, then $f(\pi'\pcomp\pi(u_{2}\cdot
    (\sigma,d_{2}^{j}))  \cdot u \mid \underline{v_{1}}) = f(\pi'\pcomp
    \pi(u_{2} \cdot (\sigma,d_{2}^{j})) \cdot u \mid
    \underline{v_{2}})$.

    Conversely, suppose $f(\pi'\pcomp \pi(u_{2}\cdot (\sigma,d_{2}^{j}))
    \cdot u \mid \underline{v_{1}}) = f(\pi'\pcomp\pi(u_{2} \cdot
    (\sigma,d_{2}^{j})) \cdot u \mid \underline{v_{2}})$. Then we have
    $f(\pi'\pcomp \pi(u_{2}) \cdot (\sigma,d_{1}^{j}) \cdot u \mid
    \underline{v_{1}}) = f(\pi'\pcomp\pi(u_{2}) \cdot
    (\sigma,d_{1}^{j}) \cdot u \mid \underline{v_{2}})$. Recall that
    $\pi(d_{2}^{0})$ and $d_{1}^{0}$ are not $f$-influencing in $\pi(u_{2})$.
    \begin{align}
        f(\pi'\pcomp\pi(u_{2}) \cdot (\sigma,d_{1}^{j}) \cdot u \mid
    \underline{v_{1}}) &= f(\pi'\pcomp\pi(u_{2}) \cdot (\sigma,d_{1}^{j}) \cdot u \mid
    \underline{v_{2}})\nonumber\\
        f(\pi'\pcomp\pi(u_{2}) \mid \underline{(\sigma,d_{1}^{j}) \cdot u \cdot
    v_{1}}) &= f(\pi'\pcomp\pi(u_{2}) \mid \underline{(\sigma,d_{1}^{j}) \cdot u
    \cdot v_{2}})&&[\text{Lemma~\ref{lem:rightCongruenceHelper},
    point~\ref{rch:middleToRight}}]\nonumber\\
        f(\pi'\pcomp\pi(u_{2}) \mid \underline{\pi'((\sigma,d_{1}^{j}) \cdot u \cdot
    v_{1})}) &= f(\pi'\pcomp\pi(u_{2}) \mid \underline{\pi'((\sigma,d_{1}^{j}) \cdot u
    \cdot v_{2})})&&[\text{Lemma~\ref{lem:nonPreIflPermutable},
    }\pi(d_{2}^{0}),d_{1}^{0} \notin \nonumber\\& &&\data(\aifl_{f}(\pi'\cdot\pi(u_{2}),*))]\nonumber\\
        \pi'(f(\pi'\pcomp\pi(u_{2}) \mid \underline{\pi'((\sigma,d_{1}^{j}) \cdot u \cdot
    v_{1})})) &= \pi'(f(\pi'\pcomp\pi(u_{2}) \mid \underline{\pi'((\sigma,d_{1}^{j}) \cdot u
    \cdot v_{2})}))&&[\text{apply }\pi' \text{ on both sides}]\nonumber\\
        f(\pi(u_{2}) \mid \underline{(\sigma,d_{1}^{j}) \cdot u \cdot
    v_{1}}) &= f(\pi(u_{2}) \mid \underline{(\sigma,d_{1}^{j}) \cdot u
    \cdot
    v_{2}})&&[\text{Lemma~\ref{lem:factoredOutputInvariant}}]\label{eq:rc6}\\
        f(\pi'\pcomp\pi(u_{2}) \cdot (\sigma,d_{1}^{j}) \cdot u \mid
    \underline{v_{1}}) &= f(\pi'\pcomp\pi(u_{2}) \cdot (\sigma,d_{1}^{j}) \cdot u \mid
    \underline{v_{2}})\nonumber\\
    f(\underline{\pi'\pcomp\pi(u_{2})} \mid (\sigma,d_{1}^{j}) \cdot u \mid
    \underline{v_{1}}) &= f(\underline{\pi'\pcomp\pi(u_{2})} \mid(\sigma,d_{1}^{j}) \cdot u \mid
    \underline{v_{2}})&&[\text{Lemma~\ref{lem:rightCongruenceHelper},
    point~\ref{rch:leftToAbstract}}]\label{eq:rc7}\\
    f(\underline{\pi'\pcomp\pi(u_{2})} \mid (\sigma,d_{1}^{j}) \cdot u
    \cdot v_{1}) &= f(\underline{\pi(u_{2})} \mid (\sigma,d_{1}^{j}) \cdot u
    \cdot v_{1})&&[\text{Lemma~\ref{lem:nonSufIflPermutable},
    }\pi(d_{2}^{0}),d_{1}^{0} \notin \nonumber\\ & && \data(\aifl_{f}(\pi(u_{2}),*))]\nonumber\\
    f(\underline{\pi'\pcomp\pi(u_{2})} \mid (\sigma,d_{1}^{j}) \cdot u
    \mid \underline{v_{1}}) &= f(\underline{\pi(u_{2})} \mid (\sigma,d_{1}^{j}) \cdot u
    \mid
    \underline{v_{1}})&&[\text{Lemma~\ref{lem:rightCongruenceHelper},
    point~\ref{rch:rightToAbstract}}]\label{eq:rc8}\\
    f(\underline{\pi'\pcomp\pi(u_{2})} \mid (\sigma,d_{1}^{j}) \cdot u
    \cdot v_{2}) &= f(\underline{\pi(u_{2})} \mid (\sigma,d_{1}^{j}) \cdot u
    \cdot v_{2})&&[\text{Lemma~\ref{lem:nonSufIflPermutable},
    }\pi(d_{2}^{0}),d_{1}^{0} \notin \nonumber\\ & && \data(\aifl_{f}(\pi(u_{2}),*))]\nonumber\\
    f(\underline{\pi'\pcomp\pi(u_{2})} \mid (\sigma,d_{1}^{j}) \cdot u
    \mid \underline{v_{2}}) &= f(\underline{\pi(u_{2})} \mid (\sigma,d_{1}^{j}) \cdot u
    \mid
    \underline{v_{2}})&&[\text{Lemma~\ref{lem:rightCongruenceHelper},
    point~\ref{rch:rightToAbstract}}]\label{eq:rc9}\\
    f(\underline{\pi(u_{2})} \mid (\sigma,d_{1}^{j}) \cdot u \mid
    \underline{v_{1}}) &= f(\underline{\pi(u_{2})} \mid(\sigma,d_{1}^{j}) \cdot u \mid
    \underline{v_{2}})&&[\text{\eqref{eq:rc7},\eqref{eq:rc8},\eqref{eq:rc9}}]\label{eq:rc10}\\
    f(\pi(u_{2}) \cdot (\sigma,d_{1}^{j}) \cdot u \mid
    \underline{v_{1}}) &= f(\pi(u_{2}) \cdot (\sigma,d_{1}^{j}) \cdot u \mid
    \underline{v_{2}})\nonumber&&[\text{\eqref{eq:rc6},\eqref{eq:rc10},
    Lemma~\ref{lem:facOpEqualityFromParts}}]
    \end{align}
    Since, $u_{1}\fequiv u_{2}$ and $\pi$ satisfies all the conditions
    of Definition~\ref{def:fEquivalence}, we infer from the last
    equality above that $f(u_{1} \cdot (\sigma,d_{1}^{j}) \cdot u \mid
    \underline{v_{1}}) = f(u_{1} \cdot (\sigma,d_{1}^{j}) \cdot u \mid
    \underline{v_{2}})$. Hence, if $f(\pi'\pcomp\pi(u_{2}\cdot
    (\sigma,d_{2}^{j})) \cdot u \mid \underline{v_{1}}) = f(\pi'\pcomp
    \pi(u_{2} \cdot (\sigma,d_{2}^{j})) \cdot u \mid
    \underline{v_{2}})$, then $f(u_{1} \cdot (\sigma,d_{1}^{j}) \cdot u \mid
    \underline{v_{1}}) = f(u_{1} \cdot (\sigma,d_{1}^{j}) \cdot u \mid
    \underline{v_{2}})$. Therefore, the permutation $\pi'\pcomp\pi$
    satisfies all the conditions of Definition~\ref{def:fEquivalence},
    so $u_{1} \cdot (\sigma,d_{1}^{j}) \fequiv u_{2} \cdot
    (\sigma,d_{2}^{j})$.
\end{proof}

\end{document}